\documentclass[11pt]{article}

\newcommand{\Dlabel}[1]{\label[definition]{#1}}
\newcommand{\Llabel}[1]{\label[lemma]{#1}}
\newcommand{\Clabel}[1]{\label[corollary]{#1}}

\newcommand{\GER}{\GE{r,\sigma}}
\newcommand{\svz}{$\sigma$-visible zigzag}
\newcommand{\cpath}{constraint-path}

\newcommand{\proofs}{the appendix}

\usepackage[
  bookmarks=true,
  bookmarksnumbered=true,
  bookmarksopen=true,
  pdfborder={0 0 0},
  breaklinks=true,
  colorlinks=true,
  linkcolor=black,
  citecolor=black,
  filecolor=black,
  urlcolor=black,
]{hyperref}

\usepackage{times}

\usepackage{mathtools}
\usepackage{amsthm}
\usepackage{amsmath}
\usepackage{amsfonts}
\usepackage{amssymb}
\usepackage{graphicx}
\usepackage{pdfpages}
\usepackage{stmaryrd}
\usepackage{enumitem}
\usepackage[capitalise]{cleveref}
\usepackage{nicefrac}
\usepackage{bm}
\usepackage{caption}
\usepackage{subcaption}
\usepackage{wasysym}

\usepackage[normalem]{ulem}

\usepackage[margin=1in]{geometry}

\newtheorem{protocol}{Protocol}

\newtheorem{definition}{Definition}

\newtheorem{corollary}{Corollary}
\newtheorem{lemma}{Lemma}

\newtheorem{theorem}{Theorem}

\newlist{parts}{enumerate}{1}
\crefname{partsi}{Part}{Parts}
\setlist[parts,1]{label=\arabic*.,ref=\arabic*}

\newcommand{\white}[1]{{\color{white}{#1}}}

\newcommand{\Int}{{\mathbb{Z}}}
\newcommand{\Nat}{{\mathbb{N}}}
\newcommand{\sfa}{\underline{a}}
\newcommand{\sfb}{\underline{b}}

\newcommand{\eqdef}{\triangleq}

\def\emptyset{\mbox{\O}}

\newcommand{\defemph}[1]{\textbf{\textit{#1}}}
\newcommand{\sat}{\models}

\newcommand{\Proc}{\mathsf{Procs}}

\newcommand{\bnode}[1]{(#1)}
\newcommand{\node}[1]{\langle#1\rangle}

\newcommand{\fip}{{\it fip}}

\newcommand{\Gcal}{\ensuremath{\mathcal{G}}}
\newcommand{\Gz}{\Gcal_0}
\newcommand{\Rrep}{\ensuremath{\mathcal{R}}}
\newcommand{\Net}{\mathsf{Net}}
\newcommand{\Chan}{{\mathsf{Chans}}}

\newcommand{\lb}[1]{\bm{L}_{\scriptscriptstyle{#1}}}
\newcommand{\ub}[1]{\bm{U}_{\scriptscriptstyle{#1}}}
\newcommand{\bcm}{\mbox{$\mathsf{bcm}$}}
\newcommand{\Ext}{\mathbb{E}}
\newcommand{\ffip}{{\sf{FFIP}}}

\newcommand{\atleast}[1]{\xrightarrow{\makebox[4mm]{\mbox{$#1$}}}}
\newcommand{\Atleast}[1]{\xrightarrow{\makebox[8mm]{\mbox{$#1$}}}}
\newcommand{\AAtleast}[1]{\xrightarrow{\makebox[11mm]{\mbox{$#1$}}}}

\newcommand{\pair}[1]{\node{#1}}
\newcommand{\timeof}{\mathsf{{time}}}

\newcommand{\base}{\ensuremath{\mathsf{base}}}
\newcommand{\head}{\ensuremath{\mathsf{head}}}
\newcommand{\tail}{\ensuremath{\mathsf{tail}}}

\newcommand{\weight}{\mathsf{wt}}

\newcommand{\te}[1]{\ensuremath{\theta_{#1}}}

\newcommand{\lam}{\rightsquigarrow}
\newcommand{\lamr}{\rightsquigarrow_r}
\newcommand{\ocomp}[2]{#1{\scriptstyle\odot}#2}

\newcommand{\concatSym}{\cdot}
\newcommand{\concat}{\!\concatSym\!}

\newcommand{\Comp}[2]{\ensuremath{#1\cdot #2}}
\newcommand{\comp}[2]{#1\concat#2}
\newcommand{\Ksub}[1]{ K_{\!\;\!_{#1}}}
\newcommand{\basic}{\mathtt{basic}}

\newcommand{\GB}[1]{\ensuremath{G_{B}({#1})}}
\newcommand{\Gb}{\ensuremath{G_{B}}}

\newcommand{\VB}{\ensuremath{V_{B}}}
\newcommand{\EB}{\ensuremath{E_{B}}}
\newcommand{\GE}[1]{\ensuremath{G_{E}({#1})}}
\newcommand{\VE}{\ensuremath{V_{E}}}
\newcommand{\EE}{\ensuremath{E_{E}}}

\newcommand{\Early}{\mathsf{Early}}
\newcommand{\Late}{\mathsf{Late}}
\newcommand{\early}[3]{\mathsf{Early}\langle#1\atleast{#3}#2\rangle}
\newcommand{\late}[3]{\mathsf{Late}\langle#1\atleast{#3}#2\rangle}
\newcommand{\bearly}[3]{\bm{\mathsf{Early}\langle#1}\atleast{\bm{#3}}\bm{#2\rangle}}
\newcommand{\blate}[3]{\bm{\mathsf{Late}\langle#1}\atleast{\bm{#3}}\bm{#2\rangle}}
\newcommand{\sigC}{\sigma_{\scriptscriptstyle{C}}}
\newcommand{\joins}{S}
\newcommand{\tet}{\ensuremath{\theta}}
\newcommand{\past}{\mathsf{past}}
\newcommand{\Prot}{\mathtt{P}}

\newcommand{\tlf}{two-legged fork}
\newcommand{\lbSym}{\bm{L}}
\newcommand{\ubSym}{\bm{U}}

\newcommand{\upb}[1]{\ub{#1}}

\newcommand{\oCompose}{\odot}

\newcommand{\Ts}[1]{\ensuremath{T_r^{#1}}}
\newcommand{\Tst}{\Ts{\tet}}

\newcommand{\angles}[1]{\ensuremath{\langle{#1}\rangle}}
\newcommand{\nodeL}[1]{\ensuremath{({#1})}}
\newcommand{\nodeF}[1]{\ensuremath{\angles{#1}}}

\newcommand{\Go}{``{\sf go}''}
\newcommand{\sigAware}{\ensuremath{\sigma\text{-recognized}}}
\newcommand{\Vyes}[1]{\ensuremath{V^r_\sigma({#1})}}
\newcommand{\Vno}[1]{\ensuremath{\overline{V}^r_\sigma({#1})}}
\newcommand{\Ayes}[1]{\ensuremath{A^r_\sigma({#1})}}
\newcommand{\Ano}[1]{\ensuremath{\overline{A}^r_\sigma({#1})}}
\newcommand{\VrInit}{\ensuremath{V^{r,0}}}

\renewcommand{\em}{\it}
\renewcommand{\cref}{\Cref}

\hypersetup{
  pdfauthor      = {Asa Dan <asa.dan.2@gmail.com>, and Rajit Manohar <rajit.manohar@yale.edu>, and Yoram Moses <moses@ee.technion.ac.il>},
  pdftitle       = {Using Time Without Clocks},
}

\begin{document}
 
\title{%
On Using Time Without Clocks
via Zigzag Causality 
}
\author{Asa Dan\\Technion\\{\small asadan@campus.technion.ac.il}
\and Rajit Manohar\\Yale University\\{\small rajit.manohar@yale.edu}
\and Yoram Moses\\Technion\\{\small moses@ee.technion.ac.il}
}
\date{}
%

\begin{titlepage}
\maketitle
\begin{abstract}
%
Even in the absence of clocks, time bounds on the duration of actions enable the use of time for distributed coordination. This paper initiates an investigation of coordination in such a setting. A new communication structure called a {\em zigzag pattern} is introduced, and shown to guarantee bounds on the relative timing of events in this clockless model. Indeed, zigzag patterns are shown to be necessary and sufficient for establishing that events occur in a manner that satisfies prescribed bounds. We capture when a process can know that an appropriate zigzag pattern exists, and use this to provide 
necessary and sufficient conditions for timed coordination of events using a full-information protocol in the clockless model. 
\end{abstract}

\vfill
\noindent
\textbf{Keywords}:
coordination, time, clocks, temporal ordering, time bounds.
\vfill

\thanks{This is an extended version of a paper that has been accepted to PODC 2017. This research was supported in part by a Ruch grant from the Jacobs Institute at Cornell Tech, and by a grant from the United States-Israel Binational Science Foundation (BSF), Jerusalem, Israel, and the United States National Science Foundation (NSF) under grant CCF 1617945. Yoram Moses is the Israel Pollak academic chair at the Technion}

\thispagestyle{empty}
\end{titlepage}


\section{Introduction}
\label{sec:intro}

Coordination is a fundamental task in distributed systems. The order in which events take place, and often also the relative timing of the events, can be of primary concern in many applications. 
Timing of actions can be useful, for example, when we wish to dispatch trains in a manner that ensures proper use of critical single-lane sections of the track, or schedule plane takeoffs to alleviate unnecessary congestion at the destination airports.  

In asynchronous systems processes have no access to clocks and, furthermore, they have no timing information except for what they can obtain based on the happens-before relation of \cite{Lamclocks}. 
In such systems, only  the ordering of events can be determined, and not their relative timing. 
Using accurate clocks, it is possible to orchestrate much more finely-tuned temporal patterns of events at the different sites than in the asynchronous setting. Moreover, this can often be achieved using significantly less communication~\cite{LamSM,BzMacm}. 
Of course, clocks are not always available, and when they are, maintaining clocks accurate and in synchrony is far from being automatic. 
But even when processes do not have access to clocks, 
bounds on the timing of communication and of actions are routinely monitored, and system designers can often have access to reliable timing information \cite{giusto2001reliable}. 

This paper initiates an investigation of the use of time for distributed coordination when processes do not have clocks, but do have bounds on the duration of events and of communication.   We call this the {\em bounded communication model without clocks} (or the ``clockless model'' for short), and denote it by \bcm. 
It is not {\it a priori} obvious that the clockless model is any more powerful than the asynchronous model. 
We will show that it is. Throughout the text we will use notation borrowed from~\cite{MB} that states timed precedence between events. 
We write $e\atleast{x}e'$ to state that the event~$e$ takes place at least~$x$ time units before~$e'$ does.%
\footnote{As discussed in~\cite{MB}, although $e\atleast{x}e'$ states a lower bound of~$x$ on the time difference between the events (i.e., $t_{e'}\ge t_e+x$), 
 the same notation can also be used to state upper bounds. Since  $t_{e'}\le t_e+y$  is equivalent to $t_e\ge t_{e'}-y$, we can capture an upper bound of~$y$ on how much later~$e'$ occurs by $e\atleast{-y}e'$. } In an asynchronous system, only the relative ordering of events can be coordinated, and not their timing. Thus, for example, the only way to ensure that an action~$\sfb$ is performed by process~$B$ no more than 10000 time steps before~$\sfa$ is performed by process~$A$ (in our notation, this is denoted by $\sfa\Atleast{\scriptscriptstyle{-10000}}\sfb$) is by having~$B$ perform~$\sfb$ {\em after}~$\sfa$. This is far from optimal, of course, as~$\sfb$ could be performed way before~$\sfa$. Similarly, the only way that
$B$ can be sure to act before~$A$ does, is if~$\sfa$ cooperates, and waits until a message chain from~$B$ informs it that~$\sfb$ has been performed.
As we shall see, in the clockless model it is possible to allow~$B$ (or~$A$) to act much earlier. 
In this paper, we will focus on two basic coordination problems in which~$B$ should act in a manner that is causally and temporally related to~$A$'s action, without requiring~$A$ to adjust its own decision to act, and often without requiring any communication between the two.   

\begin{definition}[Timed Coordination]
\Dlabel{def:1}
Given processes~$A$, $B$ and~$C$, suppose that~$A$ performs the action~$\sfa$ when it receives a \Go\ message from~$C$. 
Moreover, assume that~$C$'s decision to send this message is spontaneous. We define 
two coordination problems:%
 \begin{itemize}
\item[]$\bearly{\sfb}{\sfa}{x}$,\quad in 
which~$B$ should perform~$\sfb$  at least~$x$ time units before~$\sfa$ is performed;  
and 
\item[]$\blate{\sfa}{\sfb}{x}$,\quad 
in which similarly $B$ should perform~$\sfb$ at least~$x$ time units after~$\sfa$ is performed.
\end{itemize}
In both cases, $\sfb$ should  be performed in a run only if~$\sfa$ is performed. 
\end{definition}

In each of these coordination problems, $A$ acts unconditionally when it receives~$C$'s message.
Process~$B$ needs to perform~$\sfb$ only if it can do so in a manner that conforms to the stated bounds. 
Suppose that we wish to ensure that $\sfa\atleast{\scriptstyle{0}}\sfb$, i.e., that~$\sfa$ occurs no later than~$\sfb$.
We can of course ensure this by creating a message chain from~$A$ to~$B$, which starts at or after the occurrence of~$\sfa$. Once the final message in this chain  is received, $\sfb$ can safely be performed. 
In an asynchronous system, such a message chain would be necessary. 
Is it possible to ensure in our model that~$\sfa$ happens before~$\sfb$ without creating a message chain from~$A$ to~$B$?

\begin{figure}[h]
\begin{center}
\includegraphics[height=2in]{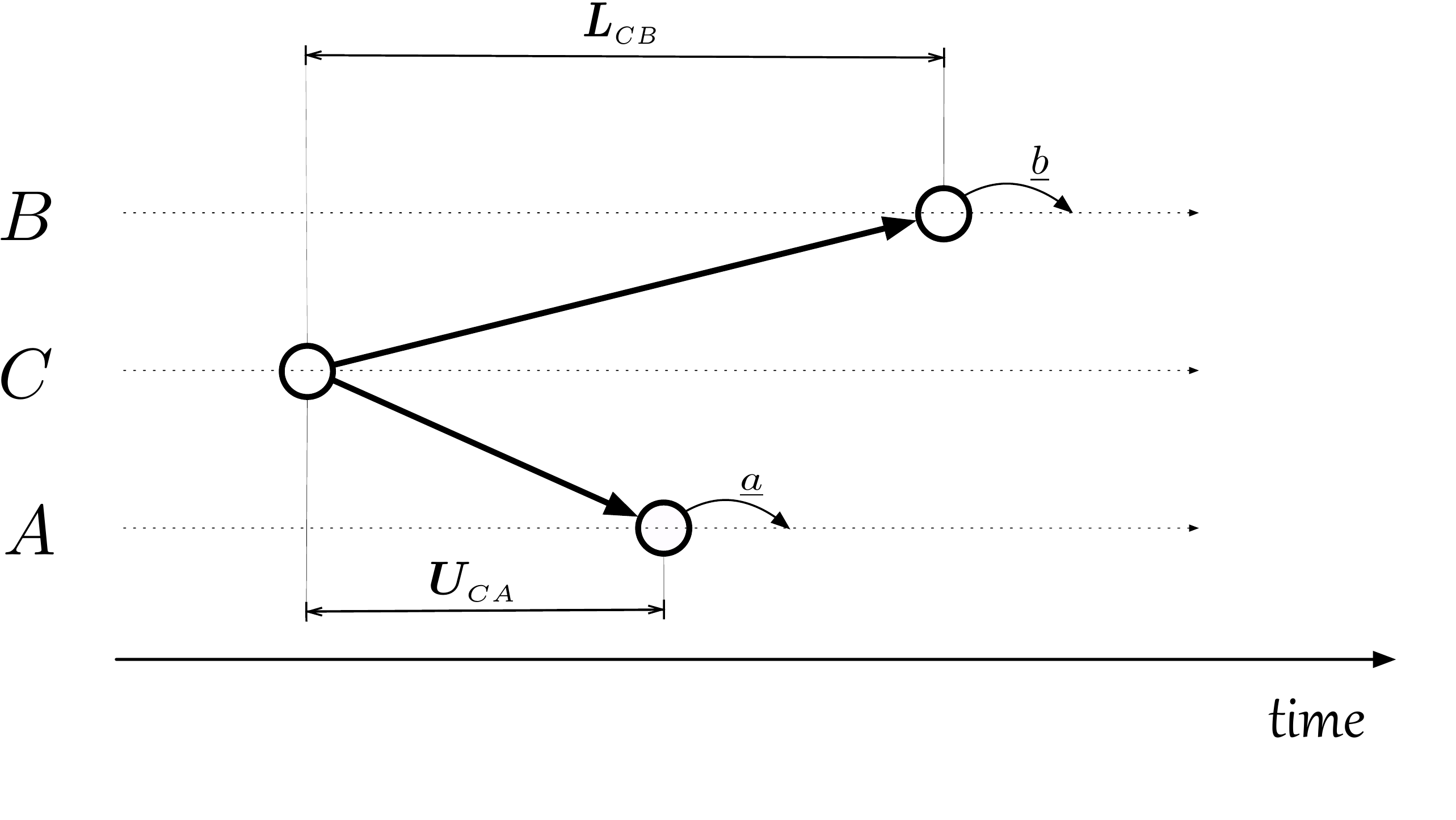}
\end{center}
\caption{Coordination without direct communication}
\label{fig:1}
\end{figure}

Let us return to the question of ensuring that $\sfa\atleast{\scriptstyle{x}}\sfb$ for 
a general value of~$x$, and consider the example depicted in \Cref{fig:1}. Here  process~$C$  simultaneously sends messages to~$A$ and~$B$. Let us denote by $\ub{CA}$ the upper bound on message transmission times for the channel $CA$,  and by~$\lb{CB}$ the lower bound  for $CB$. It is easy to check that if $\lb{CB}\ge \ub{CA}+x$ then~$B$ is guaranteed to receive~$C$'s message no less than~$x$ time units after~$A$ receives it. In that case, as $A$ is assumed to perform $\sfa$ upon receiving~$C$'s message, it is possible to ensure that~$\sfb$ happens at least~$x$ time units later than~$\sfa$, by having $B$ perform~$\sfb$ upon receiving~$C$'s message. Notice that this guarantees  a (timed) causal connection between actions at~$A$ and at~$B$ even without any communication between~$A$ and~$B$. 

Clearly, the analysis underlying the example of \Cref{fig:1} remains valid if we replace each  of the direct messages from~$C$ to~$A$ and~$B$ by a message chain, and replace the condition 
$\lb{CB}\ge\ub{CA}+x$ by a requirement that the sum of lower bounds along the chain from~$C$ to~$B$ exceeds~$x$ plus the sum of upper bounds along the chain from~$C$ to~$A$.
We remark that, in a precise sense, the asynchronous solution (for a simple happened-before requirement) is an instance of \Cref{fig:1} in which $C=A$, and  $\lb{CB}>0=\ub{CA}$. 

Note that if $\lb{CB}<\ub{CA}+x$ then, by waiting for more than $\delta=\ub{CA}-\lb{CB}+x$ time units before performing~$\sfb$, process~$B$ would also ensure that~$\sfa$ takes place $x$ time units
before~$\sfb$ does. 
But we are assuming that~$B$ \defemph{has no clock or timer} that it can use to measure the passage of~$\delta$ time steps. It can only use bounds on communication or internal actions to estimate the passage of time. 
We remark that in current-day technology, clocks and timers are often available. A vast portion of computer chip come with built-in clocks, and highly accurate clock synchronization algorithms are by now standard~\cite{IEEE1588}. But this does not cover all distributed systems of interest. Indeed, it is becoming popular to consider biological systems such as the brain or human body as instances of distributed systems. There, no explicit clock can be found, although timing appears to play a role~\cite{konishi1986,synaptic}. Another setting that fits the \bcm\ model is that of asynchronous (or self-timed) VLSI circuits, where there is no clock but there are bounds on data transfer along wires and on delays of gates~\cite{myerscircuit}. 

A natural question at this point is whether the pattern depicted in \Cref{fig:1} is typical for coordinating actions based on transmission bounds. In other words, is this essentially the only way in which~$B$ can guarantee that~$\sfa$ will be performed (sufficiently long) before~$\sfb$ in the clockless model? Interestingly, the answer is No. 
Consider the scenario depicted in \Cref{fig:2}. In this case~$E$ sends a message to~$B$ and to~$D$, while~$C$ sends a message to~$D$ and to~$A$. Moreover, $D$ receives $C$'s message before it receives~$E$'s message. Finally, $A$ performs~$\sfa$ upon receiving~$C$'s message, and~$B$ performs~$\sfb$ when it receives the message that~$E$ sends. 
As depicted in \Cref{fig:2}, denote the sending times of~$C$ and~$E$'s messages by $t_c$ and $t_e$. Moreover, let~$t_a$ and $t_d$ be the times at which $A$ and~$D$ receive~$C$'s message, and let~$t_b$ the time at which~$B$ receives~$E$'s message. 
%
%
%
%
%
%
%
%
%
Clearly, $\sfb$ is performed no earlier than time~$t_e+\lb{EB}$, yielding inequality~(i) below.
Similarly, the action~$\sfa$ is performed no later than time $t_c+\ub{CA}$, yielding inequality~(iv).
However, the fact that $E$'s message to~$D$ arrives after~$C$'s message arrives implies that~$t_e$ can not be pushed ``too far'' back relative to~$t_c$. After all, 
the message along~$ED$ took no more than~$\ub{ED}$ time units (inequality~(ii)), and the one along $CD$ took no less than $\lb{CD}$ (inequality~(iii)).
%
Altogether, we have: \\[.8ex]
\white{.}$~$\quad\quad\qquad\begin{tabular}{rlcl}
(i)~~& $t_b$  & $\ge$ & $t_e+\lb{EB}$,\\
(ii)~~ &$t_e$&$>$ &$t_d-\ub{ED}$,\\
(iii)~~ &$t_d$&$\ge$ &$t_c+\lb{CD}$, \quad~~~and\\
(iv)~~ &$t_c$&$\ge$ &$t_a-\ub{CA}$.\\[1.5ex]
\end{tabular}

By substitution, we have that~ 
  $t_b>t_a-\ub{CA}+\lb{CD}-\ub{ED}+\lb{EB}$. 
  Thus,  $t_b>t_a+x$ is guaranteed in this case if 
 \begin{equation}\label{eq:bounds}
 -\ub{CA}+\lb{CD}-\ub{ED}+\lb{EB}~~~\ge~~~ x.
 \end{equation}

\begin{figure*}
\centering
\begin{subfigure}{.5\textwidth}
  \centering
  \includegraphics[width=.6\linewidth]{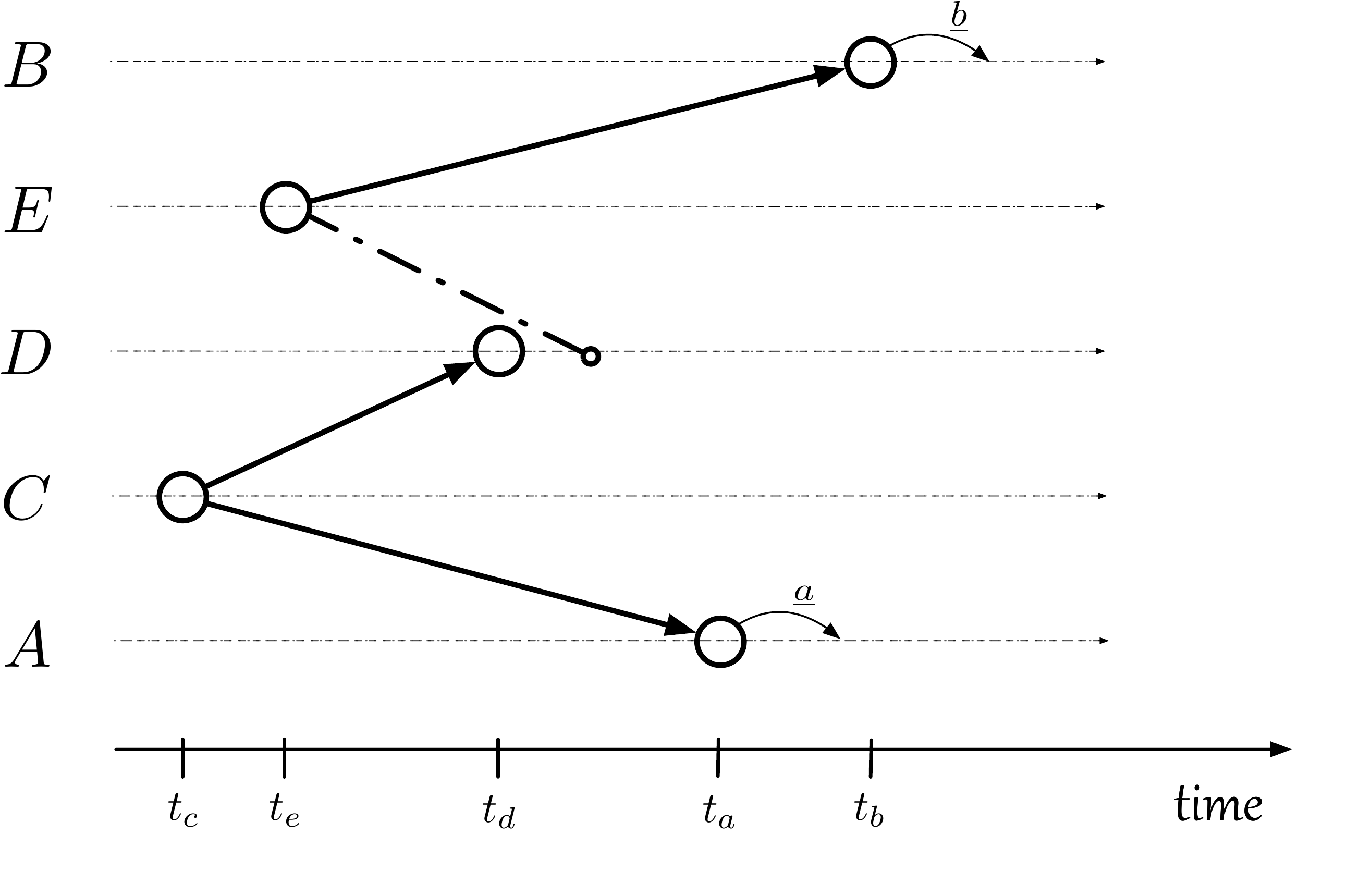}
  \caption{~Zigzag-based happens-before\\$~$}
  \label{fig:2}
\end{subfigure}%
\begin{subfigure}{.5\textwidth}
  \centering
  \includegraphics[width=.6\linewidth]{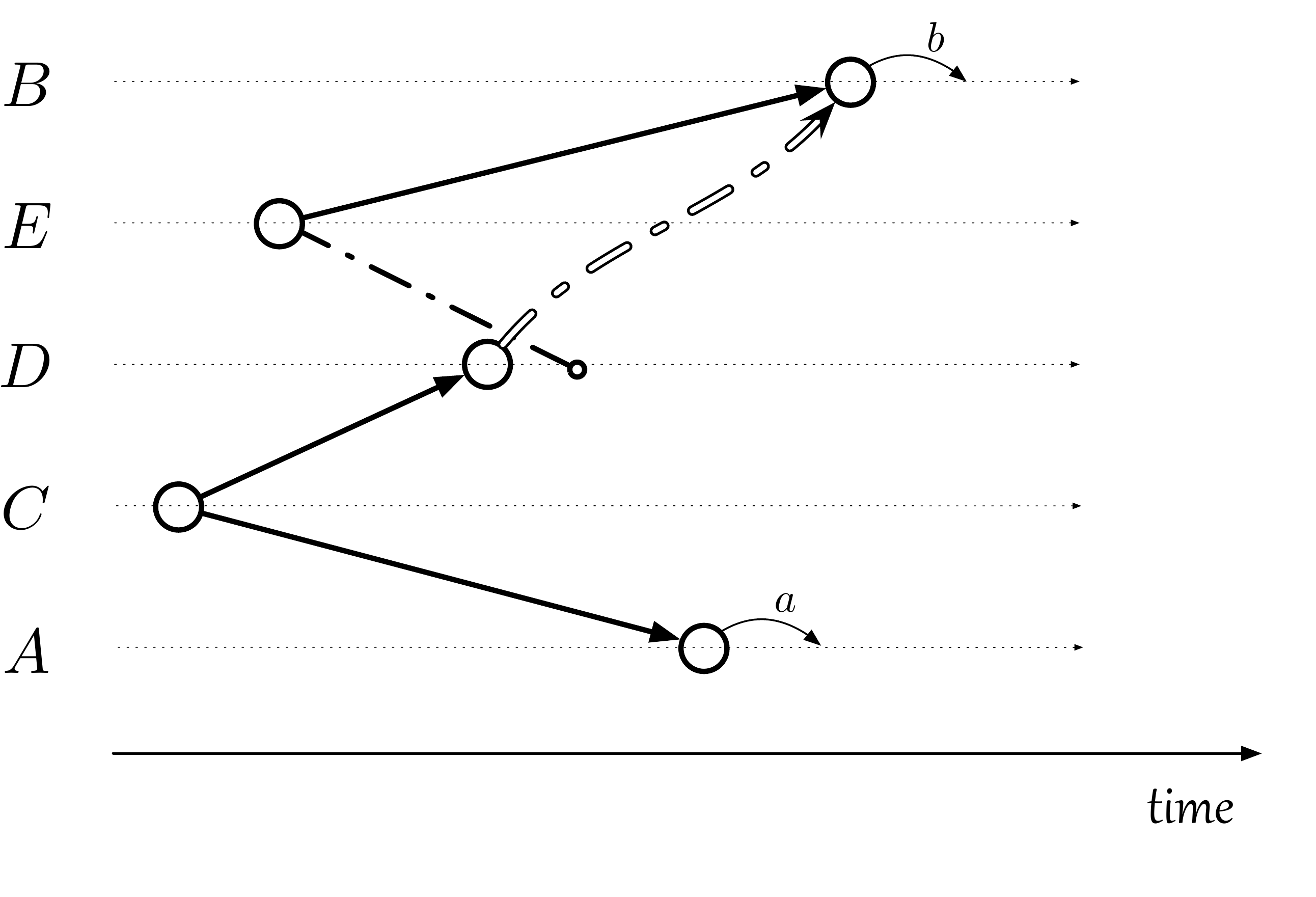}
  \caption{~Zigzag-based coordination\\$~$}
  \label{fig:3}
\end{subfigure}
\caption{A zigzag communication pattern}
\label{bigfig:2}
\end{figure*}

The reader may correctly suspect at this point that  the zigzag pattern of \Cref{fig:2}
can be extended by adding an arbitrary finite number of additional zigs and zags. Indeed, in that case a more elaborate condition in the style of \Cref{eq:bounds}, based on a longer derivation, will ensure that~$\sfa$ happens more than~$x$ time units before~$\sfb$. The first result of our analysis is a proof that this is tight. 
We will show that, in a precise sense, the existence of an appropriate zigzag pattern is a necessary condition for~$B$'s performing the action~$\sfb$ in $\late{\sfa}{\sfb}{x}$ or $\early{\sfb}{\sfa}{x}$.  

Since a zigzag pattern is necessary, we have by the {\em Knowledge of Preconditions principle} of \cite{TarKop} that~$B$ must {\em know} that a zigzag pattern exists when it performs~$\sfb$.
Interestingly, even if the processes follow a full-information protocol,%
\footnote{A {\em full-information protocol} (\fip) is one in which every message sent encodes the sender's complete history up to the point at which it is sent.} the existence of a zigzag pattern does not necessarily enable process~$B$ to correctly coordinate its action. 
In a run containing the pattern of~\Cref{fig:2}, for example, if~$B$ receives no messages from~$D$%
, then~$B$ would not be able to detect the existence of the zigzag pattern, because in the eyes of $B$ it may be possible that $C$ will send its messages only in the far future.

There are, of course, cases in which~$B$ can detect that an appropriate zigzag pattern exists. In such a case, $B$ can decide to perform its action~$\sfb$ and be sure that~$\sfa$ happens before~$\sfb$.  Consider \Cref{fig:3}. Suppose that the bounds satisfy the condition of \Cref{eq:bounds}. Moreover, let's assume that every message contains a header specifying who its intended recipients are. 
Once~$B$ receives~$E$'s message (with an indication that it was also sent to~$D$), and receives a message from~$D$ (denoted by a dashed line in \Cref{fig:3}) indicating that~$D$ heard from~$C$  before~$D$ heard from~$E$, then~$B$ can perform~$\sfb$ and be guaranteed to satsify $\late{\sfa}{\sfb}{x}$.
This is an instance of a {\em visible zigzag} pattern, which is a zigzag pattern that is extended by an appropriate set of message chains. 
Our analysis will identify particular visible zigzag patterns as necessary and sufficient for $B$'s action in instances of $\early{\sfb}{\sfa}{x}$ or $\late{\sfa}{\sfb}{x}$.

This paper is organized as follows: Section~\ref{sec:model} introduces the bounded communication model, protocols, and standard aspects of causality. It also presents two ways of describing a point on a local timeline: Since processes do not have access to clocks, one way is in terms of the local state that the process is in, and another way is as the point at which a message chain arrives from a point of the first type. Section~\ref{sec:zig-prec} introduces zigzag patterns and shows that they are necessary and sufficient for guaranteeing a precedence relation. In \Cref{sec:using} the notion of a visible zigzag pattern is introduced, and it is shown to be necessary and sufficient for optimal behavior in a coordination task. Section~\ref{sec:highlights} sketches the two variants of bounds graphs used in our technical analysis, and sketches the ideas underlying the proofs of our main theorems. Finally, Section~\ref{sec:conclusions} provides a discussion of our results, their implications, and direction for future work. Detailed proofs of our theorems and claims appear in \proofs.

\subsection{Related Work}
\label{sec:related}
Lamport's seminal work \cite{Lamclocks} on the ordering of events in asynchronous systems introduced the happened before relation, and initiated an orderly account of the role of causality in the ordering of events. Roughly speaking, his work shows that the only way to implement instances of $\late{\sfa}{\sfb}{0}$ in an asynchronous setting is by constructing a message chain from~$A$ to~$B$. Using the terminology of~\cite{Lamclocks}, one can consider the {\em causal past} of an event in an asynchronous setting to be the set of events from which it has received a message chain. In our analysis, this set also plays an important role. However, the bounds provide partial information on the  timing of events in the past, and, moreover, the past guarantees the occurrence of events that are not seen by the process (i.e., they do not have explicit message chains to the process).

Clocks are very useful tools for coordinating actions in distributed systems (see, e.g., \cite{LynchShavitME,corbett2013spanner,dickerson2010time,harris2008application,MorrisonTSO,G8271,WhiteRabbit}).
There is a vast literature on real-time systems and on time in multi-agent systems (see, e.g., \cite{kopetz2011real}). Clock synchronization based on bounds on message transmission times was studied extensively in the 70's and 80's \cite{Attiya1996,DHS,DHSS,HMM,LM,lundelius1984upper,PSR,SrikT}; see  \cite{simons1990overview} for an early survey. 
 
One important aspect that our work shares with the clock synchronization literature is the fact that bounds on the duration of events or on transmission times play an important role. Indeed, some of our technical analysis is based on bounds graphs that are strongly inspired by~\cite{PSR} and~\cite{MB}. In particular, the notion of timed precedence we use comes from~\cite{MB}. Our study diverges from the existing literature in the fact that no clocks or timers whatsoever are assumed, and the only timing information comes from the observed events and the guaranteed bounds. 
 
An early suggestion to use knowledge to study time and coordination appeared in~\cite{Moses:know92}. Knowledge theory has been used to investigate the protocols and communication patterns that can solve coordination problems in systems with global clocks or accurate timers in
\cite{BzMacm,BzM2011,ben2013agent,BzMshape,GMvectorial}. In such settings, these works provide tools for a wide variety of coordination tasks. 


\section{Model and Preliminary Definitions}
\label{sec:model}
\subsection{The  Bounded Communication Model}
\label{sec:cbm}
We focus on a simple setting of 
a communication network $\Net=(\Proc,\Chan)$ modeled by a directed graph whose nodes are the processes $\Proc=\{1,\ldots,n\}$ and whose edges are the communication channels among them. 
We identify time with the natural numbers, $\Nat$, where a single time step should be thought of as the minimal relevant unit of time. There are lower and upper bounds on message transmission times per channel, specified by a pair of functions \mbox{$\lb{},\ub{}\!:\!\Chan\to\Nat$}, that satisfy $1\le\lb{ij}\le\ub{ij}<\infty$ for all $(i,j)\in\Chan$.

Paths in~$\Net$ are specified by sequences of process names. 
We denote a singleton sequence~$[i]$ simply by~$i$, and the concatenation of sequences~$p$ and~$q$  by $\comp{p}{q}$. 
We define the composition of two sequences $p=[i_1,\ldots,i_k,j]$ and $q=[j,h_1,\ldots,h_m]$ in which the last element of~$p$ coincides with the first element of~$q$ by 
$\ocomp{p}{q}\eqdef [i_1,\ldots,i_k,j,h_1,\ldots,h_m]$. 
We extend the notation of upper and lower bounds on message transmission times to paths {$p=[i_1,\ldots,i_d]$} in the network graph by defining 
\[\lb{}(p)\eqdef\!\sum\limits_{i_k=1}^{d-1}\lb{i_ki_{k+1}}\mbox{\quad and\quad }\,\ub{}(p)\eqdef\!\sum\limits_{i_k=1}^{d-1}\ub{i_ki_{k+1}}.\]

{For ease of exposition, actions are assumed to be instantaneous}.%
\footnote{Our analysis will apply even in the case in which actions extend over time. 
Such an action will be modeled as a special channel from the process to itself, with lower and upper bounds for the channel. The invocation and completion of the action would each be instantaneous events.}
A global state (or a snapshot of the system) will have the form $\bm{g}=(\ell_e,\ell_1,\ldots,\ell_n)$, consisting of a state~$\ell_e$  for the environment~$e$, and one local state~$\ell_i$  for every process $i\in\Proc$. 
A tuple $\gamma=\big((\Net,\lb{},\ub{}),\Gz\big)$ whose first component is a time-bounded network as described above, and $\Gz$ is a set of possible initial global states, 
is called the \defemph{context} in which a protocol operates. 

A run~$r$ is an infinite sequence of global states. Thus, $r(m)$ is a global state for every $m\in\Nat$. We denote by $r_i(m)$ process~$i$'s local state in $r(m)$. 
Processes can perform (application-dependent) local actions and send messages along their outgoing edges. For simplicity, we will assume that the local state of a process consists of an initial state followed by the sequence of events (local actions, message sends and message receives)  that the process has observed. 

The bounds on message delivery are enforced by assuming the existence of a scheduler, which we call the {\em environment}. The environment's local state contains the current contents of all channels in~$\Chan$, and for every message in a channel it also records the time at which the message was sent. 
At any point in time, the environment can deliver messages to each of the processes. It can nondeterministically choose whether or not to deliver a message~$\mu$  in a channel $(i,j)\in \Chan$ at time~$t$ if the sending time~$t_\mu$  of~$\mu$ satisfies 
$\lb{ij}\le t-t_\mu<\ub{ij}$. The environment \defemph{must} deliver~$\mu$ to~$j$ at time~$t$ if $t-t_\mu=\ub{ij}$. 
We remark that if a message~$\mu$ is delivered to~$i$ at time~$t$ in the run~$r$, then~$i$'s local state at time~$t$, $r_i(t)$, will  record the fact that~$i$ received~$\mu$.

We assume a set~$\Ext$ of \defemph{external} messages, where the environment may nondeterministically choose at any point ($t>0$) whether to deliver messages from~$\Ext$ to an arbitrary process. Such delivery is spontaneous, and is independent of other (past or present) events in the run. 
For simplicity we assume that a particular external message of~$\Ext$ can be delivered to at most one process in a given run. 
Since processes  in the \bcm\ model have no clocks, we assume that their actions are event based. A process is scheduled to move only when it receives messages (either external or internal).%
\footnote{In particular, processes do not spontaneously perform actions at time~0.}
  It can then perform a finite sequence of actions. 

Recall that we assumed in \Cref{def:1} that~$C$'s decision to send a \Go\ message in an instance of $\early{\sfb}{\sfa}{x}$ or $\late{\sfa}{\sfb}{x}$ is spontaneous. Formally, we will assume that there is a message $\mu_{go}\in\Ext$ such that~$C$ will send the \Go\ message to~$A$ when it receives~$\mu_{go}$. 

Processes follow a protocol $\Prot =(\Prot _1,\ldots,\Prot _n)$, where $\Prot _i$, process~$i$'s protocol,  is a deterministic  function of~$i$'s local state. 
A specific class of protocols we use in this paper are what we call \defemph{flooding full-information protocol}s (\ffip).
An \ffip\ is a protocol in which each process that receives a message immediately sends a message, containing its entire local state, to all of its neighbors.
In a precise sense, \ffip's are general protocols for~\bcm: Just as with standard full-information protocols in the synchronous model (see, e.g., \cite{coan1986}), it is possible to simulate any given protocol in the \bcm\ model by one that communicates according to the \ffip.


Given a protocol $\Prot$ and a bounded context $\gamma=\big((\Net,\lb{},\ub{}),\Gz\big)$, we denote by $\Rrep=\Rrep(\Prot ,\gamma)$ the set of runs of~$\Prot $ in context~$\gamma$. 
We call it the system representing~$\Prot$ in~$\gamma$. 
A run~$r$ belongs to $\Rrep$ exactly if (1) $r(0)\in\Gz$, and (2) for all $m>0$, $r(m)$ is obtained from $r(m-1)$ following the rules described above. 
A more formal definition appears in \proofs.
Henceforth, whenever a system is mentioned, it is assumed to have this form. 
We say that a given protocol~$\Prot$  \defemph{implements} $\early{\sfb}{\sfa}{x}$ (resp.~$\late{\sfa}{\sfb}{x}$) if in all runs $r\in\Rrep(\Prot,\gamma)$ process~$A$ performs~$\sfa$ when it receives the \Go\ message, and process~$B$ performs~$\sfb$ in~$r$ only if~$\sfa$ is performed in~$r$, and only at a time that is consistent with the specification of~$\early{\sfb}{\sfa}{x}$ (resp.\ the specification of~$\late{\sfa}{\sfb}{x}$).

\subsection{Reasoning about \bcm\ Systems}
\label{sec:reasoning}
In the coordination problems $\early{\sfb}{\sfa}{x}$ and $\late{\sfa}{\sfb}{x}$ specified in definition~\ref{def:1}, process $B$ needs to decide whether and when to perform a particular action~$\sfb$.
In particular, it needs to estimate the relative time difference between points on different processes' timelines: It's current point, and the point at which~$A$ performs~$\sfa$. 
%
Because processes have no clocks, formally defining points on a timeline is somewhat subtle. 
Rather than distinguishing the points along the timeline of a given process according to the times at which they arise,  to which processes have no access, one useful way is to identify a local point with the local state of the process. 
We call  a pair $\sigma=\bnode{i,\ell}$ consisting of a process name and a local state for this process a {\em basic node}. 
In order to emphasize its site~$i$, we sometimes call such a node an $\bm{i}$-\defemph{node}.
We say that a basic node~$\bm{\sigma}=\bnode{i,\ell}$ \defemph{appears in~$\bm{r}$} if $r_i(m)=\ell$ holds for some time~$m$.
%

While the local state of a process in the FFIP protocol does not repeat twice in non-contiguous intervals of the same run, a local state can remain constant along some time interval. 
During such an interval  the process cannot observe the passage of time; it observes only the state transitions. 
For a basic node~$\sigma=\bnode{i,\ell}$ that appears in a run~$r$, 
we define $\bm{\timeof_r(\sigma)}$ to be the minimal~$m$ such that \mbox{$r_i(m)=\ell$}.
This allows us to treat a basic node as specifying a particular (externally observable) time in the run.%
\footnote{Since processes act in an event-driven fashion, 
 the time at which~$i$ acts in~$r$ when in local state~$\ell$  is precisely $\timeof_r(\sigma)$.}
  While a run~$r$ can be uniquely determined by the set of its basic nodes and their respective times, different runs can possess the same set of basic nodes, and differ in their timing.
 
 For a given site~$i$, an $i$-node $\sigma'$ is called a {\em successor} of another~$i$-node $\sigma$ in~$r$ if 
 $\timeof_r(\sigma)<\timeof_r(\sigma')$  and there is no $i$-node $\sigma''$ such that 
 $\timeof_r(\sigma)<\timeof_r(\sigma'')<\timeof_r(\sigma')$. 
 If $\sigma'$ is the successor of~$\sigma$, then we call~$\sigma$  the {\em predecessor} of $\sigma'$.
\begin{definition}
Given a run~$r$, we define Lamport's {\em happens-before} relation among basic nodes that appear in $r$, denoted by $\bm{\sigma'\lamr\sigma}$, to be the minimal  transitive relation that satisfies (i) Locality: If both $\sigma'$ and~$\sigma$ are $i$-nodes and $\timeof_r(\sigma')\le \timeof_r(\sigma)$, then $\sigma'\lamr\sigma$, and (ii) if a message is sent in the run~$r$ from~$\sigma'$ and delivered to~$\sigma$, then $\sigma'\lamr\sigma$. We say that $\sigma'$ is \defemph{in the past of~$\bm{\sigma}$} in~$r$ if $\sigma'\lamr\sigma$, and we define $\bm{\past(r,\sigma)}\eqdef\{\sigma' : \sigma'\lamr\sigma\}$.\footnote{
While we use a specific run $r$ in the definition of ``$\lamr$'', since we restrict attention to full-information protocols, the run does not play an essential role, as if $\sigma\lamr\sigma'$ then this relation hold w.r.t.\ run in which both nodes appear.}
\end{definition}

\paragraph{General Nodes}\quad 
 We view a process as having access to its local state, and hence to its current basic node, at any point. Indeed, since processes are assumed to be following a full-information protocol, it also has access to all basic nodes that appear in its past. But the points with which $B$ should coordinate its action (the points where $A$ performs $\sfa$) are often not in its past. So~$B$ is aware neither of the real time at which they occur, nor of the basic nodes, since it cannot identify the local state of process~$A$. Recall, however, that processes are assumed to follow an \ffip\ protocol, in which whenever a process receives a message or external input, it broadcasts this to all of its neighbors. So if $\sigma'\in\past(r,\sigma)$ and~$\sigma'$ is not an initial node (i.e.~not a node from time~0), then there are typically many message chains starting at~$\sigma'$. It is with such nodes that it may need to coordinate, and about whose timing we need to reason. We now define the class of general nodes, which can be defined as being at the end of a path in the network from a given basic node. We proceed as follows. 
 \begin{definition}
 \Dlabel{def:general}
 Let $\sigma$ be a basic $i$-node, and $p$ be a path in $\Net$ that begins at~$i$, then $\nodeF{\sigma, p}$ is a (general) node that describes the basic node that will receive the message chain that goes along $p$ starting at $\sigma$.
We say that $\bm{\tet=\nodeF{\sigma, p}}$ \defemph{appears in} a run $\bm{r}$ if both $\sigma$ appears in $r$, and~$p$ is a path in~$\Net$ (so that  there is a message chain in~$r$ that leaves $\sigma$ and goes along $p$).
\end{definition}
Note that if $p$ is a singleton (i.e. $p=[i]$), then $\sigma$ is an $i$-node and 
$\tet=\nodeF{\sigma, p}$ denotes~$\sigma$ itself. 
However, if $p$ is not a singleton, then~$\tet$ corresponds to a basic node whose identity depends on the run in question. The correspondence is defined as follows. 

\begin{definition} 
\label{def:gnode}
Let $\tet=\nodeF{\sigma, p}$ be a node that appears in the run $r\in\Rrep$.
The basic node that corresponds to $\tet$ in $r$, $\basic(\tet,r)$, is defined inductively as follows: 
\begin{itemize}
\item[(a)] If $\theta=\nodeF{\sigma,j}$ (so $p$ is a singleton), 
then $\basic(\tet,r)=\sigma$.
 \item[(b)] Le~$p$ be non singleton, $p=\Comp{p'}{j}$, and $\basic(\nodeF{\sigma,p'},r)=\sigma'$. 
If the message sent in~$r$ from~$\sigma'$ to process~$j$ is delivered at~$\sigma''$, then $\basic(\nodeF{\sigma,p},r)=\sigma''$.
 \end{itemize}
\end{definition}

General nodes will inherit properties from their corresponding basic nodes. Thus, we define $\timeof_r(\theta)\eqdef\timeof_r\big(\basic(\theta,r)\big)$, we write 
$\tet\lamr\tet'$ iff $\basic(\tet,r)\lamr\basic(\tet',r)$, and call~$\tet$ an $i$-node if $\basic(\tet,r)$ is an $i$-node. 
For a $j$-node $\theta=\node{\sigma,p}$ of~$r$ and a path~$q$ in~$\Net$, where~$q$ begins at process~$j$,
it will be convenient to write $\ocomp{\theta}{q}$ as shorthand for the node $\node{\sigma,\ocomp{p}{q}}$. 

Clearly, a node $\tet'=\nodeF{\sigma',p'}$ can appear in a run~$r$ only if $\sigma'$ appears in $r$. 
However, if~$\sigma$ appears in~$r$ and $\sigma'\not\lamr\sigma$, then~$\sigma$ might not be able to distinguish whether $\tet'=\nodeF{\sigma',p'}$ indeed appears in the current run.
We shall say that a general node $\tet'=\nodeF{\sigma',p'}$ is \defemph{$\bm{\sigma}$-recognized} iff $\sigma'\lamr\sigma$.
Note that in an \ffip\ protocol, $\sigma$ actually ``knows'' that every \sigAware\ node appears in the run.

\section{Zigzag Patterns and Timed Precedence}
\label{sec:zig-prec}
%
We adapt the notion of timed precedence from \cite{MB} to nodes in our setting. Formally, given a run $r\in\Rrep$, we say that 
a run $\bm{r}$~\defemph{satisfies} $\theta\atleast{x}\theta'$, and write $(R,r)\sat \theta\atleast{x}\theta'$, iff
both (i) the nodes~$\theta$ and~$\theta'$ appear in~$r$, and (ii)~\mbox{$\timeof_r(\theta)+x\,\le\,\timeof_r(\theta')$}.
(While the system~$\Rrep$ does not play a role in this definition, it is included here because it will play a role in our later analysis.) \\[-.6ex]

Our discussion in~\Cref{sec:intro} shows that communication as in \Cref{fig:1} ensures that $\sfa\AAtleast{\scriptscriptstyle{\lb{CB}-\ub{CA}}~}\sfb$, and similarly that a pattern as in~\Cref{fig:2} ensures a precedence as captured in \Cref{eq:bounds}. We now define general zigzag patterns and relate them to timed precedence. 
The basic building block is a {\em two-legged fork} (see~\Cref{fig:2legF}): 
\begin{definition}
 \Dlabel{def:two-leg}
A \defemph{two-legged fork} in~$r$ is a 
triple 
$F=\pair{\theta_0,\theta_1,\theta_2}$ 
of
nodes of~$r$,   such that $\theta_1=\ocomp{\theta_0}{p_1}$ and $\theta_2=\ocomp{\theta_0}{p_2}$, for process sequences~$p_1$ and~$p_2$. 
We denote $\base(F)=\theta_0$, $\head(F)=\theta_1$, and $\tail(F)=\theta_2$.  
\end{definition}

In a two-legged fork, there are direct message chains (possibly empty) from the base node to the head and to the tail of the fork. \Cref{fig:1} is an example of a two-legged fork in which the message chains consist of single messages, while \Cref{fig:2legF} depicts one with longer paths from the base node to head and tail nodes.  

\begin{figure}[h]
\label{fig:2legF}
\begin{center}
\includegraphics[height=2in]{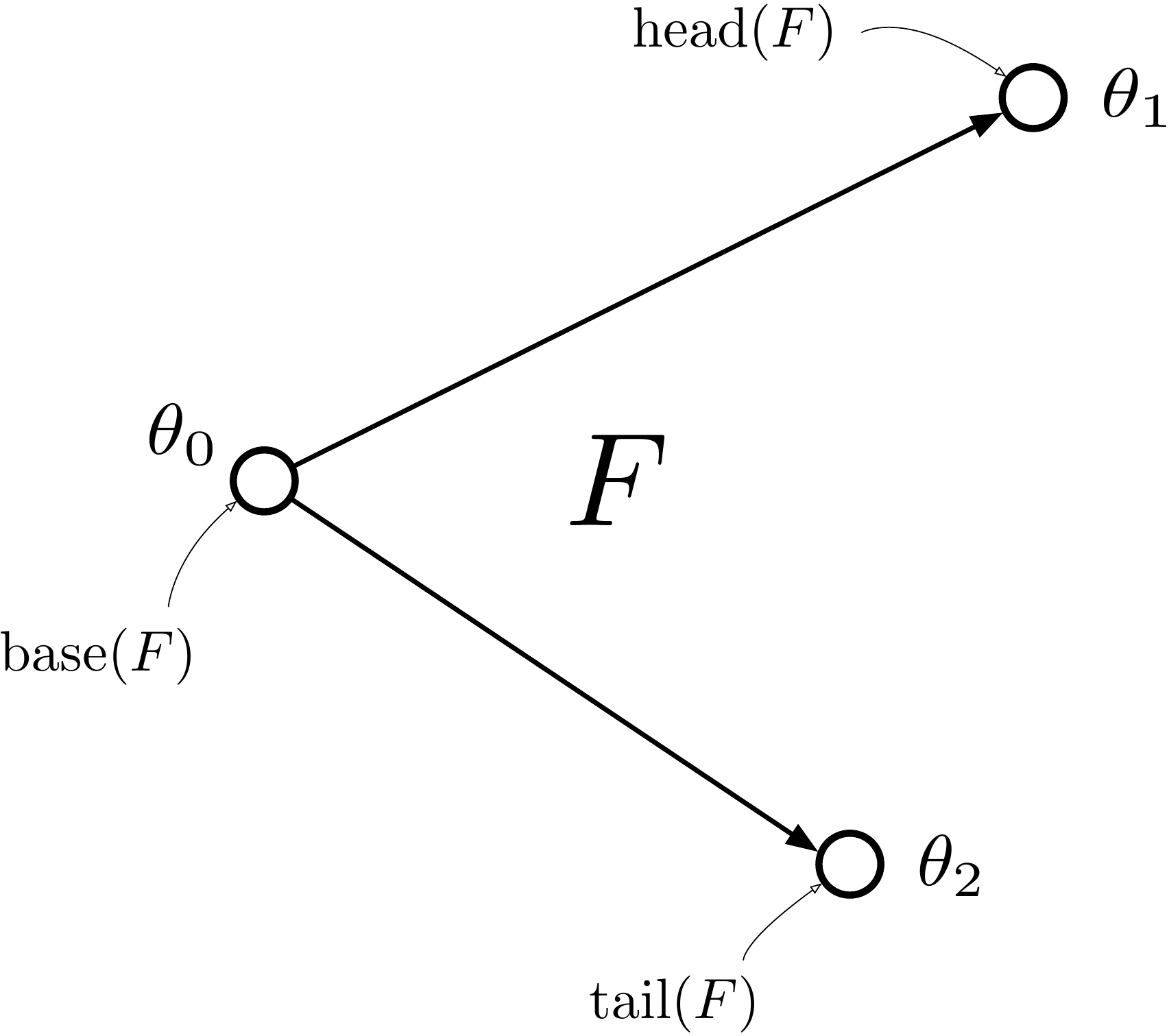}\\[2ex]
\caption{A two-legged fork~$F=\langle \theta_0,\theta_1,\theta_2\rangle$.}
\label{fig:2legF}
\end{center}
\end{figure}
Let $\theta$ be an $i$-node, and let $F=\pair{\theta,\ocomp{\theta}{p_1},\ocomp{\theta}{p_2}}$ be a two-legged fork in~$r$.
We define 
the \defemph{weight} of~$F$ to be 
 \[\weight(F)~~\eqdef~~\lb{}(p_1) - \ub{}(p_2).\]
In \Cref{fig:1}, for example, $\weight(F)=\lb{CB}-\ub{CA}$. 
The existence of a two-legged fork~$F$ in~$r$ with tail~$\theta_1$ and head~$\theta_2$ implies that \mbox{$(R,r)\sat\; \theta_1\Atleast{\weight(F)}\theta_2$}. 

A Zigzag pattern is made of a sequence of suitably composed two-legged forks. 
Roughly speaking, the head of each fork should be on the same timeline as, but appear no later than, the tail of the next fork in the sequence. If they coincide at the same basic node the forks are called {\em joined}. Otherwise, the tail will be at least one time unit later than the preceding head. 
More formally, we define
\begin{definition}
 \Dlabel{def:zig}
A \defemph{zigzag pattern} from node~$\theta$ to~$\theta'$ in the run~$r$ is a sequence $Z=(F_1,\ldots,F_c)$ of two-legged forks in~$r$, with $c\ge 1$, such that~$\tail(F_1)=\theta$ and $\head(F_c)=\theta'$. 
Moreover, if $c>1$ then for every $k=1,\ldots,c-1$ there is a process~$j$  such that 
both~$\head(F_k)$ and $\tail(F_{k+1})$ correspond to $j$-nodes, 
and $\timeof_r(\head(F_k))\le \timeof_r(\tail(F_{k+1}))$. 
\end{definition}
\Cref{fig:2} depicts a zigzag pattern consisting of $c=2$ forks, which are not joined, since the head of the lower fork and the tail of the upper one
correspond to distinct nodes on~$D$'s timeline.

The notion of weight extends to zigzag patterns. Consider a zigzag pattern  $Z=(F_1,\ldots,F_c)$, and denote by $\joins(Z)$ the number of forks $F_k\in Z$ that are not joined to their successor (i.e., $\head(F_k)$ strictly precedes $\tail(F_{k+1})$).
The \defemph{weight} of~$Z$ is defined by 
 \[\weight(Z)~\eqdef~~~\sum\limits_{k=1}^c\weight(F_k)~~+~~\joins(Z).\]
We can thus justify the claim that zigzag patterns are sufficient for establishing timed precedence in~\bcm\ systems: 
\begin{theorem}[Zigzag Sufficiency]
\label{thm:zzsimple}
Let $Z$ be a zigzag pattern from node~$\te{1}$ to $\te{2}$ in the run~$r\in\Rrep$. 
Then
\mbox{$(R,r)\sat\; 
\theta_1\Atleast{\weight(Z)}\theta_2$.}
\end{theorem}
The intuition is that each fork implies a timed precedence between its tail and its head, and the concatenation of forks in the zigzag pattern introduce a simple timed precedence between the head of one fork and the tail of its successor. Recall that, by assumption, if the successive forks are not joined, then they are separated by at least one time unit.
%

What is perhaps more instructive than \Cref{thm:zzsimple} is that, in a precise sense, the only way to guarantee a timed precedence relation is via a zigzag pattern of this type. More formally, 
 we say that a system~$\Rrep$ \defemph{supports} the statement $\te{1}\atleast{x}\te{2}$
if, for all~$r\in\Rrep$, if one of the nodes~$\te{1}$ or $\te{2}$ appears in~$r$, then both nodes appear in~$r$, and 
$(\Rrep,r)\sat\te{1}\atleast{x}\te{2}$.
We can show:
\begin{theorem}[Zigzag Necessity] 
\label{thm:zz}
\label{THM:ZZ} 
Suppose that $\Rrep$ supports \mbox{$\te{1}\atleast{x}\te{2}$}. Moreover, assume that $\te{1}$ and $\te{2}$ both appear in a run~$r\in\Rrep$, with $\timeof_r(\te{1})>0$ and \mbox{$\timeof_r(\te{2})>0$}.
Then there is a zigzag pattern~$Z$ in $r$ from $\te{1}$ to~$\te{2}$ with $\weight(Z)\ge x$. 
\end{theorem}

Suppose that a protocol guarantees a particular time precedence constraint among a given pair of actions. Then it must ensure the existence of an appropriate zigzag pattern in the run. We remark that the requirement that $\timeof_r(\te{2})>0$ in the theorem ensures that the node~$\te{2}$ is not an initial node. In our model, protocols cannot perform actions at initial nodes, and precedence among initial nodes can be obtained without the existence of zigzags.


\section{Using Zigzag Causality for Coordination}
\label{sec:using}
%
Theorems~\ref{thm:zzsimple} and~\ref{thm:zz} show that zigzag patterns are necessary and sufficient for ensuring that a precedence relation between two nodes holds. 
It follows, for example,  that~$B$ can act in an instance of $\early{\sfb}{\sfa}{x}$ only at a node that is the tail of a zigzag pattern of weight~$x$ whose head is the node at which A performs~$\sfa$. (Similarly, the roles of head and tail need to be reversed for an instance of $\late{\sfa}{\sfb}{x}$.)
However, as discussed in the introduction, it is not guaranteed that a node at either end of the zigzag pattern is able to detect the existence of the pattern,
and such endpoint node might not know that the necessary precedence condition holds. 

We will show that in order to act in one of the two coordination tasks we are considering, $B$ must know that $\te{1}\atleast{x}\te{2}$ holds, for the two nodes at which~$A$ and~$B$ act. Next, we will characterize the communication patterns that give rise to such knowledge of a timed precedence, and thus are necessary for coordinating~$A$ and~$B$'s actions. We start by defining an appropriate notion of knowledge for the bcm model, which will allow us to formulate and prove these results.

\subsection{Reasoning About Coordination}
Our focus is on coordinating actions at different sites in a manner that satisfies temporal constraints. 
We use the notion of knowledge of \cite{FHMV} to reason about what a process knows about the relevant aspects of the timing of events. We now describe just enough of the logical framework to support our analysis. 

Two runs $r,r'\in\Rrep$ are said to be \defemph{indistinguishable} at the basic node~$\sigma$, which we denote by $r\sim_\sigma r'$,  if~$\sigma$ appears both in~$r$ and in~$r'$. Intuitively, if $\sigma=(i,\ell)$ appears in both runs, then when~$i$'s local state is~$\ell$, it cannot distinguish whether the run is~$r$ or~$r'$. Knowledge is the dual of indistinguishability. I.e., a fact is known at a node if it is true of all indistinguishable runs. In particular, in this paper, we focus on knowledge of precedence statements at basic nodes. 
 We write $(\Rrep,r)\sat \Ksub{\sigma}(\theta_1\atleast{x}\theta_2)$ to state that in the run~$r\in\Rrep$ the precedence statement is known at the basic node~$\sigma$.  It is formally defined as follows:%
\footnote{It will suffice to define knowledge at basic nodes, here, since our analysis does not concern knowledge about what is known at other nodes. 
For a more general treatment, it is possible to define $(\Rrep,r)\sat \Ksub{\theta}p$ to hold precisely if $(\Rrep,r)\sat \Ksub{\sigma}p$ holds at $\sigma=\basic(\theta,r)$.}  
 
 \begin{tabular}{lll}
 $(\Rrep,r)\sat \Ksub{\sigma}(\theta_1\atleast{x}\theta_2)$ & iff & $(\Rrep,r')\sat \theta_1\atleast{x}\theta_2$ \white{.}~~~holds \white{for} \\
 && \hspace{-3mm}for all~$r'\in\Rrep$ such that $r\sim_\sigma r'$.
 \end{tabular}
When performing the action~$\sfb$ in solving a coordination problem such as $\early{\sfb}{\sfa}{x}$ or $\late{\sfa}{\sfb}{x}$,  process~$B$ must know that its current basic node satisfies the required precedence condition with respect to the node~$\theta'$ at which~$A$ performs its action~$\sfa$. This is formalized as follows. 

\begin{theorem}
\label{thm:kop}
\label{THM:KOP} 
Suppose that~$C$ sends~$A$ a \Go\ message at basic node~$\sigC$ in run \mbox{$r\in\Rrep=\Rrep(\Prot,\gamma)$}, and that~$B$ performs~$\sfb$  at node~$\sigma$ in~$r$.
If~~$\Prot$ implements $\late{\sfa}{\sfb}{x}$  
then \mbox{$(\Rrep,r)\sat\Ksub{\sigma}(\Comp{\sigC}{{\scriptstyle{A}}}\atleast{x}\sigma)$}. 
Similarly, $(\Rrep,r)\!\sat\!\Ksub{\sigma}(\sigma\atleast{x}\comp{\sigC}{{\scriptstyle{A}}})$ ~if ~$\,\,\,\Prot$ implements \mbox{$\early{\sfb}{\sfa}{x}$}.
\end{theorem}

By \Cref{thm:kop}, $B$ cannot perform~$\sfb$ in a protocol solving one of the coordination tasks of \Cref{def:1} unless it knows that it is at a node satisfying an appropriate temporal precedence to the one at which~$A$ performs~$\sfa$. Since such knowledge is also a sufficient condition for~$B$'s action, we can obtain an optimal solution for the coordination tasks by characterizing when the corresponding knowledge statements hold. 
So, an optimal protocol for $B$ when performing the coordination tasks of \Cref{def:1}, is: 
\begin{protocol}
\label[protocol]{pro:1}
In local state~$\ell$, denoting $\sigma\eqdef\nodeL{B,\ell}$: If  $B$ has not performed $\sfb$ yet, and $C$ sends a \Go\ message at a basic node $\sigC\lamr\sigma$, then: 
\begin{itemize}
\item For $\late{\sfa}{\sfb}{x}$: If ${(\Rrep,r)\sat\Ksub{\sigma}(\Comp{\sigC}{{\scriptstyle{A}}}\atleast{x}\sigma)}$, then perform $\sfb$.
\item For $\early{\sfb}{\sfa}{x}$: If ${(\Rrep,r)\sat\Ksub{\sigma}(\sigma\atleast{x}\Comp{\sigC}{{\scriptstyle{A}}})}$, then perform $\sfb$.
\end{itemize}
\end{protocol}

This description of the optimal protocols is made in terms of~$B$'s knowledge about timed precedence between nodes. Our goal is to translate this into a more concrete description, in terms of the communication pattern that is recorded in~$B$'s local state. We will do so at once for both problems by solving a more general problem. I.e., we will characterize the communication patterns that determine when 
\mbox{$(\Rrep,r)\sat\Ksub{\sigma}(\theta_1\atleast{x}\theta_2)$} holds, 
for general  nodes $\theta_1$ and~$\theta_2$. 

It can be shown (and will follow from our results) that in order to know that $\te{1}\atleast{x}\te{2}$, a node~$\sigma$ must know that a zigzag pattern of weight at least~$x$ connects these two nodes.
%
%
%
%
In contrast to the case of message chains in asynchronous systems, 
information does not flow along a zigzag pattern. Indeed, it does not pass from the tail of a fork to its head, or vice-versa. 
The shape and existence of a zigzag pattern depends on whether or not the head of one fork occurs before the tail of its successor (e.g., at node~$D$ in \Cref{fig:2}).  Thus, roughly speaking, the only way in which~$\sigma$ can observe that a zigzag pattern exists is by being informed of the ordering among adjacent forks. Moreover, if~$\sigma$ does not belong to the top fork in the pattern,  then it must also be informed of the existence of this fork. 
We thus define: 

\begin{definition}[Visible Zigzag]
 \Dlabel{def:visible-zigzag}
Let $\sigma$ be a node of run~$r\in\Rrep$, and let $Z=(F_1,\ldots,F_c)$ be a zigzag pattern in $r$.
Then~$\bm{Z}$ is called \defemph{$\bm{\sigma}$-visible in} $\bm{r}$ if both (i) $\head(F_k)\lam_r\sigma$ for all $1\leq k\leq c-1$, and (ii) $\base(F_c)=\node{\sigma',p'}$ 
for a node $\sigma'\lamr\sigma$.
\end{definition}

\begin{figure}[h]
\begin{center}
  \includegraphics[width=.65\linewidth]{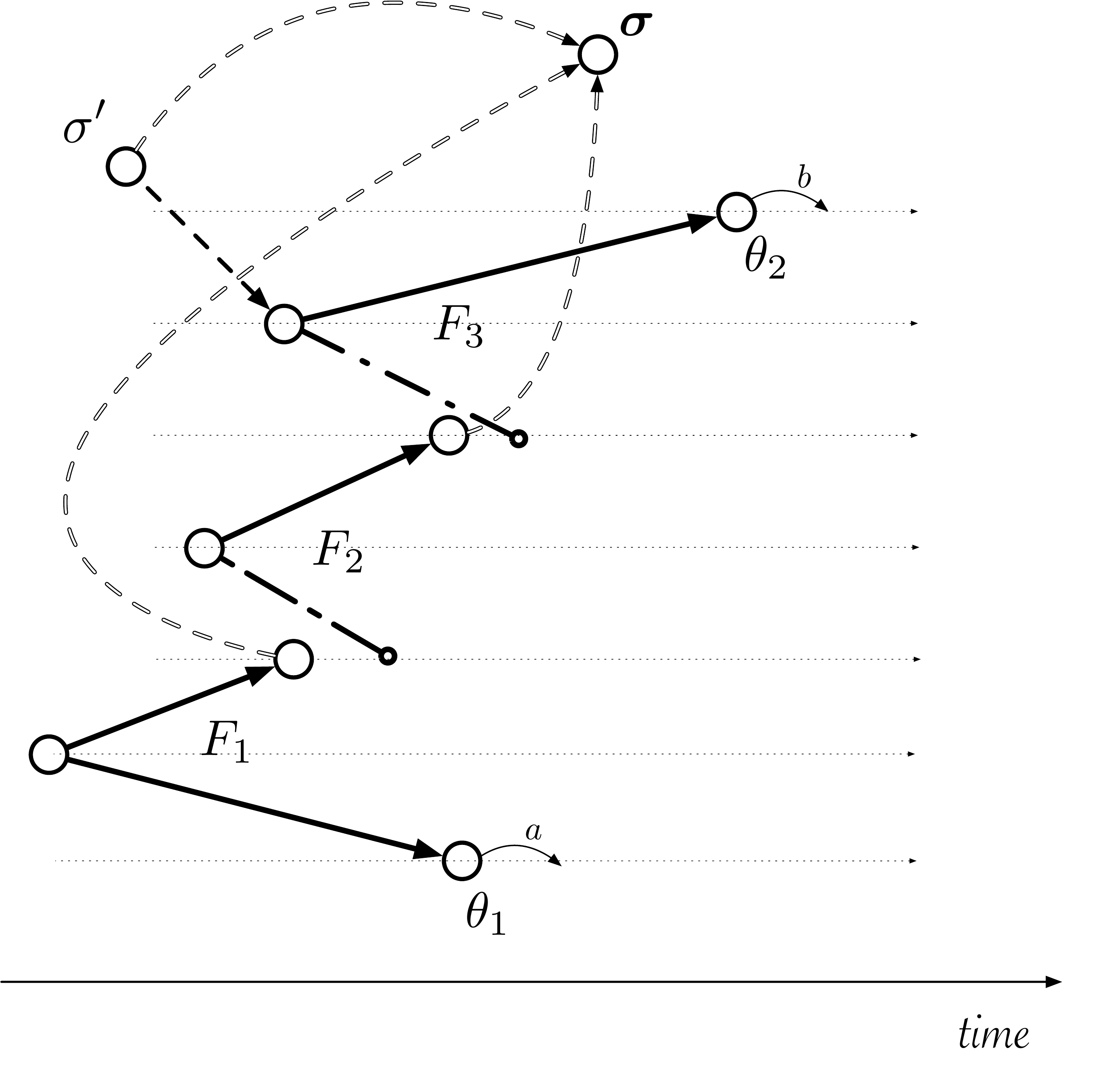}
  \caption{A $\sigma$-visible zigzag pattern from \te{1} to \te{2}} 
  \label{fig:vzz}
\end{center}
\vspace{-3mm}
\end{figure}  


In \Cref{fig:vzz} we see a $\sigma$-visible zigzag $Z=(F_1,F_2,F_3)$ from~\te{1} to~\te{2}. Note that $\head(F_1)\lam_r\sigma$ and $\head(F_2)\lam_r\sigma$, and so $\sigma$ knows that $\tail(F_2)$ doesn't appear before $\head(F_1)$, and $\tail(F_3)$ doesn't appear before $\head(F_2)$.
We remark that the definition of  a $\sigma$-visible zigzag does not require a path from the base of other forks to~$\sigma$, because for all forks except the top one, there is a path consisting of a message from the base to the head and, by \mbox{condition (i)}, a path from the fork's head to~$\sigma$. 
We can now show: 


\begin{theorem}[Visible Zigzag Theorem] 
\label{thm:vzz}
\label{THM:VZZ} 
Let~$\Rrep=\Rrep(\Prot,\gamma)$ and suppose that~$\Prot$ is an \ffip.
Moreover, let $\sigma$ be a basic node of $r\in\Rrep$, 
and let \te{1} and \te{2} be $\sigAware$ nodes in~$r$, such that both ${\timeof_r(\te{1})>0}$ and ${\timeof_r(\te{2})>0}$. 
Then \mbox{$(\Rrep,r)\models\; K_{\sigma}(\te{1}\atleast{x}\te{2})$} ~iff 
there exists a $\sigma$-visible zigzag pattern~$Z$  from~\te{1} to~\te{2} in~$r$ with 
${\weight(Z)\ge x}$. 
\end{theorem}

The Visible Zigzag Theorem 
provides a precise characterization of the pattern of communication that is necessary and sufficient for knowledge at~$\sigma$ of precedence among timepoints at distinct sites of the system. This is a fundamental aspect of information flow in \bcm\ systems. The fact that a $\sigma$-visible pattern is sufficient for such knowledge appears reasonable given our analysis so far. The main technical challenge is to prove the converse: that such a pattern is also necessary.

We can now rephrase the optimal protocol defined before (\Cref{pro:1}), in terms of concrete communication patterns:

\begin{protocol}
\label{pro:vzz}
In local state~$\ell$, denoting $\sigma\eqdef\nodeL{B,\ell}$: If  $B$ has not performed $\sfb$ yet, 
and $C$ sends a \Go\ message at a basic node $\sigC\lamr\sigma$, then:
\begin{itemize}
\item For $\late{\sfa}{\sfb}{x}$: If there is a $\sigma$-visible zigzag pattern~$Z$ in $r$ from \Comp{\sigC}{{\scriptstyle{A}}} to $\sigma$ with $\weight(Z)\ge x$, then perform $\sfb$.
\item For $\early{\sfb}{\sfa}{x}$: If there is a $\sigma$-visible zigzag pattern~$Z$ in $r$ from $\sigma$ to \Comp{\sigC}{{\scriptstyle{A}}} with $\weight(Z)\ge x$, then perform $\sfb$.
\end{itemize}
\end{protocol}
The visible zigzag patterns of~\Cref{pro:vzz} are instances of \Cref{fig:vzz}, in which one of the endpoints of the pattern is~$\sigma$ in itself. The pattern for  $\late{\sfa}{\sfb}{x}$ is illustrated in \Cref{fig:head}.  Note that there is no need for a separate message chain from~$\base(F_c)$ to~$\sigma$ in this pattern, because $\base(F_c)\lam_r\sigma=\head(F_c)$ holds the \tlf\ $F_c$, and so condition~(ii) of \Cref{def:visible-zigzag} is trivially guaranteed. In the pattern for 
the case of~$\early{\sfb}{\sfa}{x}$ we have that~$\sigma=\tail(F_1)=\theta_1$. It contains all of the message chains depicted in~\Cref{fig:vzz}.

\begin{figure}[h]
\begin{center}
  \includegraphics[width=.6\linewidth]{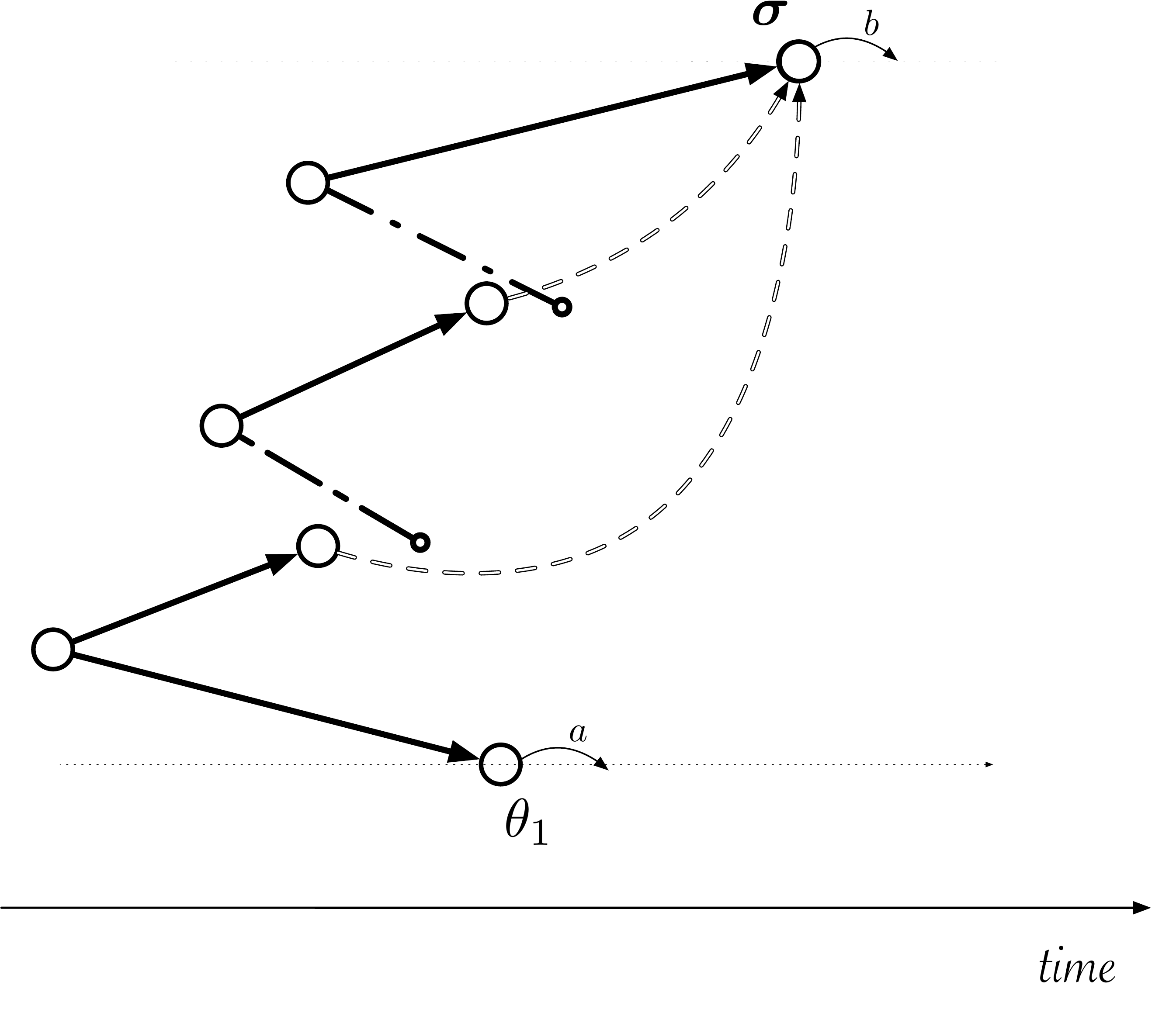}
  \caption{A visible zigzag pattern for $\late{\sfa}{\sfb}{x}$}
  \label{fig:head}
\end{center}
\vspace{-3mm}
\end{figure}  

While the conditions described above for the optimal protocol are more figurative (communication patterns), there are actually simple algorithms to check for their truth (using a structure that is described next), but this is beyond the scope of this paper.

\section{Highlights of the Analysis}
\label{sec:highlights}
In this section we 
survey the general approach used for proving our main results. Of course, the essence of the analysis has to do with extracting knowledge about timing from the actual communication in a run, given the {\em a priori} bounds on message transmission times. This has been considered in the literature, for example, in the work on clock synchronization \cite{Attiya1996,DHS,DHSS,HMM,LM,lundelius1984upper,MB,PSR,SrikT}.

\begin{figure}[h]
\begin{center}
 \includegraphics[width=.6\linewidth]{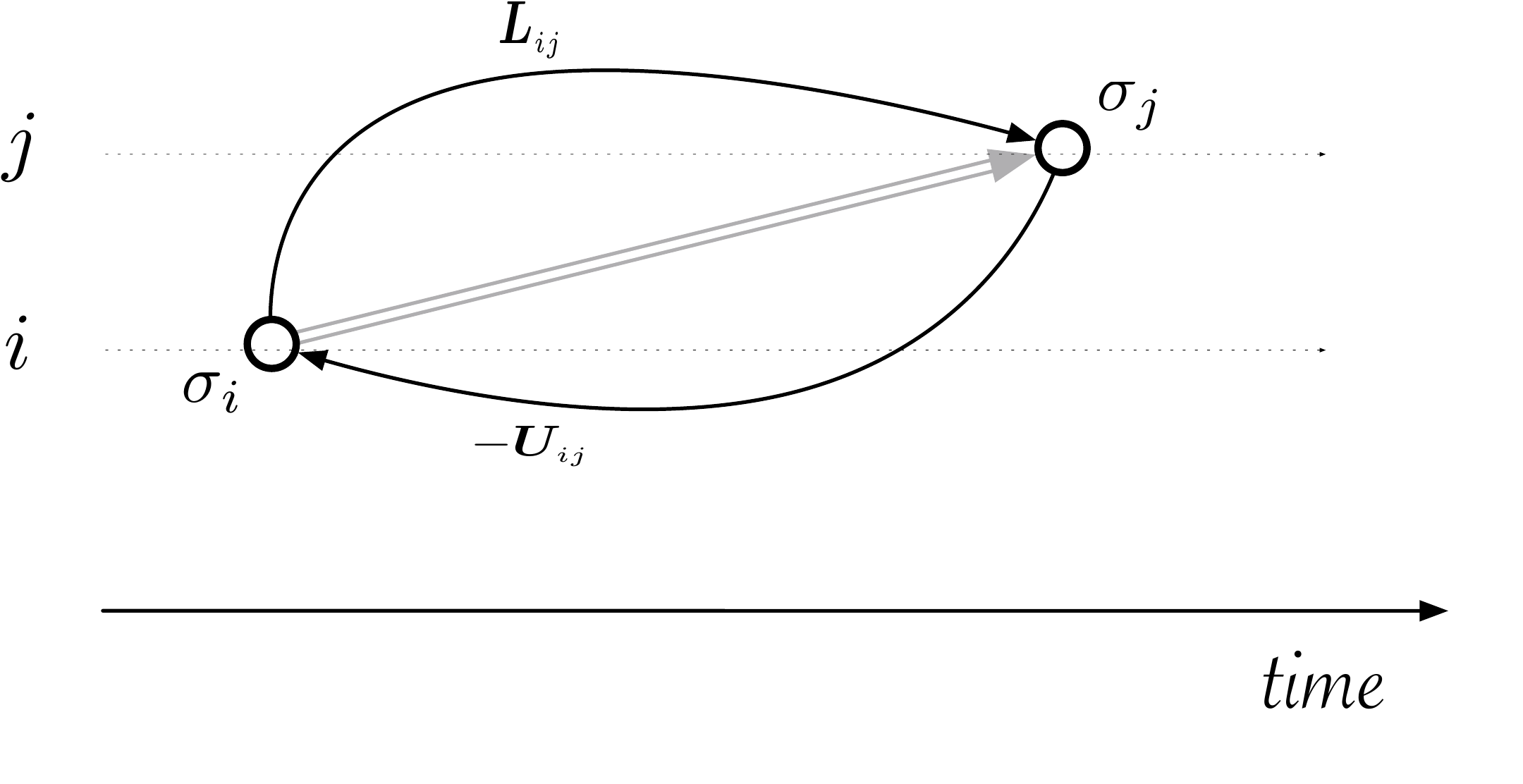}
  \caption{The bound edges created by a direct message} 
  \label{fig:Bounds2}
\end{center}
\vspace{-3mm}
\end{figure}  

Inspired by \cite{MB,PSR}, we use a weighted graph to capture the timing guarantees provided by the system, and to reason about time differences between local timepoints in a given run~$r$.
\begin{definition}
\Dlabel{def:gb}
Given a run $r\in\Rrep$, the \defemph{basic bounds graph for~$\bm{r}$} is a graph $\GB{r}=(\VB,\EB,w)$, where $\VB$ are the basic nodes that appear in~$r$. The edges of~$\EB$ are defined as follows: 
(a) If $\sigma$ and~$\sigma'$ are $i$-nodes (for the same process~$i$) and $\sigma'$ is the successor of~$\sigma$, then $(\sigma,\sigma')\in\EB$, and $w(\sigma,\sigma')=1$. 
(b) If some message sent at an $i$-node $\sigma_i$ in~$r$ is received at a $j$-node $\sigma_j$, then  both $(\sigma_i,\sigma_j)\in\EB$ and \mbox{$(\sigma_j,\sigma_i)\in\EB$}, with $w(\sigma_i,\sigma_j)=\lb{ij}$ and $w(\sigma_j,\sigma_i)=-\ub{ij}$.
\end{definition}


For an illustration of clause (b), see Figure~\ref{fig:Bounds2}.
The basic bounds graph captures timed precedence information about the temporal relation among basic nodes: (a) is justified by the fact that successive nodes are at least one time step apart, while (b) embodies the upper and lower bounds on message transmission times. \Cref{fig:Bounds2} illustrates the edges of~$\Gb$ that are induced according to case (b) by a single message delivery.
It is straightforward to check (see, e.g., \cite{MB}) that 
\begin{lemma}
\Llabel{lem:basicbounds}
Let~$p$ be a path connecting nodes~$\sigma$ and~$\sigma'$  in $\GB{r}$. 
If $w(p)=x$, then $(\Rrep,r)\models\sigma\atleast{x}\sigma'$. 
\end{lemma}
\Cref{fig:Zigbounds} highlights a path in the bounds graph that captures the timing guarantees implied by the zigzag pattern of \Cref{fig:2} via \Cref{lem:basicbounds}.

\begin{figure}[h]
\begin{center}
  \includegraphics[width=.8\linewidth]{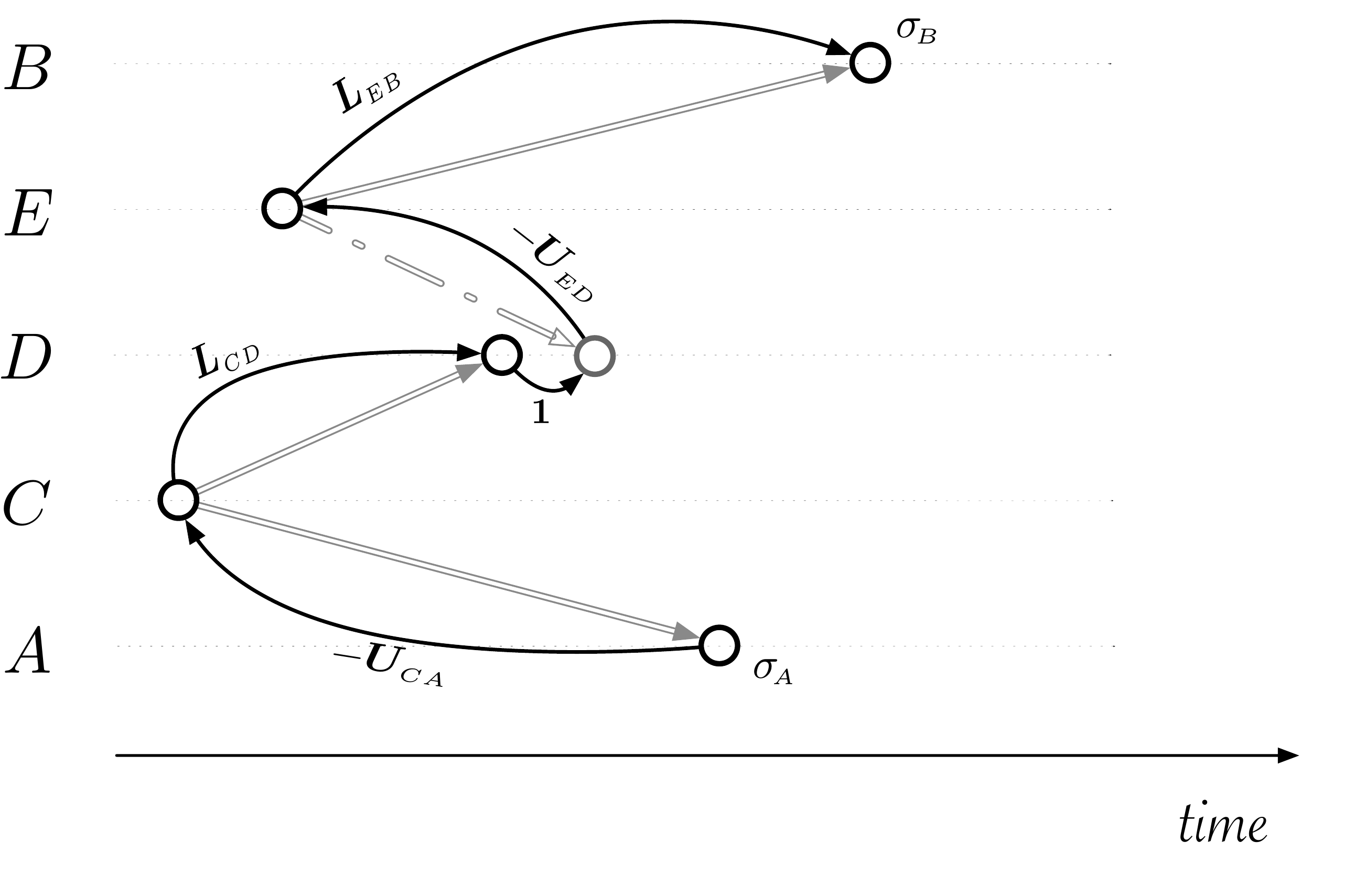} 
  \caption{A path in the bounds graph justifying \Cref{eq:bounds}\\$~$}
  \label{fig:Zigbounds}
\end{center}
\end{figure}

In addition to imposing a precedence constraint, a path in \GB{r} induces a zigzag pattern in the run~$r$.
More precisely:
\begin{lemma}
\Llabel{lem:Hzzexists}
If~$p$ is a path connecting nodes~$\sigma$ and~$\sigma'$  in $\GB{r}$, then  there exists a zigzag pattern $Z$ in $r$, from $\sigma$ to $\sigma'$, 
with \mbox{$\weight(Z)=\weight(p)$}.
\end{lemma}
The faint lines in \Cref{fig:Zigbounds} show the zigzag communication pattern underlying  the path in~$\GB{r}$, which is depicted by the solid lines.

We now give a sketch for the proof of \Cref{thm:zz}.
Assume that $\Rrep$ supports $\sigma_1\atleast{x}\sigma_2$, for two basic nodes $\sigma_1$ and $\sigma_2$ that appear in~$r$.
By \Cref{lem:basicbounds}, each path between $\sigma_1$ and $\sigma_2$ in \GB{r} defines a constraint on the difference of their times.
I.e., if $p$ is a path between $\sigma_1$ and $\sigma_2$, 
then $(\Rrep,r)\models\sigma_1\atleast{\weight(p)}\sigma_2$.
The longer the path is, the stronger the constraint. 
Thus, we are interested in finding the longest path from $\sigma_1$ to~$\sigma_2$.

Assume that $p$ is the longest path (between $\sigma_1$ and $\sigma_2$).
Our main claim is that there exists a run~$r'\in\Rrep$ such that  \mbox{$\GB{r}=\GB{r'}$}, both $\sigma_1$ and $\sigma_2$ appear in~$r'$, and 
\mbox{$\timeof_{r'}(\sigma_2)=\timeof_{r'}(\sigma_1)+\weight(p)$}.
This means that the constraint dictated by the longest path~$p$ is tight (a similar argument appears in \cite{PSR}). 
By definition of ``supports'', we obtain that $\timeof_{r'}(\sigma_2)\geq\timeof_{r'}(\sigma_1)+x$ and so $\weight(p)\geq x$.
By \Cref{lem:Hzzexists} there exists in~$r$ a zigzag pattern~$Z$ from~$\sigma_1$ to~$\sigma_2$, with $\weight(Z)=\weight(p)\geq x$, just as stated in \Cref{thm:zz}.

But what if $\GB{r}$ does not contain a path from~$\sigma_1$ to~$\sigma_2$? In such case we can show that there is a run $r''\in\Rrep$ containing~$\sigma_2$, in which~$\sigma_1$ doesn't appear. This contradicts the assumption that~$\Rrep$ supports $\sigma_1\atleast{x}\sigma_2$, since by definition of ``supports''   $\sigma_1$ and $\sigma_2$ must either both appear in~$r''$, or neither should appear. 

The proof shows that, for every node~$\sigma_2$ in~$\GB{r}$, there is a single run, $r'$, in which, intuitively, every node of~$\GB{r}$ is ``delayed'' as much as possible, relative to $\timeof_{r'}(\sigma_2)$. In other words, for every node $\sigma'$ that has a path to $\sigma_2$ in $\GB{r}$ we will have that $\timeof_{r'}(\sigma_2)=\timeof_{r'}(\sigma')+\weight(p')$, where $p'$ is the longest path from~$\sigma'$ to $\sigma_2$. 
Moreover, every node that doesn't have a path to $\sigma_2$ will not appear in $r'$.
This proves \Cref{thm:zz}.

\subsection{The Extended Bounds Graph}
The proof of \Cref{thm:vzz} is similar in its nature to the proof of \Cref{thm:zz}, but is much more complex.
While \Cref{thm:zz} states the existence of a zigzag pattern following a general run property (\textit{supports}), \Cref{thm:vzz} deals with the knowledge of a specific node.
In the previous proof we used \GB{r}.
Essentially everything that can be deduced about the timing of events in a run~$r$ based on the combined information in all processes' histories is captured by~$\GB{r}$. Figure~\ref{fig:Zigbounds}, for example, presents a path in the bounds graph that justifies the analysis leading to \Cref{eq:bounds}.
However, \GB{r} is defined by the entire run, and 
a process at a given basic node $\sigma=\node{i,\ell}$  observes only a portion of this information that is generated by the nodes in $\past(r,\sigma)$, which we denote by $\GB{r,\sigma}$.
%
This subgraph of \GB{r} does not completely capture the timing information available to~$\sigma$, however. 
For example, assume that an $i$-node $\sigma_i$ and a $j$-node $\sigma_j$ are both in $\past(r,\sigma)$, and assume that  a message sent at~$\sigma_i$ to process~$j$ isn't received at any node in $\past(r,\sigma)$.
We know that the node at which this message will be received, i.e.~$\nodeF{\sigma_i,[i,j]}$, must appear in~\Gb{r} later than $\sigma_j$.
From the upper bounds requirement, we also know that $\timeof_{r}(\sigma_i)+\upb{ij}\geq\timeof_{r}(\nodeF{\sigma_i,[i,j]})$.
Combining this with the requirement that $\timeof_{r}(\nodeF{\sigma_i,[i,j]})\geq\timeof_{r}(\sigma_j)+1$ we have that $\timeof_{r}(\sigma_i)-\timeof_{r}(\sigma_j)\geq1-\upb{ij}$, and thus $(\Rrep,r)\models\sigma_j\atleast{1-\upb{ij}}\sigma_i$.
In our setting processes follow an \ffip, and so the contents of $\past(r,\sigma)$ depend only on~$\sigma$ and not on~$r$. So this precedence holds for any run~$r'$ containing~$\sigma$. Such a run satisfies $r'\sim_\sigma r$, and we thus obtain that $(\Rrep,r)\models\Ksub{\sigma}(\sigma_j\atleast{1-\upb{ij}}\sigma_i)$.
This time precedence does not correspond to a path in \GB{r,\sigma}, 
and so \GB{r,\sigma} misses important information.

In order to fully capture the information  available to a node~$\sigma$  based on its partial view of the run, we define an extended bounds graph based on the nodes of $\past(r,\sigma)$, to which we add~$n$ auxiliary nodes $\{\psi_1,\ldots,\psi_n\}$, one per process timeline. Intuitively, each node~$\psi_j$ represents the earliest among the nodes on~$j$'s timeline at which messages will be delivered, that are beyond view (intuitively ``over the horizon'')  for~$\sigma$. This extended graph is denoted by $\bm{\GE{r,\sigma}}$. Three sets of edges $E'$, $E''$ and $E'''$ are added to the induced subgraph \GB{r,\sigma} of $\GB{r}$ to obtain the extended graph: (a) $E'$ consists of edges $(\sigma_i,\psi_i)$ from the latest $i$-node in $\past(r,\sigma)$ to~$i$'s auxiliary node, with weight~$w(\sigma_i,\psi_i)=1$; 
(b) if a message was sent from an $i$-node $\sigma_i$ in $\past(r,\sigma)$ to process~$j$ and not delivered to a node in $\past(r,\sigma)$, then an edge $(\psi_j,\sigma_i)$ is added to~$E''$  with weight $-\ub{ij}$ as in the basic bounds graph. Finally, \mbox{(c) the set~$E'''$} consists of edges $(\psi_j,\psi_i)$ with weight $-\ub{ij}$ that are added for every channel $(i,j)\in\Chan$. (Intuitively, these edges are justified by the fact that the processes follow an \ffip, and so when a message will be delivered at a node beyond the view of~$\sigma$, it will be sent to all neighbors in the $\Net$ graph.) 
 \Cref{fig:ge} illustrates the extended bounds graph $\GE{r,\sigma}$, for an $i$-node~$\sigma$. It highlights the three processes~$i$, $j$ and~$k$, and the four types of edges that appear in $\GE{r,\sigma}$. The shaded area depicts the $\past(r,\sigma)$ region.
 On the right are the auxiliary nodes $\psi_i$, $\psi_j$ and~$\psi_j$, one per process. 
 Note that the bound edges to and from auxiliary nodes handle upper bounds only. 
\begin{figure}[h]
\begin{center}
  \includegraphics[width=.9\linewidth]{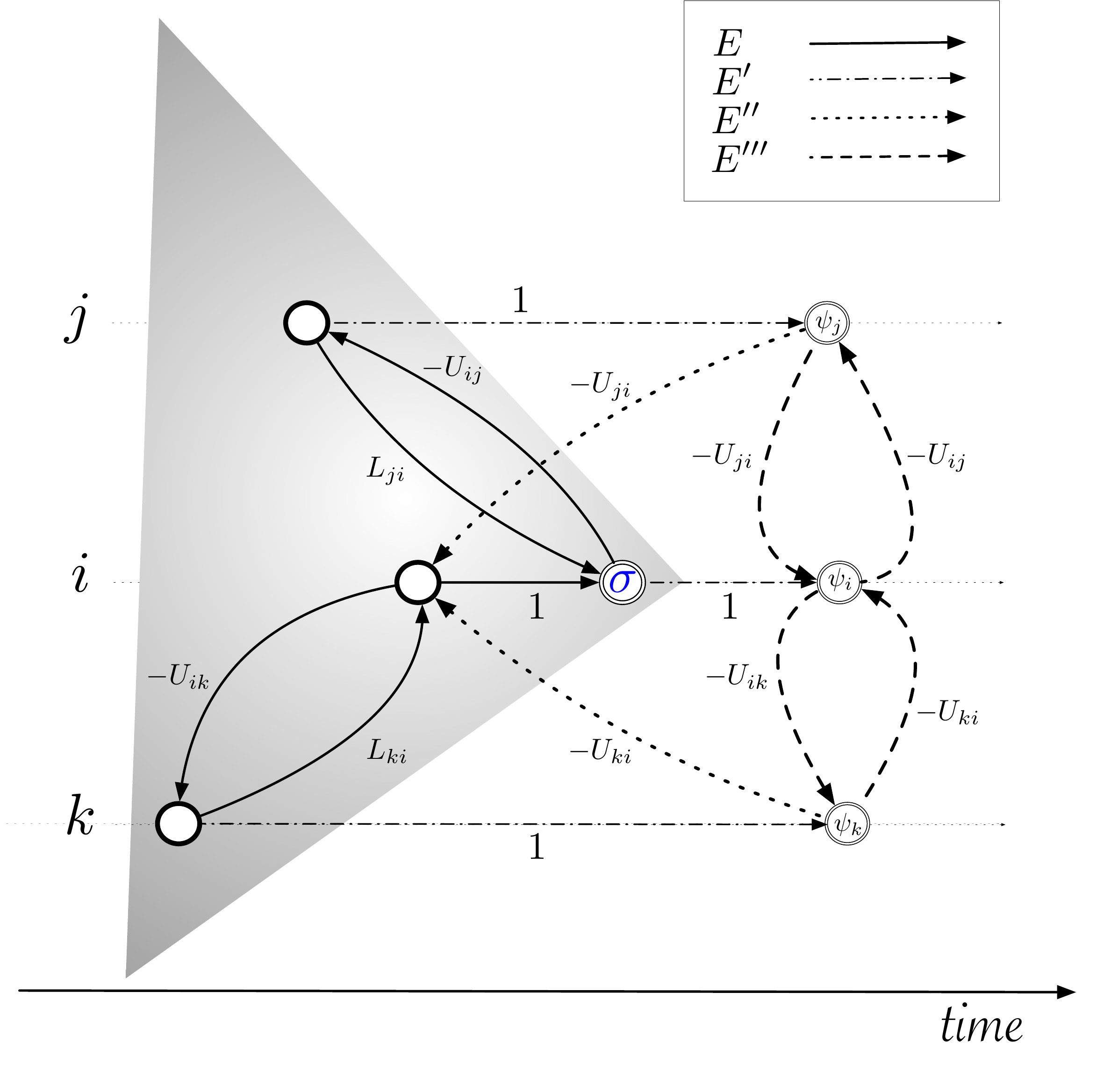}
  \caption{An illustration of the extended bounds graph~$\GE{r,\sigma}$} 
  \label{fig:ge}
\end{center}
\end{figure}  

Now, the graph \GE{r,\sigma} plays a similar role in the proof of \Cref{thm:vzz} to the role of~\GB{r} in the proof of \Cref{thm:zz}.
Indeed, \GE{r,\sigma} exhibits similar (albeit more complex) features to those of \GB{r}. 
For example, any path in $\GE{r,\sigma}$ whose  endpoints are both basic nodes from $\past(r,\sigma)$ (and not auxiliary nodes), still defines a constraint between its endpoints (in any run $r'\sim_\sigma r$). It also defines a $\sigma$-visible zigzag in $r$ with the same weight.
Note the small differences: (1) The constraint here holds in any run $r'\sim_\sigma r$ (as for any such run, $\GE{r,\sigma}=\GE{r',\sigma}$), instead of any run $r'$ with the same complete bounds graph (i.e $\GB{r'}=\GB{r}$), and (2) the zigzag pattern is a {\em $\sigma$-visible} zigzag.
Paths that start at, or end in, or auxiliary nodes also exhibit important features, which are essential for the proof.
The full details are beyond the scope of our presentation here and are available in \proofs.

Crucially, the extended bounds graph, \GE{r,\sigma}, can be used to construct valid runs of~$\Rrep$ with desirable properties. This is based on a careful assignment of times to nodes of $\GE{r,\sigma}$, in the following manner:
\begin{definition}
Let~$\sigma$ be a basic node appearing in~$r\in\Rrep$. 
 A valid timing function for $\GE{r,\sigma}=(\VE,\EE,w)$ is a function $T:\VE\to\Nat$, such that $T(\theta_1)+w(\theta_1,\theta_2)\leq T(\theta_2)$ holds for each $(\theta_1,\theta_2)\in\EE$. 
\end{definition}

Based on a valid timing function~$T$ for $\GE{r,\sigma}=(\VE,\EE,w)$, we can define a run~$r'\sim_\sigma r$ of~$\Rrep$ in which the nodes of \mbox{$\past(r,\sigma)$}  appear at the prescribed times, and all the other nodes appear no earlier than the time of the auxiliary node~$\psi_j$ that belongs to their timeline~$j$. 
This result is achieved by the fact that the bounds associated with auxiliary nodes make sure that nodes outside  $\past(r,\sigma)$ won't appear too early relative to nodes from $\past(r,\sigma)$. (That, in turn, could force a message sent outside $\past(r,\sigma)$ to be received inside $\past(r,\sigma)$, which would modify $\sigma$'s past and cause $r'\not\sim_\sigma r$).

\section{Discussion}
\label{sec:conclusions}
\label{sec:discussion}

%
The principles underlying coordination in purely asynchronous systems are by now fairly well understood, based on~\cite{Lamclocks} and the four decades since it was published. Message chains play a central role in determining the ordering of events and coordinating their timing. More recently, the study of coordination in systems with global clocks was initiated by~\cite{BzMacm}. The current paper considers yet another timing model, the \bcm\ model, in which there are no built-in timers and clock. Nevertheless, timing information can be gleaned from observed events, because there are upper and lower bounds on the  message transmission times among processes. A direct use of bounds in such a model is the one illustrated in \Cref{fig:1}: Given two message chains that start from the same point, if the sum of lower bounds on one is greater than the sum of upper bounds on the other, then the first message chain is guaranteed to end later than the second one. Indeed, the bounds can be used to provide a quantitative estimate of the time difference between these two events. 
We introduced the the notion of a zigzag message pattern and showed that it provides another way to deduce the time precedence between events. The existence of an appropriate zigzag pattern was shown to be necessary and sufficient for the message pattern of an execution of the system to ensure that a given timed precedence among events is satisfied. 

Interestingly, whereas it is possible to ensure that the existence of a message chain will be observed by the process receiving the chain's final message, this is not the case with general zigzag patterns. Information about the pattern's existence is distributed among the processes. In order to use a zigzag pattern in coordination, it is necessary for its relevant endpoints to obtain information about the order in which pivotal intermediate messages were delivered. Only then can a process know that the pattern exists, and hence to know that the precedence that the zigzag pattern implies is satisfied. Our analysis provides a characterization of when a precedence statement is known by a process at a given local state. This requires a {\em visible zigzag}, consisting of an appropriate zigzag pattern, as well as message chains informing the node about the pivotal parts of the zigzag pattern. A corollary of this is a characterization of patterns that allow coordinating actions according to $\Early$ and $\Late$ specifications. 

The main mathematical structures underlying our analysis are the basic bounds graph and the extended bounds graphs presented in \Cref{sec:highlights}. In these, the start and end points of events are nodes, and the  bounds are represented by weighted edges among these nodes. While the basic bounds graph has appeared in the analysis of clock synchronization (see, e.g., \cite{PSR}, in which it is used to capture synchronization even in the presence of clock drift), the extended bounds graph seems to be novel. It allows an analysis of the timing information at a node based on its subjective view of the computation. Events in its direct causal past, as well as the fact that events do not appear there, provide information on the timing and ordering of events. 

As remarked in the Introduction, the~\bcm\  model can easily be adapted to capture bounds on the duration of other events as well. A natural setting that fits the \bcm\ model is that of asynchronous, or self-timed, VLSI circuits, which are circuits that operate without clocks. In such settings, time bounds are often used to coordinate actions and ensure correctness of the computation. The typical way to do so is by using a simple fork as in \Cref{fig:1}. Such forks are also the basis for correct operation in synchronous circuits, where extreme care is taken to ensure that clock inputs to different flip-flops are arranged to have very similar delays from a common source for the implementation of sequencing~\cite{WesteHarris}. To the best of our knowledge, it is an open problem whether zigzag causality and our characterization of solutions to the $\Early$ and $\Late$ coordination problems may facilitate the design of new circuits.

%

\bibliographystyle{abbrv}
\bibliography{tal,z,z2}

\appendix

\section{Proving \Cref{thm:kop}}
%
The proof of \Cref{thm:kop} depends on two properties, which we now state formally  and prove in \Cref{lem:past,lem:sim}.
The first captures the fact that, intuitively, in an fip, if $\sigma'\lamr\sigma$ then the fact that~$\sigma'$ happens before~$\sigma$ is known at~$\sigma$. Namely, it will be true at every run~$r'$ that is indistinguishable to~$r$ in the eyes of~$\sigma$.

\begin{lemma}
\Llabel{lem:sim}
Let~$\Prot$ be an fip and let  $r,r'\in\Rrep(\Prot,\gamma)$. 
If $\sigma'\lamr\sigma$ and $r'\sim_\sigma r$, then $\sigma'\lam_{r'}\sigma$.
\end{lemma}
\begin{proof}
By the definition of $\lamr$, since $\sigma'\lamr\sigma$ there is a sequence of basic nodes in $r$, $(\sigma_0,\sigma_1,\ldots,\sigma_k)$, where $\sigma_0=\sigma'$, $\sigma_k=\sigma$ and $\sigma_m$, for $0<m\leq k$, is either the successor of $\sigma_{m-1}$, or receives a message sent from $\sigma_{m-1}$ (in run~$r$).
We claim that if $\sigma_m$ appears in $r'$, then so does $\sigma_{m-1}$, as in both cases the state of~$\sigma_m$ is affected by the state of~$\sigma_{m-1}$, being either its predecessor or the source of a message it receives, thus $\sigma_m$ cannot appear without $\sigma_{m-1}$, and further $\sigma_{m-1}\lam_{r'}\sigma_m$.
By induction, as $\sigma_k=\sigma$ appears in $r'$, then $\sigma_0=\sigma'$ appears in $r'$ as well and $\sigma'\lam_{r'}\sigma$.
\end{proof}

The second lemma states that in an implementation of $\late{\sfa}{\sfb}{x}$ or $\early{\sfb}{\sfa}{x}$, process~$B$ can perform~$\sfb$ only if there is a message chain to~$B$ from a node at which~$C$ sends a \Go\ message to~$A$.

\begin{lemma}
\Llabel{lem:past}
Let~$\Rrep=\Rrep(\Prot,\gamma)$ and suppose that~$\Prot$ implements either $\late{\sfa}{\sfb}{x}$ or $\early{\sfb}{\sfa}{x}$.
If $B$ performs $\sfb$ at basic node~$\sigma$ in a run $r\in\Rrep$, then there is a basic node~$\sigC$ in $r$ at which $C$ sends a \Go\ message to~$A$, such that $\sigC\lamr\sigma$.
\end{lemma}
\begin{proof}
Let~$\Rrep=\Rrep(\Prot,\gamma)$ and suppose that~$\Prot$ implements either $\late{\sfa}{\sfb}{x}$ or $\early{\sfb}{\sfa}{x}$.
Fix a run $r\in\Rrep$, and a basic node~$\sigma$ that appears in $r$, and such that $B$ performs $\sfb$ at $\sigma$.
We want to prove that there is a node~$\sigC$ of~$C$ in~$r$ that sends a \Go\ message, such that $\sigC\lamr\sigma$.

We construct a new run, $r'$, by inductively describing its global states.
At $m=0$, define $r'(0)=r(0)$. 
For $m>0$, assume we have the global states of $r'$ up to time $m-1$, and we define it for time $m$.
We do so by defining which messages will be delivered.
Let $\mu$ be a message sent at the $i$-node $\sigma_i$ to agent~$j$, that is in transit in $r'$ before time $m$.
\begin{itemize}
\item If $\sigma_i$ appears in $r$, $\timeof_{r'}(\sigma_i)=\timeof_{r}(\sigma_i)$, $\mu$ is delivered in $r$ at $m$ and $\nodeL{j,r_j(m)}\lamr\sigma$, then deliver~$\mu$ to~$j$.
\item Otherwise, deliver $\mu$ only if $\timeof_{r'}(\sigma_i)+\upb{ij}=m$.
\end{itemize}
Moreover, for any $i\in\Proc$, deliver in $r'$ an external message to agent $i$ at time $m$ if and only if the same external message is received by~$i$ at $m$ in $r$, and $\nodeL{i,r_i(m)}\lamr\sigma$.\\

Note that $r'\in\Rrep$.
Moreover, we claim that the following holds for $r'$:
\begin{enumerate}
\item $r'\sim_\sigma r$; and
\item The only external messages delivered in $r'$ are delivered to nodes from $\past(r,\sigma)$.
\end{enumerate}
Following the above claim, as $r'\sim_\sigma r$ , we have that $\sigma$ also appears in $r'$, and thus $B$ performs $\sfb$ in $r'$ as well.
As $\Prot$ implements either $\late{\sfa}{\sfb}{x}$ or $\early{\sfb}{\sfa}{x}$, and as $B$ performs $\sfb$ in $r'$, then there is in~$r'$ a node of $C$, $\sigC$, which sends a \Go\ message, and as it must do so spontaneously, he must have received an external message.
But, by the above claim the only nodes that receive external messages are from $\past(r,\sigma)$.
Thus, $\sigC\in\past(r,\sigma)$, and $\sigC\lamr\sigma$ as required.\\

To prove the above claims, we prove by induction on the time $m$ that for all $i\in\Proc$, if ${\nodeL{i,r_i(m)}\lamr\sigma}$, then $r'_i(m)=r_i(m)$.
At $m=0$, this is trivially true as $r'(0)=r(0)$.
Let $m>0$, and assume the claim is true up to time $m-1$.
Note that by the induction we can conclude that if $\sigma'$ is a basic node that appears in $r$, such that $\sigma'\lamr\sigma$ and $\timeof_r(\sigma')<m$, then $\sigma'$ also appears in $r'$, and $\timeof_{r'}(\sigma')=\timeof_r(\sigma')$.

Assume that $\nodeL{i,r_i(m)}\lamr\sigma$ for some $i\in\Proc$.
Clearly $\nodeL{i,r_i(m-1)}\lamr\sigma$, thus $r'_i(m-1)=r_i(m-1)$.
We claim that exactly the same messages are delivered to~$i$ at~$m$ in~$r$ and~$r'$, and so $r'_i(m)=r_i(m)$.
Let $\mu$ be a message that is sent in $r$ from agent $j$ at time $t_\mu$, and is delivered in~$r$ to~$i$ at~$m$.
As $\nodeL{j,r_j(t_\mu)}\lamr\nodeL{i,r_i(m)}\lamr\sigma$, then $\nodeL{j,r_j(t_\mu)}\lamr\sigma$, and as $t_\mu<m$, we have that \nodeL{j,r_j(t_\mu)} appears also in $r'$, at the same time as in $r$, and so $\mu$ is sent in $r'$ at the same time.
According to the definition of $r'$, $\mu$ cannot be delivered in $r'$ before time $m$, and it will be delivered at $m$.
Now, let $\mu$ be a message that is sent in $r'$ from agent $j$ at time $t_\mu$, and is delivered to~$i$ at~$m$, and assume that $\mu$ isn't delivered in $r$ at $m$.
In this case it must be that $t_\mu+\upb{ji}=m$. Denote $\sigma_j=\nodeL{j,r_j(t_\mu)}$.
If $\sigma_j\lamr\sigma$, then $r_j(t_\mu)=r'_j(t_\mu)$, and then $\mu$ is sent in $r$ as well, at the same time. As it is not delivered in $r$ at $m$, and as $t_\mu+\upb{ji}=m$, then it must be delivered before time $m$. But then it will be delivered at the same time in $r'$, in contradiction to it being in transit in $r'$ before time $m$.
So, it must be that $\sigma_j\not\lamr\sigma$. Note that $\timeof_r(\sigma_j)\leq t_\mu$, and so it sends to $i$ a message in $r$ that must be delivered no later than time $m$, so it must be that $\sigma_j\lamr\nodeL{i,r_i(m)}$.
But, as $\nodeL{i,r_i(m)}\lamr\sigma$, we get that $\sigma_j\lamr\sigma$ in contradiction with the last assumption, and thus proving the induction.\\

The previous two claims follow directly from the claim we have just proved, and the truth of \Cref{lem:past} follows.	
\end{proof}
We shall now prove \Cref{thm:kop}:
\begin{proof}
We prove the first claim; the proof of the second claim is analogous. 
Let~$\Rrep=\Rrep(\Prot,\gamma)$ and suppose that~$\Prot$ implements $\late{\sfa}{\sfb}{x}$.
Fix~$r\in\Rrep$ and assume that~$C$ sends a \Go\ message to~$A$ at basic node~$\sigC$ in run $r$, and that $B$ performs $\sfb$ at basic node~$\sigma$.
Let $r'\sim_\sigma r$ be a run of~$\Rrep$. 
By the definition of indistinguishability, $\sigma$ appears in $r'$.
Since $\Prot$ is a deterministic protocol, $B$ also performs~$\sfb$ at~$\sigma$ in $r'$. 
By \Cref{lem:past}, 
we have that $\sigC\lamr\sigma$.
By \Cref{lem:sim}, the fact that~$\Prot$ is a full-information protocol implies that 
every basic node~$\sigma'$ that satisfies $\sigma'\lamr\sigma$ will appear in $r'\sim_\sigma r$ as well.
Thus, in particular, $\sigC$ appears in $r'$ as well, and  a \Go\ message is sent to~$A$ at $\sigma_C$ in~$r'$.
Since $\Prot$ implements $\late{\sfa}{\sfb}{x}$, we obtain both that $A$ performs~$\sfa$ at $ \comp{\sigC}{\scriptstyle{A}}$ in~$r'$, and that 
$(R,r')\sat \comp{\!\sigC}{{\!\scriptstyle{A}}}\atleast{x}\sigma$. The claim follows. 
%
%
\end{proof}

\section{The Proof of \Cref{thm:zz}}
The \textit{basic bounds graph}, \GB{r}, defined in \Cref{def:gb}, describes the time constraints on the basic nodes that appear in run $r$. 
These are imposed by the context~$\gamma$, based on the events that occur in the run.  
We have three types of such constraints: For every message sent, the context imposes an upper bound and a lower bound on its transmission times. Moreover, events that occur at a given site are linearly ordered. (Indeed, we assume that they are separated by at least one time unit.) 
Each constraint is represented in $\GB{r}$ by a corresponding weighted edge.

While each edge in the basic bounds graph describes a time constraint between two adjacent nodes, by transitivity we have that each \textbf{path} in the basic bounds graph describes a time constraint between the source and destination nodes.
Note that there are no positive cycles in  a basic bounds graph, since such a cycle would impose the absurd constraint that a node must occur strictly later than when it does. 

%
The basic bounds graph and the zigzag pattern are closely related. 
The following lemma
 shows that every path in the bounds graph defines a zigzag pattern  whose weight equals the path's length. 
\begin{lemma}
\Llabel{lem:zzisgb}
Let $r\in\Rrep$, $x\in\Int$, let $\sigma_1$ and $\sigma_2$ be basic nodes in $r$, and let \te{1} and \te{2} be two nodes of $r$ such that $\basic(\te{1},r)=\sigma_1$ and $\basic(\te{2},r)=\sigma_2$.
For every path $P$ from $\sigma_1$ to $\sigma_2$ in \GB{r} with weight $x$, 
there is a zigzag pattern~$Z$  from \te{1} to \te{2} in~$r$ with $\weight(Z)=x$.
\end{lemma}
\begin{proof}
Let $x\in\Int$,  let $\basic(\te{1},r)=\sigma_1$ and $\basic(\te{2},r)=\sigma_2$, and assume there is a path $P$  from $\sigma_1$ to $\sigma_2$ in \GB{r} with weight $x$.
We prove by induction on $m$, the number of edges in the path $P$, that 
 there is a zigzag pattern $Z$ in $r$ from \te{1} to \te{2} with $\weight(Z)=x$.
If $m=0$, then $\basic(\te{1},r)=\basic(\te{2},r)$, and we define 
two \tlf s, $F_0\eqdef(\te{1},\te{1},\te{1})$ and $F_1\eqdef(\te{2},\te{2},\te{2})$ and a zigzag $Z\eqdef(F_0,F_1)$ and we are done (note that $F_0$ and $F_1$ are joined, and so $\weight(Z)=0$).
Let $m>0$, and assume the claim is true for any path with $m-1$ edges.
Assume that $\sigma_1$ is followed in $P$ by $\sigma'_1=\nodeL{k,l'}$.
Looking at the suffix of $P$ starting from $\sigma'_1$, we have by the assumption 
a zigzag pattern~$Z'$ in $r$ from the node \nodeF{\sigma'_1,k} corresponding to basic node~$\sigma'_1$ to \te{2}, say 
with $\weight(Z')=x-w(\sigma_1,\sigma'_1)$. (Here $w(\sigma_1,\sigma'_1)$ is the weight of the edge $(\sigma_1,\sigma'_1)$ in \GB{r}.)
Denoting $Z'=(F_1,\ldots,F_c)$, note that $\timeof_r(\tail(F_1))=\timeof_r(\sigma'_1)$. 
We shall construct the desired zigzag pattern~$Z$, whose tail will be the node~$\sigma_1$, which will be connected to~$Z'$ in a manner that depends on the edge that connects~$\sigma_1$ to the tail~$\sigma'_1$ of~$Z'$. We distinguish three cases: 
\begin{enumerate}
\item If $(\sigma_1,\sigma_1')$ is an edge with weight of $\lb{ik}$, corresponding to a message sent from $\sigma_1$ to $\sigma'_1$ in~$r$, then we define $Z=(F_0,F_1,\ldots,F_c)$, where $F_0\eqdef(\te{1},\te{1}\concat k,\te{1})$.  Note that $\basic(\te{1}\concat k,r)=\sigma'_1$ corresponds to $\tail(F_1)$, and so in this case $F_0$ and $F_1$ are joined. 
By construction, $\weight(F_0)=w(\sigma_1,\sigma'_1)$ and $\weight(Z)=\weight(F_0)+\weight(Z')$. Thus, $\weight(Z)=x$, as desired.
\item  If $(\sigma_1,\sigma_1')$ is an edge with weight of $-\ub{ki}$, corresponding to a message sent from $\sigma'_1$ to $\sigma_1$ in~$r$, then we define $Z=(F_0,F_1',\ldots,F_c)$, where $F_0=(\te{1},\te{1},\te{1})$ and $F'_1=(\base(F_1),\head(F_1),\tail(F_1)\concat i)$. 
In this case $F_0$ and $F'_1$ are joined. By construction, $\weight(F_0)=0$ and $\weight(F'_1)=\weight(F_1)+w(\sigma_1,\sigma'_1)$. Thus, again, $\weight(Z)=\weight(Z')+w(\sigma_1,\sigma'_1)=x$, as desired.
\item  If $(\sigma_1,\sigma_1')$ is an edge with weight of~$1$, as $i=k$ and $\sigma'_1$ is the successor of $\sigma_1$, 
then we define $Z=(F_0,F_1,\ldots,F_c)$ where $F_0=(\te{1},\te{1},\te{1})$. In this case the forks~$F_0$ and~$F_1$ are not joined. By construction, $\weight(F_0)=0$ and since $F_0$ and~$F_1$ are not joined we obtain that $\weight(Z)=\weight(Z')+1=x$, and we are done. 
\end{enumerate}
In all cases above, $Z$ is a zigzag pattern in $r$ from $\te{1}$ to $\te{2}$ with $\weight(Z)=x$ as required.
\end{proof}
We remark that the above lemma can be extended to be an if and only if statement (with some small technical changes).\\

Observing \GB{r}, we have the set of basic nodes that appear in $r$.
In particular, \GB{r} contains the dynamics between the nodes, that created their actual local states.
The only thing missing in \GB{r}, are the times at which those basic nodes occur (though these times are well defined in $r$, observing \GB{r} without knowing that it came from $r$, we can't know the times).
In fact, there are different runs whose basic bounds graphs are completely identical.
The difference between two such runs is just in the time in which the nodes of \GB{r} appear.
Moreover, given the graph \GB{r}, we can construct actual runs that produce such graph ($r$, for example).
For that task we have to assign times to the nodes of \GB{r}, i.e. the times at which those nodes will appear in the run.
However, not any such assignment of times is legal.  
If we want to keep the local states as they are, then the order of the nodes of each agent must remain the same as in \GB{r}, and all the respective messages must be sent and received by the same nodes. In such case, we must assure that the messages transmission time doesn't exceed bounds.
More formally, we define a \textit{valid timing function} of a basic bounds graph:
\begin{definition}
\Dlabel{def:vtf}
Let $r\in\Rrep$, and let $V'$ be a subset of the nodes of \GB{r}. A \defemph{valid timing function} for~$V'$ w.r.t.\ \GB{r} is a function $T:V'\to\Nat$ 
such that $T(\sigma_1)+w(\sigma_1,\sigma_2)\leq T(\sigma_2)$ holds for each edge $(\sigma_1,\sigma_2)$ in \GB{r} with $\sigma_1,\sigma_2\in V'$.
\end{definition}

As we will see, a valid timing function~$T$ defined for the set~$V$ of all nodes of \GB{r} can be used to produce a legal run~$r[T]$ such that $\GB{r[T]}=\GB{r}$, in which the timing of the nodes is according to~$T$.

Let $\sigma$ and $\sigma'$ be basic nodes in $r$.
Recall that any path in \GB{r} from $\sigma'$ to $\sigma$ defines a constraint on how far $\sigma'$ can appear before $\sigma$ in $r$ (this is proved here indirectly by combining \Cref{thm:zzsimple} and \Cref{lem:zzisgb}). 
Equivalently, such a path constrains the time difference between~$\sigma'$ and~$\sigma$ in any other run $r'\in\Rrep$ for which $\GB{r'}=\GB{r}$.
The longer (i.e., ``heavier'') the path from $\sigma'$ to~$\sigma$ is, the stronger the constraint. 
We will prove that the longest such path gives a tight constraint.
I.e., if the longest path from~$\sigma'$ to~$\sigma$ in \GB{r} is, say, $P$, then there is a run $r'$ where $\timeof_{r'}(\sigma)-\timeof_{r'}(\sigma')=\weight(P)$.
We do so by proving the existence of a valid timing function $T$ such that $T(\sigma)-T(\sigma')=\weight(P)$, 
and then proceed to show that there is such run $r[T]$
 for which $\GB{r[T]}=\GB{r}$, 
and where the timing of the nodes are according to $T$.
We will actually prove something stronger, which is that there exists a valid timing function $T'$ in which for \defemph{all} nodes~$\sigma'$ that are connected to~$\sigma$ by a path in \GB{r}, 
$T'(\sigma)-T'(\sigma')=\weight(P')$ holds, where $P'$ is the longest path between $\sigma'$ and $\sigma$ in \GB{r}. 
Furthermore, there is a corresponding run $r[T']$ where the time of the nodes are according to $T'$.

By the definition of \textit{supports}, if \Rrep\ supports $\sigma'\atleast{x}\sigma$, then $\timeof_{r'}(\sigma)-\timeof_{r'}(\sigma')\geq x$ holds for every run $r'\in\Rrep$ that contains both nodes~$\sigma$ and~$\sigma'$.
In $r[T']$ we have that $\timeof_{r[T']}(\sigma)-\timeof_{r[T']}(\sigma')=\weight(P)$ (where $P$ is the longest path in \GB{r} between $\sigma'$ and $\sigma$) and so it must be that $\weight(P)\geq x$.
Observing the path~$P$, we have by \Cref{lem:zzisgb} a zigzag pattern $Z$ in $r$ from $\sigma'$ to $\sigma$ with $\weight(Z)=\weight(P)\geq x$.
this result proves \Cref{thm:zz} (at least for basic nodes). 

There is however another complexity, as not all of the basic nodes that appear in $r$ have a path to $\sigma$ in \GB{r}. 
For this reason, we will have to use a slightly stronger result than the above-mentioned claims. 
We show that even if $T$ is a valid timing function only on a \textbf{subset} of nodes (of \GB{r}), then yet a run $r[T]$  that contains only that subset of nodes (and some other initial nodes, i.e. nodes from time $0$), that appear at times according to $T$, can be constructed. This will be true, however, only for subsets of nodes that satisfy a particular closure property: 
\begin{definition}  
Let $r\in\Rrep$, let $V$ be the set of nodes of \GB{r}, and let $V'\subseteq V$ be a subset of $V$.
We say that~$V'$ is \defemph{precedence-closed} w.r.t. $\GB{r}$ (or p-closed for short) if for every edge ($\sigma_1,\sigma_2$) of \GB{r}, if 
$\sigma_2\in V'$, then $\sigma_1\in V'$ as well. 
\end{definition}

Intuitively, the fact that~$V'$ is p-closed means that the constraints imposed by $\GB{r}$ based on nodes 
outside of~$V'$ do not affect the timing of nodes in~$V'$.
We will be interested in a particular p-closed set of nodes of \GB{r}: The nodes from which there is a path to a specific  node~$\sigma$: 
\begin{definition}  
Let $r\in\Rrep$, let $\GB{r}=(V,E)$ and let $\sigma\in V$. 
We define the \defemph{$\bm{\sigma}$-precedence set} in \GB{r}, denoted~$V_\sigma$, by: 
$V_\sigma\eqdef\{\sigma'\in V : $ there is a path in \GB{r} from $\sigma'$ to $\sigma\}$.
\end{definition}

We remark that $V_\sigma$ is the minimal set of nodes that contains~$\sigma$ and is p-closed (as shall be proved next).
Recall that there are three types of edges in \GB{r}: (i) Between two successive nodes of the same agent, (ii) from a node at which a message is sent to the node at which it is received, and in the opposite direction (iii) from the node at which a message is received to the node at which it was sent.
Thus, if a message sent at an $i$-node~$\sigma$ is received at a $j$-node $\sigma'$, then $\sigma'\in V_\sigma$, as there is an edge (with weight $-\ub{ij}$) that goes ``back'' from $\sigma'$ to $\sigma$ (see \Cref{fig:Bounds2}).
If the protocol is an \ffip, and $(j,i)\in\Chan$, then the protocol ensures that~$j$ also sends a message at~$\sigma'$ back to agent $i$, which is received at the node $\sigma''=\basic(\nodeF{\sigma', [j,i]},r)$.
As there is an edge from $\sigma''$ to $\sigma'$ (the edge in the opposite direction to the send-receive), and there is an edge from $\sigma'$ to $\sigma$, we get that also $\sigma''\in V_\sigma$.
Clearly (in an \ffip) process~$i$ will also send a message at~$\sigma''$ back to agent $j$, which is received at a node that will be also in $V_\sigma$, and so on.
Consequently, $V_\sigma$ is typically infinite, and contains nodes that appear in the far future of~$\sigma$ (at least in the case of an \ffip).
Intuitively, the timing of each one of these nodes imposes a constraint on the timing of~$\sigma$. 
On the other hand, for example,  the successor of $\sigma$ on~$i$'s timeline might not appear in~$V_\sigma$, in which case it might be delayed arbitrarily relatively to $\sigma$.

We first formally prove that $V_\sigma$ is in fact a p-closed set:
\begin{lemma} 
\Llabel{lem:closure}
Let $r\in\Rrep$, $\sigma$ be a basic node in $r$ and let $V$ be the set of nodes of \GB{r}.
Then $V_\sigma$ is p-closed.
\end{lemma}
\begin{proof}
Let $\sigma'\in V_\sigma$, and let $\sigma''\in V$ such that $(\sigma'',\sigma')$ is an edge in \GB{r}.
$\sigma'\in V_\sigma$, and so there is a path from $\sigma'$ to $\sigma$ in \GB{r}.
As $(\sigma'',\sigma')$ is an edge in \GB{r}, we can go in \GB{r} from $\sigma''$ to $\sigma'$, and then from the path that $\sigma'$ has to $\sigma$, and so to reach from $\sigma''$ to $\sigma$. 
Thus, there is a path in \GB{r} from $\sigma''$ to $\sigma$, and so $\sigma''\in V_\sigma$. It follows that~$V_\sigma$ is p-closed.
\end{proof}

As we used in \Cref{thm:zz} the general nodes notation (instead of using basic nodes), we shall define \mbox{$V_\tet\eqdef V_{\basic(\tet,r)}$}.
For a node~$\theta$ such that $\basic(\tet,r)=\sigma$, we now define the promised valid timing function on $V_\tet$ ($=V_\sigma$) in which $T(\sigma)-T(\sigma')=d(\sigma')$ holds for every node $\sigma'\in V_\sigma$, where $d(\sigma')$ is the weight of the longest path in \GB{r} from $\sigma'$ to $\sigma$.
%
\begin{definition} [Slow-Timing] 
\Dlabel{def:slow}
Let $r\in\Rrep$, and let \tet\ be a node of $r$ so that $\basic(\tet,r)=\sigma$. 
Moreover, let~$D$ denote the weight of the longest path in \GB{r} that ends in the node~$\sigma$. 
The \textit{slow timing function} of $\tet$ in $r$, $\Tst:V_\tet\to\Nat$, is defined as follows:
For each $\sigma'\in V_\tet$, define $d(\sigma')$ to be the weight of the longest path from $\sigma'$ to $\sigma$ in \GB{r}.
We define $\Tst(\sigma')=D-d(\sigma')$.
\end{definition}
We remark that computing longest paths in the bounds graph (and in particular, computing~$D$) is an easy task, because there are no positive cycles in the graph. (E.g., Bellman-Ford can be used, if we work with $w'=-w$). 
Clearly, every path of \GB{r} that ends in~$\sigma$ starts at some node~$\sigma'$. Moreover, both $\sigma$ and~$\sigma'$ appear at finite times in~$r$. Since~$\sigma'\atleast{\scriptscriptstyle{D}}\sigma$ in~$r$, if follows that $D$ must be finite. 
By definition of~$D$, it follows that $\Tst$ assigns non negative times to all nodes in~$V_\tet$. 
We now turn to show that $\Tst$ is a valid timing function.
\begin{lemma}  
\Llabel{lem:slowrun}
Let $r\in\Rrep$, and let \tet\ be a node of $r$.
Then $\Tst$ is valid timing function for $V_\theta$ w.r.t.\ \GB{r}.
\end{lemma}
\begin{proof}
Assume that $\basic(\tet,r)=\sigma$, and let $\sigma_1,\sigma_2\in V_\sigma$ such that $(\sigma_1,\sigma_2)$ is an edge in \GB{r}.
Assume by way of contradiction that $\Tst(\sigma_2)<\Tst(\sigma_1)+w(\sigma_1,\sigma_2)$.
According to \Cref{def:slow} we have that $\Tst(\sigma_2)-\Tst(\sigma_1)=d(\sigma_1)-d(\sigma_2)<w(\sigma_1,\sigma_2)$.
Recall that $d(\sigma_1)$ is the weight of the longest path from $\sigma_1$ to $\sigma$.
Let us observe the following alternative path from $\sigma_1$ to $\sigma$: 
We begin by walking on the edge $(\sigma_1,\sigma_2)$  and from $\sigma_2$ on its longest path to $\sigma$, whose weight is $d(\sigma_2)$. 
We thus obtain
a path from $\sigma_1$ to $\sigma$ with weight of $w(\sigma_1,\sigma_2)+d(\sigma_2)>d(\sigma_1)$, contradicting the fact that $d(\sigma_1)$ is the weight of the longest path from $\sigma_1$ to~$\sigma$.
\end{proof}

Now comes the ``heart'' of the claim, that each valid timing function on a p-closed set describes a run that contains (almost) exactly the nodes of this set, and where the times of those nodes are exactly according to the timing function.
Note that in every run $r$, every process~$i$ has at least one basic node, which is \nodeL{i,r_i(0)}. This is~$i$'s initial node (the node with $i$'s initial local state).
If we want to find a run that contains all the nodes from some p-closed set, this run must still contain at least one basic node for each agent.
It will be convenient to mark the set of initial basic nodes in a run.
Thus,
we define the set of initial basic nodes of $r$ by \[\VrInit=\{\sigma' : \mbox{$\sigma'$ appears in $r$ and }\timeof_r(\sigma')=0\}.\]
Essentially all of the information about the relative timing of nodes in a run~$r$ is captured by paths in the bounds graph \GB{r}. This is not true, however, for nodes of $\VrInit$, because $\timeof_r(\sigma)=0$ for all nodes $\sigma\in\VrInit$. It follows that the  relative timing of initial nodes is tightly correlated. This does not cause special problems, because in our model, a process performs an action only when it receives a message --- either an external input or a message from one of its neighbors in the network. Nevertheless, the initial nodes need to be handled in a slightly different manner when we use the bounds graph \GB{r} to construct runs that conform with particular timing functions. We proceed as follows.

\begin{lemma}[Run by timing] 
\Llabel{lem:runbytime}
Let $r\in\Rrep$, let $V$ be the set of nodes of \GB{r}, and let $V'\subseteq V$ be a p-closed subset of $V$.
Moreover, let $T:V'\to\Nat$ be a valid timing function. 
Let us denote $\GB{r}\!\!\downarrow_{V'}=(V',E')$ the subgraph of $\GB{r}$ that is induced by $V'$ . 
Then there exists a legal run in~$\Rrep$, denoted by $r[T]$, such that \mbox{(i) ~$\GB{r[T]}=(V'\cup\VrInit, E')$}, 
and (ii)  $\timeof_{r[T]}(\sigma')=T(\sigma')$ holds for each $\sigma'\in \big(V'\setminus \VrInit\big)$.
\end{lemma}
\begin{proof}
For ease of exposition, we shall denote $r[T]$ by~$r'$. This run is defined as follows:
\begin{itemize}
\item  $r'(0)=r(0)$.
\item Let $i\in\Proc$ and $\sigma_i\in V'$ be an $i$-node. If $\sigma_i$ receives external inputs in $r$, then process~$i$ receives the same inputs at time $T(\sigma_i)$ in $r'$.
\item Let~$\mu$ be a message that is sent at an $i$-node $\sigma_i\in V'$ and is received in $r$ at the $j$-node~$\sigma_j$.
Note that in such case $\sigma_j\in V'$, by the p-closedness of $V'$ (since there is an edge $\sigma_j\Atleast{{-\scriptscriptstyle{\ub{ij}}}}\sigma_i$ in~$\GB{r}$).  
If $\mu$ is in transit in the run $r'$ just before time $T(\sigma_j)$, then it is delivered to~$j$ at time $T(\sigma_j)$ in~$r'$. 
\item No other messages are delivered in the run~$r'$. 
\end{itemize}
We claim the following by induction on $m$:
\begin{enumerate}
\item The prefix of $r'$ up to time $m$ is a prefix of a legal run in $\Rrep$;
\item All the basic nodes in $r'$ up to time $m$ are either from $V'$ or $\VrInit$, and
\item For all $j\in\Proc$:
	\begin{enumerate}
	\item If $m=T(\sigma_j)$ for some $j$-node $\sigma_j\in\big(V' \setminus \VrInit\big)$, then $\nodeL{j,r'_j(m)}=\sigma_j$.
	\item For all $j$-nodes $\sigma'_j\in \big(V' \setminus \VrInit\big)$, if $m<T(\sigma'_j)$, then $\nodeL{j,r'_j(m)}\ne \sigma'_j$.
	\end{enumerate} 
\end{enumerate}
At $m=0$, the claims are true as (1) $r'(0)=r(0)\in G_0$, thus it is a prefix of legal run, 
(2) 
the set of nodes of $r'$ at time $0$ are exactly $\VrInit$,
and (3) if $T(\sigma_j)=0$, then $\sigma_j\in\VrInit$. Prove: if $\sigma_j$ has a predecessor, say $\sigma'_j$, then there will be an edge $\sigma'_j\atleast{1}\sigma_j$ in $\GB{r}$. From the p-closedness of $V'$ it must be that $\sigma'_j\in V'$, and from the validity of $T$ it must be that $T(\sigma'_j)\leq T(\sigma_j)-1<T(\sigma_j)=0$, contradicting the fact that $T\geq0$. Thus, $\sigma_j$ has no predecessor, and thus $\sigma_j\in\VrInit$.
The last two claims follows directly (at $m=0$).

Let $m>0$, and assume that the claims are true up to $m-1$. 
To show that the prefix of~$r'$ up to time $m$ is a prefix of a legal run in $\Rrep$, we must make sure that no message is delivered before its lower bound, or  remains in transit after its upper bound has expired. 
Let $\mu$ be a message that is sent in $r'$ by the $i$-node~$\sigma_i$ at time~$t_\mu<m$ to agent~$j$, and that is still in transit just before time $m$.
By the induction, as all the nodes before time $m$ are either from $V'$ or $\VrInit$, and as nodes from $\VrInit$ don't send messages (as they are also the initial nodes in $r'$), then $\sigma_i\in V'$.
By the induction, clearly $t_\mu=T(\sigma_i)$.
Assume that $\mu$ is received in $r$ by a node $\sigma_j$, then it will be delivered, according to the definition of $r'$, at $m'=T(\sigma_j)$ (only if $m'>t_\mu$).
As $T$ is a valid timing function, $\lb{ij}\leq m'-t_\mu\leq\ub{ij}$, thus it will be delivered in bounds (and $m'>t_\mu$), and in particular if $T(\sigma_j)=m$ then it will be delivered at $m$.
This proves the first claim.

Assume that $m=T(\sigma_j)$ for some $j$-node $\sigma_j\in \big(V'\setminus \VrInit\big)$, and let~$\sigma'_j$ be~$\sigma_j$'s predecessor.
Note that $\sigma'_j\in V'$.
Denote $T(\sigma'_j)=m'$.
From the validity of $T$ it follows that $m'<m$, and that there is no other $j$-node $\sigma''_j\in V'$ for which  $m'<T(\sigma''_j)<m$.
According to the definition of $r'$, no messages will be delivered to~$j$ between times $m'+1$ and $m-1$, so its local state will remain unchanged, i.e.,  $\nodeL{j,r'_j(m-1)}=\nodeL{j,r'_j(m')}=\sigma'_j$.
We must prove that the messages that are delivered in $r'$ at $m$ to $j$ are exactly the messages that are delivered to $\sigma_j$ in~$r$. This will imply that  $\nodeL{j,r'_j(m)}=\sigma_j$.
Observing all the nodes that send the messages to $\sigma_j$, it follows from the validity of the timing function that their timings are before $m$, and so by the induction they appear in $r'$, and as we have seen above these messages will be delivered at $m$.
clearly no other messages are delivered in $r'$ to $j$ at $m$, according to the definition of $r'$.

If $m<T(\sigma_j)$ for some $j$-node $\sigma_j\in \big(V'\setminus \VrInit\big)$, then we have by the induction that $\nodeL{j,r'_j(m-1)}\neq\sigma_j$. If $j$ doesn't receive messages at time $m$, then $\nodeL{j,r'_j(m)}=\nodeL{j,r'_j(m-1)}\neq\sigma_j$.
Otherwise, it must be that $m=T(\sigma'_j)$ for some $j$-node $\sigma'_j\in V'$, and then $\nodeL{j,r'_j(m)}=\sigma'_j\neq\sigma_j$.

By the induction, all the nodes before time~$m$ are from $V'$ or $\VrInit$.
As any message that is delivered at $m$ in $r'$ is delivered there to a node from $V'$, it follows that this property still holds at time $m$.

Thus, we proved the induction, which establishes that $r'\in\Rrep$,
 and that $r'$ contains exactly the nodes from $V'\cup\ \VrInit$, and with the timing according to $T$, as claimed.
\end{proof}

The last thing we need before we can prove \Cref{thm:zz} is the following technical claim:
\begin{lemma}  
\Llabel{lem:thetasaresame}
Let $r\in\Rrep$, let $V$ be the set of nodes of \GB{r}, and let $V'\subseteq V$ be a p-closed subset of $V$.
Let~\tet\ be a node of~$r$, and let $\sigma=\basic(\tet,r)$.
Moreover, let $T:V'\to\Nat$ be a valid timing function, and let $r[T]$ be a run satisfying the conditions guaranteed by \Cref{lem:runbytime}.  
Then either both \tet\ and $\sigma$ appear in $r[T]$ and $\sigma=\basic(\tet,r[T])$, or 
neither \tet\ nor $\sigma$ appears in $r[T]$.
\end{lemma}
\begin{proof}
Assume that $\tet=\nodeF{\sigma',p'}$.
As $\basic(\tet,r)=\sigma$, the node~$\sigma$ (by its definition) receives a direct message chain that goes from~$\sigma'$ through $p'$.
Thus, if $\sigma$ appears in $r[T]$ (or in any other run from $\Rrep$), then it must receive this message chain from $\sigma'$, and so $\sigma'$ must also appear in $r[T]$, and it must send that message chain through $p'$. Thus, $\tet$ also appears in $r[T]$, and $\sigma=\basic(\tet,r[T])$.

Now, assume that $\tet$ appears in $r[T]$, and $\basic(\tet,r[T])=\sigma''$. We want to prove that $\sigma''=\sigma$.
By \Cref{lem:runbytime}, as $\sigma''$ appears in $r[T]$, either $\sigma''\in V'$ or $\sigma''\in\VrInit$.
If $\sigma''\in\VrInit$, then $\tet$ receives no message in $r[T]$, which is possible only if $p'$ is a singleton, and then $\sigma=\sigma'=\sigma''$.
If $\sigma''\in V'$, then it appears in $r$, and it receives the same message chain that goes from~$\sigma'$ through $p'$, just as $\sigma$ does (in $r$).
It is possible only if $\sigma=\sigma''$.
\end{proof}

{\bf We are now ready to prove \Cref{thm:zz}:}
\begin{proof}[Proof of \Cref{thm:zz}]
Let $x\in\Int$ and assume
that \Rrep\ supports $\te{1}\atleast{x}\te{2}$. 
Let ${r\in\Rrep}$ be some run such that~\te{1} and \te{2} are nodes of $r$, and where both $\timeof_r(\te{1})>0$ and $\timeof_r(\te{2})>0$.
Assume that ${\basic(\te{1},r)=\sigma_1}$ and ${\basic(\te{2},r)=\sigma_2}$.
We denote $r[\Ts{\te{2}}]$ by~$r'$.
By \Cref{lem:runbytime}, $\sigma_2$ appears in $r'$.
Thus, by \Cref{lem:thetasaresame}, the node~\te{2} also appears in $r'$, and $\basic(\te{2},r')=\sigma_2$.
As \te{2} appears in $r'$ and as \Rrep\ supports $\te{1}\atleast{x}\te{2}$, and according to the definition of \textit{supports}, we have that \te{1} also appears in $r'$ and $\timeof_{r'}(\te{1})+x\leq\timeof_{r'}(\te{2})$.
Once again, by \Cref{lem:thetasaresame}, $\sigma_1$ appears in $r'$, and $\basic(\te{1},r')=\sigma_1$.
By the definition of $r'$, for every node $\sigma'$ of $r'$ that is not an initial node in $r$ (i.e., $\timeof_r{\sigma'}>0$) there is a path $P'$ (the longest path) in \GB{r} from $\sigma'$ to $\sigma_2$ with $\weight(P')=\timeof_{r'}(\sigma_2)-\timeof_{r'}(\sigma)$.
Thus, as $\timeof_r(\sigma_1)=\timeof_r(\te{1})>0$, there is a path $P$ (the longest path) in \GB{r} from $\sigma_1$ to $\sigma_2$ with $\weight(P)=\timeof_{r'}(\sigma_2)-\timeof_{r'}(\sigma_1)=\timeof_{r'}(\te{2})-\timeof_{r'}(\te{1})$.
As $\timeof_{r'}(\te{2})-\timeof_{r'}(\te{1})\geq x$, we have that $\weight(P)\geq x$.
Now, by \Cref{lem:zzisgb}, there is a zigzag pattern $Z$ in $r$ from \te{1} to \te{2} with $\weight(Z)=\weight(P)\geq x$.
\end{proof}



\section{Proof of \cref{thm:vzz}} 
%
Let $r\in\Rrep(\Prot,\gamma)$ for an \ffip\ $\Prot$, and let $\sigma$ be a basic node of $r$.
The first (and easy) direction in \Cref{thm:vzz} claims that if there is in $r$ a $\sigma$-visible zigzag $Z$ relating two \sigAware\ nodes, say \te{1} and \te{2}, with $\weight(Z)=x$, then 
\mbox{$(\Rrep,r)\models\Ksub{\sigma}(\te{1}\atleast{x}\te{2})$}.

Let $r,r'\in\Rrep$ such that $r'\sim_{\sigma}r$, and assume that $Z$ is a $\sigma$-visible zigzag in $r$ from \te{1} to \te{2} with $\weight(Z)=x$. 
We will soon show that~$Z$ is also a zigzag pattern in~$r'$. Since $\weight(Z)$ depends only on the pattern and not on the run~$r$, we have by \Cref{thm:zzsimple}  that $(\Rrep,r')\sat\te{1}\atleast{x}\te{2}$, and by definition of knowledge we have that \mbox{$(\Rrep,r)\models\Ksub{\sigma}(\te{1}\atleast{x}\te{2})$}, as desired. 

To complete the argument, we now show that the pattern~$Z$ appears in the run~$r'$. Roughly speaking, this amounts to showing that, since~$Z$ is $\sigma$-visible in~$r$ (i) each one of the two-legged forks in~$Z$ is guaranteed to appear in~$r'$, and (ii) these two-legged forks form a zigzag pattern in~$r'$, i.e., the head of each fork appears no later than the tail of its successor.  
%
%
Recall that $Z=(F_1,F_2,\ldots,F_c)$ is a $\sigma$-visible zigzag from \te{1} to \te{2} in $r$. By definition of visible zigzags, it follows that 
 $\base(F_c)=\nodeF{\sigma',p'}$, for some node $\sigma'$ satisfying $\sigma'\lam_r\sigma$.
Moreover, assume that $\tail(F_c)=\base(F_c)\oCompose p''$ and $\head(F_c)=\base(F_c)\oCompose p'''$.
As $\sigma'\lam_r\sigma$ and $\Rrep=\Rrep(\Prot,\gamma)$ for an \ffip~$\Prot$, the three nodes $\base(F_c)$, $\tail(F_c)$ and $\head(F_c)$ appear in $r'$, and thus~$F_c$ is also a \tlf\ in~$r'$.
Let $m<c$. Since~$Z$ is a visible zigzag we have that $\head(F_m)\lam_r\sigma$, it follows by \Cref{lem:sim} and the definition of two-legged forks (and by the fact that the protocol is an \ffip) that $\tail(F_m)$, $\base(F_m)$ and $\head(F_m)$ also appear in~$r'$, and so $F_m$ is also a \tlf\ in $r'$.
Moreover, $\basic(\head(F_m),r')=\basic(\head(F_m),r)$, and as $\tail(F_{m+1})$ appears in $r$ no earlier than $\head(F_m)$ does, the same occurs in~$r'$.
Thus, $Z$ is also a zigzag in~$r'$, as claimed 
%
$\hfill\square$ \\


We now turn to proving the other direction of \Cref{thm:vzz}, which claims that the existence of a $\sigma$-visible zigzag is also necessary for~$\sigma$'s  knowledge of the timed precedence. The argument will extend and generalize the proof of \Cref{thm:zz}. In that proof, a run $r'=r[Ts_r^\sigma]$ satisfying the constraints imposed by $\GB{r}$ was constructed, in which roughly speaking every node was moved as far in the future possible with respect to~$\sigma$. A path in the constraints graph~$\GB{r}$ then determined the timing of every node in~$r'$, and such a path implied the existence of a corresponding zigzag pattern. The basic nodes of~$r'$ all appeared in~$r$ as well. 
For the current theorem, we need a more subtle definition of an alternative run in which a visible zigzag appears. In this run, only the basic nodes that are in the past of $\sigma$ must be kept the same as in $r$, while other basic nodes are not necessarily the same. Moreover, its construction will make use of $\GE{r,\sigma}$, the extended bounds graph of $\sigma$ in~$r$. 

In \Cref{sec:highlights} we defined $\GE{r,\sigma}$, the extended bounds graph of $\sigma$, whose purpose is to capture all of the timing information that is available to~$\sigma$. 
We define it here once again, more formally. 
We begin by defining a subgraph of \GB{r} that $\sigma$ can construct from its local state (at least under the assumption of FIP protcol):
\begin{definition}
Let~$\sigma$ be a node of~$r$.
We define the \defemph{local bounds graph} of $\sigma$ in $r$ to be the weighted graph $\GB{r,\sigma}$ that is induced from \GB{r} by the nodes $V'=\{\sigma' : \sigma'\lam_r\sigma\}$.
\end{definition}
The following is a simple definition that will be useful next:
\begin{definition}
Let $\sigma$ and $\sigma'$ be nodes of~$r$, and let~$\sigma'$ be an $i$-node.
We say that~$\sigma'$ is a \defemph{boundary node} of $\sigma$, if (i) $\sigma'\lam_r\sigma$ and (ii) $\sigma''\lam_r\sigma'$ holds for each $i$-node $\sigma''\lam_r\sigma$.
\end{definition}
The boundary nodes are simply the last nodes, one for each agent, that appear in the past of $\sigma$.
 
Can we repeat \Cref{lem:runbytime} while using $\GB{r,\sigma}$ as the set of nodes? I.e., can we define a valid timing function over the nodes of \GB{r,\sigma} and get a legal run?
Clearly, such a timing function will ensure the order of nodes from the past of $\sigma$, and guarantee that messages sent among such nodes are delivered within the stated bounds. However, there might still be messages that are sent by nodes from the past of $\sigma$ that are not received in~$r$ inside the past of $\sigma$, and \GB{r,\sigma} doesn't treat such messages.
The result is that a 
valid timing function for \GB{r,\sigma} might cause such message to be received before a boundary node~$\sigma'$ of $\sigma$. This is, of course, inconsistent with
node~$\sigma$'s view in~$r$, since~$\sigma'\lamr\sigma$ and thus~$\sigma$ knows that the message was not delivered before~$\sigma'$. 
In the proof of \Cref{thm:zz} we avoided this 
problem by considering only p-closed sets of nodes of \GB{r}. 
However, the set of nodes of \GB{r,\sigma} is typically not  p-closed, as even $\sigma$ in itself probably sends messages to other nodes, which are clearly not in its past. Recall that \GB{r} contains an edge from a node at which a message is received to one at which the message was sent (with bound $-\ub{ji}$). Consequently, there are paths in \GB{r} from nodes that are not in \GB{r,\sigma} to nodes inside \GB{r,\sigma}.

To overcome this issue, we add special nodes to $\GB{r,\sigma}$, and extend it to a new graph, denoted $\GE{r,\sigma}$.
As described in \Cref{sec:highlights}, we add to this graph a new node for each agent (we call it an \textit{auxiliary node}), which stands for the successor of this agent's boundary node.
By adding an edge with weight $1$ from a boundary node to the auxiliary node of the same agent, we make sure that the auxiliary node must come after the boundary node.
In addition, we add edges that make sure that messages that go outside the past of~$\sigma$ will not be delivered earlier than at the relevant auxiliary node.
\begin{definition} [Extended local  bounds graph]
Let $r\in\Rrep$, let $\sigma$ be a node of~$r$ and let $\GB{r,\sigma}=(V,E,W)$ be the local  bounds graph of $\sigma$ in $r$.
We define the \defemph{extended local  bounds graph} of $\sigma$ in $r$ to be the weighted graph $\GE{r,\sigma}=(V\cup V', E\cup E'\cup E''\cup E''',W\cup W')$ where:
\begin{itemize}
\item $V'=\{\psi_i : i\in\Proc\}$; (each $\psi_i$ will stand for the auxiliary node of process~$i$)
\item $E'=\{(\sigma',\psi_i): i\in\Proc \mbox{ and $\sigma'$ the boundary $i\text{-}node$ of $\sigma$ in $r$}\}$, and $W'(\sigma',\psi_i)=1$; 
\item $E''=\{(\psi_i,\sigma_j): i,j\in\Proc$ and $\sigma_j\in V$ is a $j$-node that sends a message to process $i$ that is not received in any $i\text{-}node$ from $V\}$, and $W'(\psi_i,\sigma_j)=-\ub{ji}$; and
\item $E'''=\{(\psi_i,\psi_j) : (j,i)\in \Chan\}$, and $W'(\psi_i,\psi_j)=-\ub{ji}$.
\end{itemize}
\end{definition}

We call the nodes~$V$ of \GB{r,\sigma}\ \defemph{original nodes}. 
The edges of $E'$ make sure that auxiliary nodes don't appear before the respective boundary nodes.
The edges of $E''$ make sure that messages that are sent from within the past of $\sigma$, and are not received in the past of $\sigma$, can be received at the nodes of $V'$ (or after them) without exceeding their upper bounds.
The edges of $E'''$ ensure that messages sent after a boundary node of \GB{r,\sigma}\ can be delivered at or after the auxiliary node of the receiving process. This enables the construction of a run consistent with the extended graph in which the past of~$\sigma$ coincides with what it is in~$r$. 

We state without a proof that a valid timing function over the nodes of $\GE{r,\sigma}$ can be extended to a valid run, where the original nodes will appear according to the timing function. A more restricted version of this property will be proven and used in our analysis below. 
And what about the timing of the auxiliary nodes? The answer is that any direct message chain that goes outside the past must be delivered after their timings.
Later on we prove these results for a specific run.\\

Recall that in the proof of \Cref{thm:zz} we have found a valid timing that was based on the longest paths in \GB{r}, and then constructed a run with this timing.
As times in that run were corresponding to paths in $\GB{r}$, we could prove the existence of respective zigzags.
However, in the case of \Cref{thm:zz}, when we were talking on two nodes, say \te{1} and \te{2}, we knew exactly what are the corresponding basic nodes, and then we could find the longest path between them in \GB{r}.
\GE{r,\sigma}, on the contrary, contains only the basic nodes from $\past(r,\sigma)$, so how can we have a path in \GE{r,\sigma} between two general \sigAware\ nodes? 
Moreover, are paths in \GER\ still produce timing constraints?

It appears that, just as in \GB{r}, paths in \GE{r,\sigma} between two basic nodes indeed define constraints on their timing (constraints that apply in any run $r'\in\Rrep$ such that $r'\sim_\sigma r$).
Moreover, a path in \GE{r,\sigma} whose one or two endpoints is an auxiliary node, put constraint for some set of pairs of \sigAware\ nodes. 
For this reason, we say that such path is a \textit{\cpath} between the corresponding pair of \sigAware\ nodes (even though these nodes don't appear explicitly in the graph). \\

We start by observing 4 different types of paths in \GER:
\begin{enumerate}
\item A path of the form $I=(\psi_{i_c},\psi_{i_{c-1}},\ldots,\psi_{i_2},\sigma_1)$ that contains a (possibly empty) sequence of auxiliary nodes, and ends at a basic node from $\past(r,\sigma)$ (in this case~$\sigma_1$);
\item A path whose both ends are basic nodes from $\past(r,\sigma)$;
\item A path of the form $S=(\sigma_2,\psi_{j_2},\psi_{j_3},\ldots,\psi_{j_d})$  that starts at a basic node from $\past(r,\sigma)$ (in this case~$\sigma_2$) and followed by a sequence of auxiliary nodes; and
\item A path that contains only auxiliary nodes.
\end{enumerate}

Note that each path in \GER\ corresponds to one of the above types, or to a concatenation of paths from these types.
As we shall see, each path type defines a constraint between two nodes of a specific form, and also ensures the existence of a respective zigzag pattern between these two nodes (just as \GB{r} did).
For each type, we first describe the intuition for the relevant constraints and then prove the existence of the respective zigzag patterns.
These zigzag patterns are actually $\sigma$-visible zigzags.

we start with a type~1 path of the form: $I=(\psi_{i_c},\psi_{i_{c-1}},\ldots,\psi_{i_2},\sigma_1)$.
According to the definition of \GER, there is an edge between $\psi_{i_k}$ and $\psi_{i_{k-1}}$ only if $(i_{k-1},i_k)\in\Chan$, and the edge has a weight of $-\upb{i_{k-1}i_k}$.
Thus, it must be that $(i_{k-1},i_k)\in\Chan$ for each $3\leq k\leq c$.
Note that $\weight(\psi_{i_c},\psi_{i_{c-1}},\ldots,\psi_{i_2})=-\ubSym([i_2,i_3,\ldots,i_c])$.
Next, there is an edge between the auxiliary node of process $i_2$ ($\psi_{i_2}$) to a basic node ($\sigma_1$) only if a message is sent at the basic node to process $i_2$, and that message isn't received in the past of $\sigma$. 
Assuming that $\sigma_1$ is an $i_1$-node, the weight of that edge will be $-\upb{i_1i_2}$, and so we get that $\weight(I)=-\ubSym([i_1,i_2,\ldots,i_c])$.
Denote $q=[i_1,i_2,\ldots,i_c]$, so that $\weight(I)=-\ubSym(q)$.
The general node $\te{1}\eqdef\nodeF{\sigma_1,q}$ describes the node that receives the message chain that is sent at $\sigma_1$ and goes along~$q$.
According to the upper bounds we have that ${\timeof_r(\te{1})\leq\timeof_r(\sigma_1)+\ubSym(q)}$, thus ${\timeof_r(\sigma_1)\geq\timeof_r(\te{1})+(-\ubSym(q))}$, and so we get that $(\Rrep,r)\models\te{1}\atleast{\weight(I)}\sigma_1$.
The result is that $I$ defines a constraint between \te{1} (which is constructed according to $I$) and $\sigma_1$.
More formally we define:
\begin{definition}
Let $I=(\psi_{i_c},\psi_{i_{c-1}},\ldots,\psi_{i_2},\sigma_1)$ be a path in \GER, where $\sigma_1$ is an $i_1$-node.
Then we say that $I$ is a \cpath\ in \GER\ between \nodeF{\sigma_1,[i_1,i_2,\ldots,i_c]} and $\sigma_1$.
\end{definition}
Such path defines a \svz:
\begin{lemma}
\Llabel{lem:type1}
Let $I=(\psi_{i_c},\psi_{i_{c-1}},\ldots,\psi_{i_2},\sigma_1)$ be a path in \GER, where $\sigma_1$ is an $i_1$-node.
Then there is a \svz~$Z$ between \nodeF{\sigma_1,[i_1,i_2,\ldots,i_c]} and $\sigma_1$, where $\weight(Z)=\weight(I)$, and where $\head(Z)\lam_r\sigma$.
\end{lemma}
\begin{proof}
We construct a \tlf\ $F$ so that $\base(F)=\head(F)=\nodeF{\sigma_1,i_1}$ and $\tail(F)=\nodeF{\sigma_1,[i_1,i_2,\ldots,i_c]}$.
We define $Z=(F)$. As $\base(F)=\head(F)=\sigma_1\lam_r\sigma$, $Z$ is a \svz.
Moreover, $\weight(Z)=\weight(F)=-\ubSym([i_1,i_2,\ldots,i_c])=\weight(I)$.
\end{proof}

Next, we have a path of type 2, i.e.~a path whose both endpoints are basic nodes.
Note that a path that contains only basic nodes (and no auxiliary nodes), clearly defines a constraint between the basic nodes at its ends, as such path is also found in $\GB{r}$.
But, what happen when in the middle of a path of basic nodes, we add a walk through auxiliary nodes?
For example, consider the following path: $Q=[\sigma_i,\psi_i,\psi_k,\sigma_j]$.
Note that we can go from $\sigma_i$ to $\psi_i$ only if they belong to the same process (process~$i$) and $\sigma_i$ is a boundary node w.r.t $\sigma$. The weight of the respective edge is $1$.
From $\psi_i$ to $\psi_k$ we walk on an edge of weight $-\upb{ki}$, and from $\psi_k$ to $\sigma_j$ there is an edge with weight $-\upb{jk}$ only if a message that is sent at $\sigma_j$ to process $k$, isn't received in the past of $\sigma$.
This message is received at the \sigAware\ node \nodeF{\sigma_j,[j,k]}, and this node must appear no later than $\upb{jk}$ after $\sigma_j$. 
\nodeF{\sigma_j,[j,k]} will send a message to agent $i$, and the $i$-node that will receive that message, \nodeF{\sigma_j,[j,k,i]}, must appear no later than $\upb{ki}$ afterwards. 
Now, as \nodeF{\sigma_j,[j,k,i]} cannot be in the past of $\sigma$, it must appear after $\sigma_i$, i.e. at least 1 time unit after it.
We denote $\te{i}=\nodeF{\sigma_j,[j,k,i]}$.
So, we have that (i) $\timeof_r(\sigma_i)+1\leq\timeof_r(\te{i})$ and (ii) $\timeof_r(\te{i})\leq\timeof_r(\sigma_j)+\upb{jk}+\upb{ki}$.
Thus, $\timeof_r(\sigma_i)+1\leq\timeof_r(\sigma_j)+\upb{jk}+\upb{ki}$, and $\timeof_r(\sigma_i)+(1-\upb{jk}-\upb{ki})\leq\timeof_r(\sigma_j)$.
Note that $\weight(Q)=1-\upb{jk}-\upb{ki}$, thus $\timeof_r(\sigma_i)+\weight(Q)\leq\timeof_r(\sigma_j)$ and indeed the path~$Q$ defines a constraint between $\sigma_i$ and $\sigma_j$ (its two endpoints).
Generalizing the last result, we can see how a path $Q'$ whose both ends are basic nodes, say $\sigma_1$ and $\sigma_2$, defines a constraint between $\sigma_1$ and $\sigma_2$, i.e.~$\timeof_r(\sigma_1)+\weight(Q')\leq\timeof_r(\sigma_2)$.
More formally we define:
\begin{definition}
Let $Q$ be a path in $\GE{r,\sigma}$.
If both the first and last nodes of $Q$ are basic nodes, say $\sigma_1$ and $\sigma_2$ respectively, then we say that $Q$ is a \cpath\ in \GER\ between $\sigma_1$ and $\sigma_2$.
\end{definition}
Such path defines a \svz:
\begin{lemma}
\Llabel{lem:type2}
Let $Q$ be a path in $\GE{r,\sigma}$ whose first node is $\sigma_1$ and last node is $\sigma_2$, where both $\sigma_1$ and $\sigma_2$ are basic nodes from $\past(r,\sigma)$.
Then there is a \svz~$Z$ between $\sigma_1$ and $\sigma_2$, where $\weight(Z)=\weight(Q)$, and where $\head(Z)\lam_r\sigma$.
\end{lemma}
\begin{proof}
We can divide $Q$ to alternate segments of paths: A segment along original (basic) nodes (from \GB{r,\sigma}), followed by a segment along auxiliary nodes, and then another segment along original nodes and so on.
We define those segments such that they all begin and end at an original node (we can do so as $Q$ begins and ends at original nodes).
So, assume that $Q$ is the concatenation $Q=P_0\oCompose S_1\oCompose P_1\oCompose S_2\oCompose P_2\oCompose\ldots\oCompose P_k$ where (1) the last node of each segment is the first node of the following segment, (2) $P_i$, for $0\leq i\leq k$, contains only basic nodes, and $S_i$, for $1\leq i\leq k$, contains only auxiliary nodes, (3) $P_0$ starts at $\sigma_1$, and (4) $P_k$ ends at $\sigma_2$.
We prove that looking at a prefix of $Q$ that is constructed by the first $d=2m+1$ segments, we have a \svz\ $Z_m$ between $\sigma_1$ and the last node of $P_m$, with a weight equals to the weight of the prefix, and $\head(Z_m)\lam_r\sigma$.
We prove this for every $m$, and so there is also a \svz\ $Z$ between $\sigma_1$ and $\sigma_2$, with $\weight(Z)=\weight(Q)$ and $\head(Z)\lam_r\sigma$ as required.

For $m=0$ ($d=1$) this is followed by \Cref{lem:zzisgb}, as $P_0$ is also a path in \GB{r}.
Let $m>0$, and assume that the above claim is true up to $(m-1)$. 
I.e., we have a \svz~($Z_{m-1}$) from $\sigma_1$ to~$\sigma'$, where~$\sigma'$ is the last node of~$P_{m-1}$,
and such that $\weight(Z_{m-1})$ is equal to weight of the first $2m-1$ segments of $Q$.
Note that $\head(Z_{m-1})=\sigma'\lam_r\sigma$.

The next segment is $S_{m}$.
Assume that $S_{m}=(\sigma_i,\psi_{l_1}, \psi_{l_2},\ldots,\psi_{l_k},\sigma_j)$. 
The only way we can go from the original node, $\sigma_i$, to the following auxiliary node, $\psi_{l_1}$, is if $\sigma_i$ is a boundary node of $\sigma$, and $\psi_{l_1}$ is the auxiliary node of the same agent (and the weight of the relevant edge is 1).
Afterwards, we can go from $\psi_{l_n}$ to $\psi_{l_{n+1}}$ only if $(l_{n+1},l_n)\in\Chan$. Generally, $q=[l_k,l_{k-1},\ldots,l_1]$ must be a path in $\Net$.
From $\psi_{l_k}$ we can go to $\sigma_j$ only if a message was sent at $\sigma_j$ to process $l_k$, and this message was not received in the past of $\sigma$.
Assuming that $\sigma_j$ is a $j$-node, the result is that $S_m$ represents the \sigAware\ node $\te{i}=\nodeF{\sigma_j,\Comp{j}{q}}$ 
whose agent equals to the agent of the boundary node $\sigma_i$, and thus it must appear after it.
And so we can define a new \tlf\, $F_m$, where $\head(F_m)=\base(F_m)=\sigma_j$ and $\tail(F_m)=\te{i}$, and such that $\weight(F_m)=\weight(S_m)-1$.
As the head of the previous \svz\ ($Z_{m-1}$) is $\sigma_i$, which comes before the tail of $F_m$, we can concatenate them to a new \svz\ $Z'_m$, whose weight perfectly matches the weight of the path $Q$ up to (including) $S_m$, and $\head({Z'_m})=\sigma_j\lam_r\sigma$.

Next, observing $P_m$, note that it is also a path in $\GB{r}$, and by \Cref{lem:zzisgb} there exists a zigzag pattern, denoted by $Z'$, in $r$, from the first node of $P_m$ ($\sigma_j$) to its last node, with weight $\weight(Z')=\weight(P_m)$.
As all of the nodes of $P_m$ are from the past of $\sigma$, the zigzag is clearly $\sigma$-visible, with its head also in the past of~$\sigma$.
As its tail is equal to the head of $Z'_m$, we can concatenate them and get a new \svz\ $Z_m$ with $\weight(Z_m)$ equals to the prefix of $Q$ up to the end of $P_m$, as required.
\end{proof}

Now assume we have a type 3 path: $S=(\sigma_2,\psi_{j_2},\psi_{j_3},\ldots,\psi_{j_d})$ .

According to the definition of \GER, there is an edge between $\sigma_2$ (a basic node from $\past(r,\sigma)$) and $\psi_{j_2}$ only if $\sigma_2$ is a boundary node of process $j_2$, and then the edge between them has a weight of 1.
Just as in the case of type 1 paths, there is an edge between $\psi_{j_k}$ and $\psi_{j_{k+1}}$ only if $(j_{k+1},j_k)\in\Chan$, and the edge has a weight of $-\upb{j_{k+1}j_k}$.
Thus, it must be that $(j_{k+1},j_k)\in\Chan$ for each $1\leq k\leq d-1$.
Note that $\weight(\psi_{j_2},\psi_{j_3},\ldots,\psi_{j_d})=-\ubSym([j_d,j_{d-1},\ldots,j_2])$.

Denote $q=[j_d,j_{d-1},\ldots,j_2]$, and note that $\weight(S)=1-\ubSym(q)$.
Let $\tet'=\nodeF{\sigma',p'}$ be a \sigAware\ node, where $p'$ is a non singleton path in $\Net$ that ends at process $j_d$.
Since $p'$ is not a singleton, and $\tet'$ is a \sigAware\ node, it must appear only after the boundary node of process~$j_d$ from $\past(r,\sigma)$.
Now consider the node $\tet_2=\ocomp{\tet'}{q}$.
As $\te{2}$ is a node that receives a message chain from $\tet'$, we have that
$\timeof_r(\tet_2)\leq\timeof_r(\tet')+\ubSym(q)$.
\te{2} is an $j_2$-node, that must be outside $\past(r,\sigma)$, and so it must appear after the boundary node of process $j_2$, and we have that $\timeof_r(\sigma_2)+1\leq\timeof_r(\tet_2)$.
Combining the two inequalities, we get that $\timeof_r(\sigma_2)+1\leq\timeof_r(\tet')+\ubSym(q)$, thus $\timeof_r(\sigma_2)+(1-\ubSym(q)\leq\timeof_r(\tet')$.
As $\weight(S)=1-\ubSym(q)$ we get that $(\Rrep,r)\models\sigma_2\atleast{\weight(S)}\tet'$.
The result is that $S$ defines a constraint between $\sigma_2$ (it's first node) and a \sigAware\ node, $\tet'$, which is an arbitrary \sigAware\ node whose message chain ends at the process of the last auxiliary node in $S$.

More formally we define:
\begin{definition}
Let $S=(\sigma_2,\psi_{j_2},\psi_{j_3},\ldots,\psi_{j_d})$ be a path in \GER, where $\sigma_2$ is a $j_2$-node, and let $\tet'=\nodeF{\sigma',p'}$ be a \sigAware\ node.
If $p'$ is a non singleton path in $\Net$ that ends at process $j_d$,
then we say that $S$ is a \cpath\ in \GER\ between $\sigma_2$ and $\tet'$.
\end{definition}
Such path defines a \svz:
\begin{lemma}
\Llabel{lem:type3}
Let $S=(\sigma_2,\psi_{j_2},\psi_{j_3},\ldots,\psi_{j_d})$ be a \cpath\ between $\sigma_2$ and $\tet'$ in \GER.
Then there is a \svz~$Z$ between $\sigma_2$ and $\tet'$ where $\weight(Z)=\weight(S)$.
\end{lemma}
\begin{proof} 
We define a new trivial \tlf\ $F_1$ where $\head(F)=\base(F)=\tail(F)=\nodeF{\sigma_2,j_2}$.
Now, assume that $\tet'=\nodeF{\sigma',p'}$, where the last process of $p'$ is $j_d$.
Define $q=[j_d,j_{d-1},\ldots,j_2]$.
Note that $\weight(S)=1-\ubSym(q)$.
We define a second \tlf\ $F_2$ where $\head(F_2)=\base(F_2)=\nodeF{\sigma',p'}$, and $\tail(F_2)=\nodeF{\sigma',\ocomp{p'}{q}}$.
Note that $\weight(F_2)=-\ubSym(q)$.

We can now define the zigzag pattern $Z=(F_1,F_2)$. 
As $F_1$ and $F_2$ are not joined, $\weight(Z)=\weight(F_1)+1+\weight(F_2)=1-\ubSym(q)=\weight(S)$.
As $\head(F_1)=\sigma_2\lam_r\sigma$ and $\base(F_2)=\nodeF{\sigma',p'}$ where $\sigma'\lam_r\sigma$, then $Z$ is a \svz, as required.
\end{proof}

Finally, assume we have a type 4 path.
Type 4 are paths that contain only auxiliary nodes.
Assume that $P=(\psi_{i_d},\psi_{i_{d-1}},\ldots,\psi_{i_1})$, and define $q=[i_1,i_2,\ldots,i_d]$.
Note that $\weight(P)=-\ubSym(q)$.
Let $\te{1}$ and \te{2} be two \sigAware\ nodes, such that $\te{1}=\ocomp{\te{2}}{q}$.
Clearly $\timeof_{r}(\te{1})\leq\timeof_{r}(\te{2})+\ubSym(q)$, and so $\timeof_{r}(\te{2})\geq\timeof_{r}(\te{1})+(-\ubSym(q))=\timeof_{r}(\te{1})+\weight(P)$, and we get that $P$ defines constraint between $\te{1}$ and $\te{2}$.

More formally we define:
\begin{definition}
Let $P=(\psi_{i_d},\psi_{i_{d-1}},\ldots,\psi_{i_1})$ be a path in \GER, and let $\te{1}$ and \te{2} be two \sigAware\ nodes.
Define $q=[i_1,i_2,\ldots,i_d]$.
If $\te{1}=\ocomp{\te{2}}{q}$, we say that $P$ is a \cpath\ from~\te{1} to~\te{2}.
\end{definition}
Such path defines a \svz:
\begin{lemma}
\Llabel{lem:type4}
Let $\te{1}$ and \te{2} be two \sigAware\ nodes, such that $\te{1}=\ocomp{\te{2}}{q}$, and let $P$ be a (type~4) \cpath\ between \te{1} and \te{2}.
Then there is a \svz~$Z$ between $\te{1}$ and $\te{2}$ where $\weight(Z)=\weight(P)$.
\end{lemma}
\begin{proof}
We construct a \tlf\ $F$ so that $\base(F)=\head(F)=\te{2}$ and $\tail(F)=\te{1}=\ocomp{\te{2}}{q}$.
Assume that $\te{2}=\nodeF{\sigma_2,p_2}$, and that $P=(\psi_{i_d},\psi_{i_{d-1}},\ldots,\psi_{i_1})$, so that $q=[i_1,i_2,\ldots,i_d]$.
As $\sigma_2\lam_r\sigma$ and $base(F)=\te{2}=\nodeF{\sigma_2,p_2}$, then $Z=(F)$ is a $\sigma$-visible zigzag with $\weight(Z)=\weight(F)=-\ubSym(q)=\weight(P)$ as required.
\end{proof}

Every path in \GER\ can be uniquely represented by exactly one of the following:
\begin{itemize}
\item Joined concatenation of type 1 and type 2 paths.
\item Joined concatenation of type 1,2 and 3 paths.
\item Type 4 path.
\end{itemize}
We use above the term \textit{Joined concatenation} whereby the last node of a path is also the first node of the following path.
Recall that type 1 and 2 might contain only a single node, and the same goes for a joined concatenation of them.

For a single path of each type, we defined above the nodes between which there is a corresponding \cpath, and proved the existence of a corresponding \svz.
We shall now do the same when concatenating the different types (the two first options in the list above).
The first concatenation is of type 1 and type 2 paths:
\begin{definition}
Let $I$ be a type 1 path that is a \cpath\ from \te{1} to $\sigma_1$, and let $Q$ be a type 2 path from $\sigma_1$ to $\sigma_2$.
We say that $P=\ocomp{I}{Q}$ is a \cpath\ in \GER\ between $\te{1}$ and $\sigma_2$.
\end{definition}
Such path defines a \svz:
\begin{lemma}
\Llabel{lem:type12}
Let $I$ be a type 1 path that is a \cpath\ from \te{1} to $\sigma_1$, and let $Q$ be a type 2 path from $\sigma_1$ to $\sigma_2$.
Then there exists a \svz~$Z$ between \te{1} and $\sigma_2$ where $\weight(Z)=\weight(\ocomp{I}{Q})$.
\end{lemma}
\begin{proof}
This follows directly by concatenating the zigzags we get from \Cref{lem:type1,lem:type2}. 
\end{proof}

The second concatenation is of type 1, 2 and 3 paths.
To keep it simple, we say that $P=\ocomp{I}{Q}$, which is a concatenation of type 1 and 2 paths, is called type 5 path. 
\begin{definition}
Let $P'$ be a type 5 path that is a \cpath\ from \te{1} to $\sigma_2$, and let $\tet'$ be a node and $S$ be a type 3 path such that $S$ is a \cpath\ from $\sigma_2$ to $\tet'$.
We say that $P=\ocomp{P'}{S}$ is a \cpath\ in \GER\ between $\te{1}$ and $\tet'$.
\end{definition}
Such path defines a \svz:
\begin{lemma}
\Llabel{lem:type123}
Let $P'$ be a type 5 path that is a \cpath\ from \te{1} to $\sigma_2$, and let $\tet'$ be a node and $S$ be a type 3 path such that $S$ is a \cpath\ from $\sigma_2$ to $\tet'$.
Then there exists a \svz~$Z$ between \te{1} and $\tet'$ where $\weight(Z)=\weight(\ocomp{P'}{S})$.
\end{lemma}
\begin{proof}
This follows directly by concatenating the zigzags we get from \Cref{lem:type12,lem:type3}. 
\end{proof}

In the definitions we have seen above, any \cpath\ $P$ between two nodes, $\tet$ and $\tet'$, in \GER\, implies a \svz~$Z$ between these two node, with $\weight(Z)=\weight(P)$.
We generalize it in the following lemma:
\begin{lemma}
\Llabel{lem:paths}
Let $r\in\Rrep$, let $\sigma$ be a node of~$r$ and let $\te{1}$ and $\tet'_1$ be \sigAware\ nodes. 
Moreover, let $\te{2}=\Comp{\tet'_1}{p'}$ be another \sigAware\ node.
If $P$ is a \cpath\ in $\GE{r,\sigma}$ between \te{1} and $\tet'_1$, then there exists a $\sigma$-visible zigzag~$Z$ in~$r$, between~$\te{1}$ and~$\te{2}$, with $\weight(Z)=\weight(P)+\lbSym(p')$
\end{lemma}
\begin{proof}
We already know that if $P$ is a \cpath\ in $\GE{r,\sigma}$ between~\te{1} and~$\tet'_1$, then there exists a \svz~$Z'$ in~$r$, between~$\te{1}$ and~$\tet'_1$, with $\weight(Z')=\weight(P)$.
As $\te{2}=\Comp{\tet'_1}{p'}$, it directly follows that we can define \te{2} to be the head of the topmost \tlf\ of $Z'$ (instead of $\tet'_1$), and get a new $\sigma$-visible zigzag, $Z$, between~$\te{1}$ and~$\te{2}$, with $\weight(Z)=\weight(P)+\lbSym(p')$ as required.
\end{proof}

Let $\te{1}$ and \te{2} be \sigAware\ nodes in $r$.
If we can find a path in $\GE{r,\sigma}$ that will give us a $\sigma$-visible zigzag $Z$ from \te{1} to \te{2} (according to \Cref{lem:type4,lem:type12,lem:type123}), and we can construct a valid run $r'\sim_\sigma r$ where $\timeof_{r'}(\te{2})-\timeof_{r'}(\te{1})=\weight(Z)$, 
then we can prove \Cref{thm:vzz}.
We can do so because if $\Ksub{\sigma}(\te{1}\atleast{x}\te{2})$ and $r'\sim_\sigma r$, then $(\Rrep,r')\models\te{1}\atleast{x}\te{2}$, and as $Z$ is also a zigzag in $r'$, where $\weight(Z)=\timeof_{r'}(\te{2})-\timeof_{r'}(\te{1})$, it must be that $\weight(Z)\geq x$ as required.

In a similar manner to the proof of \Cref{thm:zz}, for a given $\te{1}$ we find a single timing (and a single run) that will give us the required results for any other node \te{2}.
But, what happens if for some node \te{2} we cannot find a path in \GE{r,\sigma} that will give us the corresponding zigzag?
For example, assume that $\te{1}=\nodeF{\sigma_1,i}$, and $\te{2}=\nodeF{\sigma_2,j}$, so that both $\te{1}$ and $\te{2}$ are in the past of $\sigma$.
In this case, it can be seen that if there is no path in \GE{r,\sigma} from $\sigma_1$ to $\sigma_2$, then there is no \cpath\ between \te{1} and \te{2}.
In such case there is no constraint for how early \te{2} can appear before \te{1} (but maybe only constraint on how much late it can appear), and so $\sigma$ cannot know statements of the form $\te{1}\atleast{x}\te{2}$ as we can find an equivalent run in which $\te{2}$ appears less than $x$ after $\te{1}$, for any $x$.
For those nodes, like \te{2} in this example, we will define the timing using a parameter that controls how early those nodes will appear.\\

Let $\sigma'$ be a basic node such that $\sigma'\lamr\sigma$, and let $\GB{r,\sigma}=(V,E)$ and $\GE{r,\sigma}=(V',E')$.
We divide the nodes from $V'$ to four sets w.r.t $\sigma'$:
\begin{itemize}
\item $\Vyes{\sigma'}=\{\sigma''\in V : $ there is a path from $\sigma'$ to $\sigma''\}$
\item $\Vno{\sigma'} =\{\sigma''\in V : $ there is no path from $\sigma'$ to $\sigma''\}$
\item $\Ayes{\sigma'} =\{\psi\in V'-V : $ there is a path from $\sigma'$ to $\psi\}$
\item $\Ano{\sigma'}  =\{\psi\in V'-V : $ there is no path from $\sigma'$ to $\psi\}$
\end{itemize}

\newcommand{\Tf}[1]{\ensuremath{T_\gamma[r,\sigma,{#1}]}}
\newcommand{\Tfs}{\ensuremath{T_\gamma}}

Now we define the promised timing function. This timing function defines the times of nodes from $\Vyes{\sigma'}$ and $\Ayes{\sigma'}$ to be according to the longest paths from $\sigma'$ to them, and it define the times of nodes from $\Vno{\sigma'}$ to be according to the longest paths from them to $\sigma$ (this is done just to make sure they receive a valid timing) and at least $\gamma$ (a timing parameter) before the earliest node from $\Vyes{\sigma'}$.
As we shall see, there is no importance for the nodes of $\Ano{\sigma'}$.
\begin{definition} [Fast-Timing]
Let $r\in\Rrep$, $\gamma\in\Nat$, $\sigma$ be a node of $r$, and let $\GE{r,\sigma}=(V',E',W')$.
Moreover, let $\tet'=\nodeF{\sigma',p'}$ be a \sigAware\ node.
The \defemph{$\bm{\gamma}$-Fast timing} of $\tet'$ in $r$ w.r.t $\sigma$, $\Tf{\tet'}$, is defined as follows:\\
For each $\sigma''\in\Vyes{\sigma'}$, define $d(\sigma'')$ to be the weight of the longest path from $\sigma'$ to $\sigma''$ in \GE{r,\sigma}. The same goes for each $\psi\in \Ayes{\sigma'}$.\\
For each $\sigma''\in\Vno{\sigma'}$, define $f(\sigma'')$ to be the weight of the longest path from $\sigma''$ to $\sigma$ in \GE{r,\sigma}.\\
Denote $F_1=max\{f(\cdot)\}$, $F_2=min\{f(\cdot)\}$ and $D=min\{d(\cdot)\}$.
If $\Vno{\sigma'}=\emptyset$, then define $F_1=F_2=0$.
We define $\Tf{\tet'}:V'\to\Nat$ such that:
\begin{itemize}
\item If $\sigma''\in\Vno{\sigma'}$, then $\Tf{\tet'}(\sigma'')=F_1-f(\sigma'')$.
\item If $\sigma''\in\Vyes{\sigma'}$, then $\Tf{\tet'}(\sigma'')=1+F_1-F_2+\gamma-D+d(\sigma'')$.
\item If $\psi\in\Ayes{\sigma'}$, then $\Tf{\tet'}(\psi)=1+F_1-F_2+\gamma-D+d(\psi)$.
\item If $\psi\in\Ano{\sigma'}$, then $\Tf{\tet'}(\psi)=0$.
\end{itemize}
\end{definition}

The fast timing is a valid timing function w.r.t \GE{r,\sigma}, in the sense of \Cref{def:vtf}:
\begin{lemma} 
\Llabel{lem:fasttiming} 
Let $r\in\Rrep$, $\gamma\in\Nat$, $\sigma$ be a node of $r$, $\GE{r,\sigma}=(V',E',w)$, and let $\tet'$ be a \sigAware\ node.
Define $\Tfs=\Tf{\tet'}$.
Then:
\begin{enumerate}
\item $\Tfs(v)\geq 0$ for any $v\in V'$.
\item For each pair of nodes $v_1,v_2\in V'$, such that $(v_1,v_2)\in E'$, we have that $\Tfs(v_2)\geq\Tfs(v_1)+w(v_1,v_2)$.
\item If $\sigma_i$ and $\sigma'_i$ are two $i$-nodes from the past of $\sigma$, and $\timeof_r(\sigma_i)<\timeof_r(\sigma'_i)$, then $\Tfs(\sigma_i)<\Tfs(\sigma'_i)$.
\item Let $\sigma_1$ and $\sigma_2$ be two basic nodes such that $\sigma_1\in\Vyes{\sigma'}$ and $\sigma_2\in\Vno{\sigma'}$. Then ${\Tfs(\sigma_2)+\gamma<\Tfs(\sigma_1)}$.
\end{enumerate}
\end{lemma}
\begin{proof}

The first claim is immediate from the definition of $\Tfs$. For the second claim, assume by way of contradiction that there are $(v_1,v_2)\in E'$, where $\Tfs(v_2)<\Tfs(v_1)+w(v_1,v_2)$.

If $v_1\in\Vyes{\sigma'}$ or $v_1\in\Ayes{\sigma'}$, then also $v_2\in\Vyes{\sigma'}$ or $v_2\in\Ayes{\sigma'}$, thus $\Tfs(v_2)-\Tfs(v_1)=d(v_2)-d(v_1)<w(v_1,v_2)$.
Recall that $d(v_2)$ is the weight of the longest path from $\sigma'$ to $v_2$. 
Let us observe the following alternative path from $\sigma'$ to $v_2$: 
We begin by walking on the longest path from $\sigma'$ to $v_1$, whose weight is $d(v_1)$, and then walk over the edge $(v_1,v_2)$. And so we got a path from $\sigma'$ to $v_2$ with weight of $d(v_1)+w(v_1,v_2)>d(v_2)$ contradicting the fact that $d(v_2)$ should be the longest path.

Otherwise, assume that only $v_2\in\Vyes{\sigma'}$ or $v_2\in\Ayes{\sigma'}$. By the construction of \Tfs, the times of nodes to whom there is a path from $\sigma'$ is always bigger 
than the times of those that don't have such paths. Thus, the only edge that can cause troubles is if $(v_1,v_2)$ represents a message sending and receiving, where it has a weight bigger than 1. However, in such case $(v_2,v_1)\in E'$ and then it cannot be that there is a path towards $v_2$ and not towards $v_1$.

The last case is when there is no path from $\sigma'$ to both $v_1$ and $v_2$.
If both are original nodes, then $\Tfs(v_2)-\Tfs(v_1)=f(v_1)-f(v_2)<w(v_1,v_2)$.
Recall that $f(v_1)$ is the weight of the longest path from $v_1$ to $\sigma$. 
Let us observe the following alternative path from $v_1$ to $\sigma$: 
We begin by walking on $(v_1,v_2)$ and then along the longest path from $v_2$ to $\sigma$, whose weight is $f(v_2)$.
And so we got a path from $v_1$ to $\sigma$ with weight of $f(v_2)+w(v_1,v_2)>f(v_1)$ contradicting the fact that $f(v_1)$ should be the longest path.

Otherwise, if both nodes are auxiliary nodes, the claim is trivially true as $\Tfs(v_2)-\Tfs(v_1)=0$ and clearly $w(v_1,v_2)<0$.
It is even stronger if $v_1$ is auxiliary and $v_2$ is original, as then $\Tfs(v_2)-\Tfs(v_1)>0$ while still $w(v_1,v_2)<0$.
The last case is when $v_2$ is auxiliary and $v_1$ is original. 
As $v_2\in\Ano{\sigma'}$, then there is no path in $\GE{r,\sigma'}$ to $v_2$. However, as $\sigma\in\Vyes{\sigma'}$, that means that there is no path from $\sigma$ to $v_2$. This is possible only if there is no path in $\Net$ between from the agent of $v_2$ to the agent of $\sigma$, which means that no node of the agent of $v_2$ can be in $\sigma$'s past, so there cannot be an original node $v_1$ such that $(v_1,v_2)$ is an edge.\\

The third claim is immediate from the second claim, as there is a $1$ weight edge between successive nodes of the same agent.

Concerning the last claim: $\Tfs(\sigma_1)-\Tfs(\sigma_2)=\gamma+1+(f(\sigma_2)-F_2)+(d(\sigma_1)-D)$.
Note that $F_2\leq f(\sigma_2)$ and $D\leq d(\sigma_1)$, and so we get that $\Tfs(\sigma_1)-\Tfs(\sigma_2)>\gamma$ as required.

\end{proof}

\newcommand{\fastG}[2]{\ensuremath{fast^{#1}_\sigma(r,{#2})}}
\newcommand{\fast}[1]{\fastG{\gamma}{#1}}

Based on the fast-timing, we create a full run:
\begin{definition} [Fast run]    
Let $r\in\Rrep$, $\gamma\in\Nat$, $\sigma$ be a node of $r$ and let $\tet'$ be a \sigAware\ node.
Define $\Tfs=\Tf{\tet'}$.
The \defemph{$\bm{\gamma}$-fast run} of~$\tet'$ in~$r$, $r'=\fast{\tet'}$, is defined inductively as follows:
For $m=0$, r'(0)=r(0).
For $m>0$, assume that we have the global states of $r'$ up to time $m-1$. 
We have to decide what the environment does at time~$m$ with messages in transit.
Let $\mu$ be such message, that is sent in $r'$ by the $i$-node~$\sigma_i$ at time~$t_\mu$ to agent~$j$.
\begin{enumerate}
\item If $\sigma_i\lam_r\sigma$, and its message to $j$ is received in $r$ by a node $\sigma_j\lam_r\sigma$, then deliver $\mu$ only if $m=\Tfs(\sigma_j)$. Otherwise,
\item Assume that $\tet'=\nodeF{\sigma',p'}$. If $\mu$ is a message that is sent by a node $\tet''=\nodeF{\sigma',p''}$ in $r'$, where $\Comp{p''}{j}$ is a prefix of $p'$, 
then deliver $\mu$ only if $m=t_\mu+\ub{ij}$. Otherwise,
\item If $m\geq max(t_\mu+\lb{ij}, \Tfs(\psi_j))$, then deliver $\mu$. Otherwise, don't deliver $\mu$ at time $m$.
\end{enumerate}
Concerning external inputs (messages from the environment), deliver such to an agent $j\in\Proc$ only if ${m=\Tfs(\sigma_j)}$ for some $j$-node $\sigma_j\lam_r\sigma$ that receives that external input in $r$.\footnote{
While we often refer to the run $r$, in many cases it uses only for clarity, and not truly required. For example, there is no need to mention $r$ when discussing the messages that $\sigma$ receives, as $\sigma$ already contains this information}
\end{definition}
\begin{lemma}  
\Llabel{lem:fast}
Let $r\in\Rrep$, $\gamma\in\Nat$, let $\sigma$ be a node of $r$ and let $\tet'=\nodeF{\sigma',p'}$ be a \sigAware\ node, where $\timeof_r(\tet')>0$.
Define $r'=\fast{\tet'}$.
Then:
\begin{enumerate}
\item $r'\in\Rrep$
\item $r'\sim_\sigma r$
\item Let $\sigma_1$ be a basic node such that $\sigma_1\lam_r\sigma$, and $\timeof_r(\sigma_1)>0$. Then 
$\timeof_{r'}(\sigma_1)=\Tfs(\sigma_1)$.
\item Let $\sigma_1$ and $\sigma_2$ be two basic nodes such that $\sigma_1\in\Vyes{\sigma'}$ and $\sigma_2\in\Vno{\sigma'}$. Then ${\timeof_{r'}(\sigma_2)+\gamma<\timeof_{r'}(\sigma_1)}$.
\item Let $\te{2}=\nodeF{\sigma_2,p_2}$ be a \sigAware\ node, where $\sigma_2\in\Vyes{\sigma'}$.
Then there exists another \sigAware\ node $\tet'_1$ such that (i) $\te{2}=\ocomp{\tet'_1}{p'_1}$, (ii) there is a \cpath\ $P$ in \GE{r,\sigma} between $\tet'$ and $\tet'_1$, and (iii)  $\timeof_{r'}(\te{2})-\timeof_{r'}(\tet')=\weight(P)+\lbSym(p'_1)$.
\end{enumerate}
\end{lemma}

\begin{proof}
We prove by induction on $m$ that:
\begin{enumerate}
\item The prefix of $r'$ up to time $m$ is a prefix of a legal run in $\Rrep$, and
\item Let $p''$ be a prefix of $p'$ (the path that defines $\tet'$). If $\timeof_{r'}(\nodeF{\sigma',p''})=m$, then $m=\Tfs(\sigma')+\ubSym(p'')$.
\item For all $j\in\Proc$:
	\begin{itemize}
	\item If $m=\Tfs(\sigma_j)$ for some $j$-node $\sigma_j\lam_r\sigma$, then $\nodeL{j,r'_j(m)}=\sigma_j$ and $\timeof_{r'}(\sigma_j)=m$.
	\item Otherwise, if $m\neq\Tfs(\sigma_j)$ for every $j$-node $\sigma_j\lam_r\sigma$, and $m<\Tfs(\sigma_j)$ for some $j$-node $\sigma_j\lam_r\sigma$, then no messages are delivered to agent $j$ (so its local state remains the same).
	\end{itemize} 
\end{enumerate}
At $m=0$, the claims are true:
\begin{enumerate}
\item we have that $r'(0)=r(0)\in G_0$, so the prefix of $r'$ at time 0 is also a prefix of $r\in\Rrep$;
\item If $\timeof_{r'}(\nodeF{\sigma',p''})=0$, then it must be that $p''$ is a singleton and $\sigma'$ is an initial node of $r'$ (a node from time $0$). However, as $r'(0)=r(0)$, $\sigma'$ will be in this case also an initial node of $r$. 
According to the definition of the model, initial nodes never send messages (basic nodes send messages only at the moment they receive one, but initial nodes never receive messages, as they cease to exist once a message is received).
Thus, if $\timeof_{r}(\sigma')=0$, no message chain can leave $\sigma'$ in~$r$.
As $\tet'=\nodeF{\sigma',p'}$, it means that $p'$ must be a singleton in this case,
and so $\timeof_{r}(\tet')=\timeof_{r}(\sigma')=0$. 
However, recall that $\timeof_{r}(\tet')>0$, and so $\sigma'$ cannot be an initial node in~$r$ and~$r'$, and we get that $\timeof_{r'}(\nodeF{\sigma',p''})\neq0$ for any prefix $p''$ of $p'$; and
\item Assume that $\Tfs(\sigma_j)=0$ holds for some $j$-node $\sigma_j\lam_r\sigma$.
If $\sigma_j$ is not an initial node in~$r$, it must have a predecessor, say $\sigma'_j$, where also $\sigma'_j\lam_r\sigma$, and $\timeof_{r}(\sigma'_j)<\timeof_{r}(\sigma_j)$.
From $\Cref{lem:fasttiming}$ we get in this case that $\Tfs(\sigma'_j)<\Tfs(\sigma_j)=0$, contradicting the positiveness of $\Tfs$.
Thus, $\sigma_j$ must be an initial node in $r$.
As $r'(0)=r(0)$, we have that $\sigma_j$ is an initial node of $r'$ as well, and so $\timeof_{r'}(\sigma_j)=0$ as required.
The last claim is trivially true, as no messages are delivered at time $0$.
\end{enumerate}
Let $m>0$, and assume that the claims are true up to time $m-1$. 

For the validity of $r'$ at $m$, we must make sure that no message is delivered before its lower bound, or is remained in transit if its upper bound will pass.
Let $\mu$ be a message that is sent in $r'$ by the $i$-node~$\sigma_i$ at time~$t_\mu<m$ to agent~$j$, and that is still in transit before time $m$.
In the fast run definition there are three conditions that describe if $\mu$ will be delivered at $m$ or don't.
concerning the first condition, if $\sigma_i\lam_r\sigma$, then by the induction $t_\mu=\timeof_{r'}(\sigma_i)=\Tfs(\sigma_i)$. 
if this message is received in $r$ by a node $\sigma_j\lam_r\sigma$, then it will be delivered at $m'=\Tfs(\sigma_j)$.
By lemma \ref{lem:fasttiming}, $\lb{ij}\leq \Tfs(\sigma_j)-\Tfs(\sigma_i)\leq\ub{ij}$ and so $\lb{ij}\leq m'-t_\mu\leq\ub{ij}$. Thus, the message $\mu$ will be delivered in bounds, and in particular if $m=\Tfs(\sigma_j)$ then it will be delivered at $m$.

Delivery of a message according to the second condition is promised to be in the upper bound, hence inside bounds.
The third condition promises to obey the lower bound but possibly not the upper bound, if $t_\mu+\ub{ij}<\Tfs(\psi_j)$.
If $\sigma_i\lam_r\sigma$, then it must be that $\mu$ is received in $r$ outside the past of $\sigma$, or otherwise we would have stopped in the first condition. 
So, if indeed $\sigma_i\lam_r\sigma$, then $(\psi_j,\sigma_i)$ is an edge in \GE{r,\sigma} with weight of $-\ub{ij}$, and as the fast timing is a valid timing, it must be that $\Tfs(\psi_j)-\ub{ij}\leq\Tfs(\sigma_i)=t_\mu$, so that $t_\mu+\ub{ij}\geq\Tfs(\psi_j)$ so it won't be delivered behind the upper bound.

Otherwise, if $\sigma_i\not\lam_r\sigma$, then $\sigma_i$ sent $\mu$ after receiving another message.
It has received that other message following either the second condition or the third.
If it was due to the third condition, then it must be that $t_\mu\geq\Tfs(\psi_i)$. As $\Tfs(\psi_j)-\ub{ij}\leq\Tfs(\psi_i)$, then once again we get that $t_\mu+\ub{ij}\geq\Tfs(\psi_j)$.
Recall that $\tet'=\nodeF{\sigma',p'}$.
If $\sigma_i$ received a message due to the second condition, then there exists some path $p''=[i_0,i_1,\ldots,i_c]$ 
which is prefix of $p'$ (the path that defines $\tet'$), where $i_c=i$ and $i_0$ is the agent of $\sigma'$.
By the induction, it follows that $t_\mu=\timeof_{r'}(\sigma_i)=\Tfs(\sigma')+\ubSym(p'')$.
Note that $(\psi_i,\psi_{i_{c-1}},\ldots,\psi_{i_1},\sigma')$ is a path in \GE{r,\sigma} with weight of $-\ubSym(p'')$, hence $\Tfs(\psi_i)-\ubSym(p'')\leq\Tfs(\sigma')$ and so $\Tfs(\psi_i)\leq t_\mu$. 
Combining it with the fact that $\Tfs(\psi_j)-\ub{ij}\leq\Tfs(\psi_i)$, we get for the last time that $t_\mu+\ub{ij}\geq\Tfs(\psi_j)$ as required.\\

For the second claim of the induction, let $p''$ be a prefix of $p'$ (the path that defines $\tet'$).
Assume that $p''=\Comp{p'''}{j}$ and that $m=\timeof_{r'}(\nodeF{\sigma',p''})$.
This means that the node \nodeF{\sigma',p'''} (we assume it is an $i$-node) sent a message to agent~$j$, and that message was delivered in~$r'$ at time $m$.
This delivery must be according to the second delivery rule in the definition of~$r'$, thus $m=\timeof_{r'}(\nodeF{\sigma',p'''})+\upb{ij}$.
By the induction, we get that $\timeof_{r'}(\nodeF{\sigma',p'''})=\Tfs(\sigma')+\ubSym(p''')$, and so $m=\Tfs(\sigma')+\ubSym(p''')+\upb{ij}=\Tfs(\sigma')+\ubSym(p'')$ as required.

Next, we want to prove the third claim of the induction.
Assume that $m=\Tfs(\sigma_j)$ for some \mbox{$j$-node} $\sigma_j\lam_r\sigma$.
By the induction, $j$ didn't receive any message since the last $m'<m$ for which $m'=\Tfs(\sigma'_j)$ for some $j$-node $\sigma'_j\lam_r\sigma$, or otherwise, if there is no such $m'$, then since $m'=0$. By \Cref{lem:fasttiming}, that previous $\sigma'_j$ is the node that comes before $\sigma_j$ in $r$ (or if $m'=0$, then~$\sigma_j$ is the first node of~$j$ in~$r$, after the initial node).
Now we must prove that the messages that are delivered in $r'$ at $m$ to $j$ are exactly the messages that are delivered to $\sigma_j$, and then it follows that $\nodeL{j,r'_j(m)}=\sigma_j$ and that $\timeof_{r'}(\sigma_j)=m$.
Observing all the nodes that send the messages to $\sigma_j$, their timings is before $m$, and thus by the induction they appear in $r'$, and as we have seen above these messages will be delivered at $m$.
We still must prove that no other messages are delivered in $r'$ to $j$ at $m$.
Clearly no other message that is received in the past of $\sigma$ will be delivered here (according to the first condition of delivering messages in the fast run).
We still have to prove that it is true for the two other conditions.
The third condition is easy, as $m=\Tfs(\sigma_j)<\Tfs(\psi_j)$ thus no such message will be delivered at $m$.
Concerning the second condition, assume that there is some message that is sent in~$r'$ to agent $j$ from a node $\nodeF{\sigma',p''}$, where $q=\Comp{p''}{j}$ is a prefix of $p'$ (the path that defines $\tet'$). 
If this message is received at time~$m$, then we have seen that $m=\Tfs(\sigma')+\ubSym(q)$.
From similar reasons that described above it must be that $\Tfs(\psi_j)-\ubSym(q)\leq\Tfs(\sigma')$, and so $m=\Tfs(\sigma')+\ubSym(q)\geq\Tfs(\psi_j)$ contradicting the fact that $m=\Tfs(\sigma_j)<\Tfs(\psi_j)$.

Even if $m\neq\Tfs(\sigma_j)$ for every $\sigma_j\lam_r\sigma$, but there is some $\sigma'_j\lam_r\sigma$ such that $m<\Tfs(\sigma'_j)$, then our last analysis still applies, as $m<\Tfs(\sigma'_j)<\Tfs(\psi_j)$, so there will be no messages delivered to $j$ at $m$ in such case.
Thus, we proved the induction, which proves the first three claims of \Cref{lem:fast}.
The fourth claim follows directly from the third claim and from \Cref{lem:fasttiming}.\\

Now we shall prove the last claim of \Cref{lem:fast}.
Recall that $\tet'=\nodeF{\sigma',p'}$, and let $\te{2}=\nodeF{\sigma_2,p_2}$ be a \sigAware\ node, where $\sigma_2\in\Vyes{\sigma'}$.
We want to prove that there exists another \sigAware\ node, $\tet'_1$, such that (i) $\te{2}=\ocomp{\tet'_1}{p'_1}$, (ii) there is a \cpath\ $P$ in \GE{r,\sigma} between $\tet'$ and $\tet'_1$, and (iii)  $\timeof_{r'}(\te{2})-\timeof_{r'}(\tet')=\weight(P)+\lbSym(p'_1)$.

Assume that $p'=[i_0,i_1,\ldots,i_c]$ (so that $\sigma'$ is an $i_0$-node).
According to the second condition of message delivery in the fast run, we get that $\timeof_{r'}(\tet')-\timeof_{r'}(\sigma')=\ubSym(p')$.
Note that $I=(\psi_{i_c},\psi_{i_{c-1}},\ldots,\psi_{i_1},\sigma')$ is a type 1 path, that is also a \cpath\ from $\tet'$ to $\sigma'$, with $\weight(I)=-\ubSym(p')=\timeof_{r'}(\sigma')-\timeof_{r'}(\tet')$.

First, lets assume that $p_2$ is a singleton, so effectively $\te{2}=\sigma_2$.
Note that $\sigma'$ is not an initial node, as we assume that $\timeof_r(\tet')>0$, and message chains cannot leave an initial node.
As in \GER\ there are no edges that enter initial nodes, there cannot be paths from $\sigma'$ to initial nodes.
Thus, as $\sigma_2\in\Vyes{\sigma'}$, it must be that $\timeof_r(\sigma_2)>0$.
By the third claim, we have that $\timeof_{r'}(\sigma_2)=\Tfs(\sigma_2)$ and 
$\timeof_{r'}(\sigma')=\Tfs(\sigma')$, and so $\timeof_{r'}(\sigma_2)-\timeof_{r'}(\sigma')=\Tfs(\sigma_2)-\Tfs(\sigma')=d(\sigma_2)$.
By the definition of the fast timing, $d(\sigma_2)$ is the weight of a (longest) type 2 path, $Q$, in \GE{r,\sigma} from $\sigma'$ to $\sigma_2$. Thus, we have that $\weight(Q)=\timeof_{r'}(\sigma_2)-\timeof_{r'}(\sigma')$.
We can concatenate $I$ and $Q$, and get a \cpath\ $P=\ocomp{I}{Q}$ from $\tet'$ to $\sigma_2$ with $\weight(P)=\weight(I)+\weight(Q)=\timeof_{r'}(\sigma_2)-\timeof_{r'}(\tet')$.
We choose $\tet'_1=\te{2}=\sigma_2$, and we are done.

Assume otherwise, that $p_2$ is not a singleton.
The messages in the message chain that leaves $\sigma_2$ and goes along $p_2$ are received according to the second and third conditions of message delivery in the definition of $r'$.
We choose to paths in $\Net$, $p'_1$ and $p'_2$, such that (i) $p_2=\ocomp{p'_2}{p'_1}$ and (ii) denoting $\tet'_1=\nodeF{\sigma_2,p'_2}$ (so that $\te{2}=\ocomp{\tet'_1}{p'_1}$), we have that all the messages in the message chain that leaves $\tet'_1$ and goes along $p'_1$ are received in their lower bound (according to the third condition of message delivery), and where the last message that is received by $\tet'_1$ from the message chain that leaves $\sigma_2$ and goes along $p'_2$, is not received in its lower bound.
If there are \textbf{no} messages that are received in the lower bound along $p_2$, then we choose $p'_1$ to be a singleton, and $\tet'_1=\te{2}$.
If there are \textbf{only} messages that are received in the lower bound along $p_2$, then we choose $p'_2$ to be a singleton, and $\tet'_1=\sigma_2$.
Note that $\timeof_{r'}(\te{2})-\timeof_{r'}(\tet'_1)=\lbSym(p'_1)$.
We must prove now that there is a \cpath\ $P$ from $\tet'$ to $\tet'_1$ with $\weight(P)=\timeof_{r'}(\tet'_1)-\timeof_{r'}(\tet')$ and we are done.

Now there are three options: (1) All the above mention messages are received in the lower bound, otherwise (2) the last message that is received at $\tet'_1$ is received according to the second condition of message delivery in~$r'$, i.e.~in its upper bound, or otherwise (3) the last message that is received at $\tet'_1$ is received by the limitation of the timing of some auxiliary node (in the third condition of message delivery in~$r'$).
 
In the first case, we have that $\tet'_1=\sigma_2$, and just as the case where $p_2$ was a singleton, we can find a path $P$ from $\tet'$ to $\sigma_2$ where $\timeof_{r'}(\tet'_1)-\timeof_{r'}(\tet')=\weight(P)$, 
as required.

In the second case, according to the definition of $r'$ it must be that $\sigma_2=\sigma'$ and that $p'_2$ is a prefix of $p'$.
So, assume that $p'=\ocomp{p'_2}{p''}$, and that $p''=[l_1,l_2,\ldots,l_e]$.
Note that $P=(\psi_{l_e},\psi_{l_{e-1}},\ldots,\psi_{l_1})$ is a \cpath\ (type~4) from $\tet'$ to $\tet'_1$, and that $\timeof_{r'}(\tet'_1)-\timeof_{r'}(\tet')=-\ubSym(p'')=\weight(P)$, 
as required.

Considering the last case, assume that $\tet'_1$ is a $j$-node, so that $\timeof_{r'}(\tet'_1)=\Tfs(\psi_j)$.
Clearly $\Tfs(\psi_j)>0$, or otherwise $r'$ wouldn't be valid, and so $\psi_j\in\Ayes{\sigma'}$.
Thus, there is a path from $\sigma'$ to $\psi_j$ in \GER.
Define $\sigma''$ to be the last basic node on that path.
We can divide this path to a type 2 path, $Q$, from $\sigma'$ to $\sigma''$, and a type 3 path, $S$, from $\sigma''$ to $\psi_j$.
Note that $S$ defines a \cpath\ between $\sigma''$ and $\tet'_1$.
Concerning the weights of these paths, we have that $\weight(Q)=d(\sigma'')$ and $\weight(S)=d(\psi_j)-d(\sigma'')$, where $d(\cdot)$ is the weight of the longest path from $\sigma'$ to the relevant node.
We can combine the type 1 path $I$ (the \cpath\ from $\tet'$ to $\sigma'$, with $\weight(I)=-\ubSym(p')$), the type 2 path $Q$ (from $\sigma'$ to $\sigma''$) and the type 3 path $S$ (the \cpath\ from $\sigma''$ to $\tet'_1$), and get a new path, $P$, that is a \cpath\ from $\tet'$ to $\tet'_1$.
Note that $\weight(P)=(-\ubSym(p'))+(d(\sigma''))+(d(\psi_j)-d(\sigma''))=d(\psi_j)-\ubSym(p')$, and that
$\timeof_{r'}(\tet'_1)-\timeof_{r'}(\sigma')=\Tfs(\psi_j)-\Tfs(\sigma')=d(\psi_j)$.
We already know that $\timeof_{r'}(\sigma')-\timeof_{r'}(\tet')=-\ubSym(p')$, and so $\timeof_{r'}(\tet'_1)-\timeof_{r'}(\tet')=d(\psi_j)-\ubSym(p')=\weight(P)$ as required.

\end{proof}

Combining the results of \Cref{lem:fast,lem:paths}, we can get the following result:
\begin{corollary}
\Clabel{cor:fastisvzz}
Let $r\in\Rrep$, $\gamma\in\Nat$, $\sigma$ be a node of $r$ and let $\tet'$ be a \sigAware\ node, where $\timeof_r(\tet')>0$.
Define $r'=\fast{\tet'}$.
Let $\te{2}=\nodeF{\sigma_2,p_2}$ be a \sigAware\ node, where $\sigma_2\in\Vyes{\sigma'}$.
Then there is a $\sigma$-visible zigzag in $r$ from $\tet'$ to \te{2}, with weight that equals to $\timeof_{r'}(\te{2})-\timeof_{r'}(\tet')$.
\end{corollary}

{\bf Finally, we are ready to prove \Cref{thm:vzz}:}
\begin{proof}[Proof of \Cref{thm:vzz}]
Let~$\Rrep=\Rrep(\Prot,\gamma)$ and suppose that~$\Prot$ is an \ffip.
Moreover, let $\sigma$ be a basic node of ${r\in\Rrep}$, let $x\in\Int$, and let $\te{1}=\nodeF{\sigma_1,p_1}$ and $\te{2}=\nodeF{\sigma_2,p_2}$ be $\sigma$-aware nodes in~$r$, such that both ${\timeof_r(\te{1})>0}$ and ${\timeof_r(\te{2})>0}$.
Recall that $\tet'=\nodeF{\sigma',p'}$ can be a $\sigma$-aware node only if $\sigma'\lam_r\sigma$. Thus, since \te{1} and \te{2} are $\sigma$-aware nodes, we have that both $\sigma_1\lam_r\sigma$ and $\sigma_2\lam_r\sigma$.
Assume that \mbox{$(\Rrep,r)\models K_{\sigma}(\te{1}\atleast{x}\te{2})$}.
Concerning $\sigma_2$, the basic node of $\te{2}$ in $r$, either $\sigma_2\in\Vyes{\sigma_1}$ or $\sigma_2\in\Vno{\sigma_1}$.
First, assume that ${\sigma_2\in\Vyes{\sigma_1}}$.
Define $r'=\fastG{0}{\te{1}}$.
According to the definition of knowledge, as $r'\sim_\sigma r$ (by \Cref{lem:fast}) and \mbox{$(\Rrep,r)\models K_{\sigma}(\te{1}\atleast{x}\te{2})$}, we have that $(\Rrep,r')\models\te{1}\atleast{x}\te{2}$.
Thus, $\timeof_{r'}(\te{1})+x\leq\timeof_{r'}(\te{2})$.
By \Cref{cor:fastisvzz} we can conclude that there is a $\sigma$-visible zigzag from \te{1} to \te{2} in~$r$ with weight of $\timeof_{r'}(\te{2})-\timeof_{r'}(\te{1})\geq x$ as required. 

If $\sigma_2\in\Vno{\sigma_1}$, then define $\gamma=max(0,\ubSym(p_2)-\lbSym(p_1)-x)$ and $r''=\fast{\te{1}}$.
Once again, $(\Rrep,r'')\models\te{1}\atleast{x}\te{2}$ and thus $\timeof_{r''}(\te{1})+x\leq\timeof_{r''}(\te{2})$.
However, we will prove that in contrast $\timeof_{r''}(\te{2})<\timeof_{r''}(\te{1})+x$, and so it cannot be that $\sigma_2\in\Vno{\sigma_1}$.
By \Cref{lem:fast}, as $\sigma_2\in\Vno{\sigma_1}$ and $\sigma_1\in\Vyes{\sigma_1}$, we get that $\timeof_{r''}(\sigma{2})<\timeof_{r''}(\sigma_1)-\gamma$. 
Note that $\timeof_{r''}(\te{2})\leq\timeof_{r''}(\sigma_2)+\ubSym(p_2)$ and $\timeof_{r''}(\sigma_1)\leq\timeof_{r''}(\te{1})-\lbSym(p_1)$.
Combining the last three inequalities, we get that: \[\timeof_{r''}(\te{2})\leq\timeof_{r''}(\sigma_2)+\ubSym(p_2)<\timeof_{r''}(\sigma_1)-\gamma+\ubSym(p_2)\leq\timeof_{r''}(\te{1})-\lbSym(p_1)+\ubSym(p_2)-\gamma\]
Note that $\gamma\geq\ubSym(p_2)-\lbSym(p_1)-x$, hence $-\lbSym(p_1)+\ubSym(p_2)-\gamma\leq x$, and so $\timeof_{r''}(\te{2})<\timeof_{r''}(\te{1})+x$, as claimed.
\end{proof}


\end{document}